\documentclass[11pt,twoside]{article}

\usepackage{fancyhdr}
\usepackage{epsfig}
\usepackage{latexsym}
\usepackage{microtype}
\usepackage{lastpage}
\usepackage{graphicx}
\usepackage{lipsum}
\usepackage[latin9]{inputenc}
\usepackage{amsmath}
\usepackage{amssymb}
\usepackage{amsfonts}
\usepackage{enumitem}
\usepackage[ruled, lined, longend]{algorithm2e}
\SetKwInOut{KwIn}{Input}
\SetKwInOut{KwOut}{Output}

\usepackage[hidelinks]{hyperref}
\usepackage[dvipsnames*,svgnames]{xcolor}
\hypersetup{colorlinks=true,allcolors=DarkBlue,linkcolor=black}

\newtheorem{theorem}{Theorem}
\newtheorem{definition}[theorem]{Definition}

\newtheorem{corollary}[theorem]{Corollary}
\newtheorem{lemma}[theorem]{Lemma}
\newtheorem{remark}[theorem]{Remark}

\newtheorem{case}{Case}

\newtheorem{example}[theorem]{Example}

\newenvironment{proof}{\par\noindent {\bf Proof:}\ \ \ }{\hfill$\Box$\hspace{1ex}}

\newcommand{\TheAuthor}{}
\newcommand{\Author}[1]{\renewcommand{\TheAuthor}{#1}}

\newcommand{\TheTitle}{}
\newcommand{\Title}[1]{\renewcommand{\TheTitle}{#1}}

\Author{Eleni Mandrali}
\Title{Describing Weighted Safety with Weighted LTL\\ over Product $\omega$-valuation Monoids}
\newcommand\blfootnote[1]{%
	\begingroup
	\renewcommand\thefootnote{}\footnote{#1}%
	\addtocounter{footnote}{-1}%
	\endgroup
}

\begin{document}

	\blfootnote{This postdoctoral research is cofinanced by
		Greece and European Union (European Social Fund) through the Operational
		program "Development of human resources, education and lifelong learning " in
		the context of the Project "REINFORCEMENT OF POSTDOCTORAL RESEARCHERS- Second
		Cycle " (MIS\ 5033021) that is implemented by State Scholarship Foundation
		(IKY). }   	
	
	\parindent=8mm

	\vspace{1cm}
	\begin{center}
		{\Large\bf {Describing weighted safety with weighted LTL}
			
			\vspace{2mm}
			{\Large\bf  over product $\omega$-valuation monoids }
		}
	\end{center}
	\vspace{4mm}

	\begin{center}
		{\large Eleni \textsc{Mandrali}}\footnote{This work is a part of the author's post-doctoral research at the Department of Mathematics of the Aristotle University of Thessaloniki, 54124 Thessaloniki,
			\c Greece, Email: {\tt elemandr@math.auth.gr, elemandr@gmail.com}}
	\end{center}
	\vspace{3ex}
	
	\date{}
	
	\begin{abstract}
		
		We define the notion of $k$-safe infinitary series over idempotent ordered
		totally generalized product $\omega $-valuation monoids that satisfy
		specific properties. For each element $k$ of the underlying structure
		(different from the neutral elements of the additive, and the multiplicative
		operation) we determine two syntactic fragments of the weighted $LTL$ with
		the property that the semantics of the formulas in these fragments are $k$%
		-safe infinitary series. For specific idempotent ordered totally generalized
		product $\omega $-valuation monoids we provide algorithms that given a
		weighted B\"{u}chi automaton and a weighted $LTL$ formula in these
		fragments, decide whether the behavior of the automaton coincides with the
		semantics of the formula.

		\smallskip
		
		
	\end{abstract}

	\section{\protect\bigskip Introduction}
	
	Verifying the reliability of information and communication systems with the
	use of mathematical logic has been an important research field over the last
	years. A successful paradigm for achieving formal verification with the use
	of mathematical logic is model checking. Linear Temporal Logic (\textit{LTL}
	for short), introduced by Pnueli in \cite{Pn-Th}, has been successfully used
	in checking the correctness of programs with important applications in
	industry \cite{Ba-Ka}. One of the most efficient ways to develop
	model-checking tools incorporating \textit{LTL} is the automata-theoretic
	approach proposed by Wolper and Vardi in \cite{Va-Wo}. In addition, \textit{%
		LTL} is a logical formalism capable of describing properties of the linear
	behavior of systems like safety and liveness. The former expresses the
	fact that "nothing bad" can happen and is a very important part of verifying
	the reliability of programs \cite{Ba-Ka}. Safety $LTL$ formulas (cf. \cite%
	{Si-Sa}) over infinite words define safety infinitary languages. An
	infinitary language $L$ is a safety language if for every infinite word $w$
	the following holds: if every non-empty prefix of $w$ can be extended to a
	word in $L,$ then $w$ belongs to $L$ as well \cite{Al - Re}.
	
	Designing effective algorithms
	capable of reasoning about quantitative aspects of correctness of systems'
	behavior, as well as identifying and studying quantitative properties that generalize the meaning of their boolean counterparts has been, and is an active research field \cite{No-Pa,Ca-Ap,Dr-Mo,He-Qu,Sa-Bo,Bo-Av,Al-Bo,Ha-Ku}. The notion of weighted safety was for the first time
	introduced in \cite{Ha-Ku}. More precisely, given a rational value $v$, the authors defined the notion of $v$-safety for series, over finite words over an alphabet $A$,
	with rational values. A series $L$ is $v$-safe, if every word $w$, to which
	$L$ assigns a value greater or equal to $v$, has a prefix that cannot be extended to a word over $A$ that is assigned through $L$ to a value less than $%
	v$. A series $L$ is weighted safe if there exists a rational value $v$ such that
	$L$ is $v$-safe. In that paper the authors studied the relationship of
	weighted safe series and weighted automata whose behavior is defined as
	follows: the weight of each path is obtained by adding the rational weights
	of the transitions forming the path, and the value that the behavior of the
	automaton assigns to each word is defined to be the minimum weight of all
	accepting paths of the weighted automaton over this specific word. By determining the
	structural properties of deterministic weighted automata, the authors answered
	decidability questions related to the weighted safety of their behavior. Multi-valued safety and liveness linear-time properties have been introduced in $\cite{Dr-Mo}$ where automata based model-checking techniques have been developed. Quantitative safety and liveness has also been considered recently, in $\cite{He-Qu}$, for infinitary series over complete lattices. In that work, the authors proved important results related to the safety-liveness decomposition of the infinitary series under consideration. They also provided results related to the monitorability of infinitary series that satisfy specific approximation conditions related to safety, and co-safety. It is proved in that work that the two definitions of safe infinitary series, the one presented in $\cite{He-Qu}$, and the one presented in $\cite{Dr-Mo}$, are equivalent. However the definition of liveness in the two works significantly differs, and in consequence different safety-liveness decomposition results are obtained. In $\cite{Sa-Bo}$ the authors further considered a definition of threshold safety for infinitary series over complete lattices. According to that definition, an infinitary series $s$ is threshold safe if for every value $v$ of the domain the infinitary language consisting of all infinite words whose image through $s$ is greater or equal to $v$ is a safety property. They prove that for totally ordered domains the definition of threshold safety and the definition of safety presented in $\cite{He-Qu}$ are equivalent. For such domains they provide a topological characterization of their safe infinitary series. They also provide decidability results for the safety of well known classes of quantitative automata.
	
	In this paper, we firstly aim to define, given a specific threshold $k$ of a weight domain, a notion of weighted safety with respect to the given threshold for
	infinitary series that is capable of applying to a variety of fields. For the
	domain of weights we choose idempotent ordered totally generalized product $%
	\omega$-valuation monoids that satisfy specific properties. Infinitary
	series express the quantitative behavior of infinite executions of systems,
	and the $k$-safe infinitary series that we introduce in this work can
	be understood as series that describe systems whose behavior is
	characterized by the existence of a "quantitative safe" area. This area can
	be described as an infinitary safety language that is formed by all infinite
	words whose coefficient (on the series describing the system's behavior)
	satisfy a specific quantitative property, more specifically exceeds the given
	threshold (according to the order relation induced by the idempotency of the
	monoid). We aim to describe this notion of safety with weighted $LTL$
	formulas, and for specific idempotent ordered totally generalized product $%
	\omega $-valuation monoids provide algorithms that given a weighted B\"{u}%
	chi automaton and a weighted $\mathit{LTL}$ formula whose semantics is $k$-safe, decide whether the behavior of the automaton coincides with the
	semantics of the formula.
	
	To motivate the choice of the underlying weight structure we recall that the theory
	of weighted automata and weighted logics over semirings has important practical applications \cite{Dr-Ku}, however the
	spectrum of applications of weighted models can be broadened by considering
	weighted models and weighted specification languages over abstract algebraic
	structures that are not necessarily semirings. To that direction in \cite%
	{Ch-Do}, Chatterjee, Doyen and Henzinger considered weighted automata over
	the real numbers where the weight of a run over a finite (infinite) word is
	computed with the use of a function that assigns a real value to every
	finite (infinite) run. For finite words, examples of such functions are the
	functions that compute the average value of real numbers and the function of
	sum of real numbers, while for infinite words they used functions as Sup,
	LimSup, LimInf, Limit Average and Discounted Sum. With the use of the above
	functions we can model a wide range of procedures of systems. As is stated
	in \cite{Ch-Do} the peak of energy consumption can be attributed as the
	maximum of a sequence of weights that represent energy consumption, energy
	use can be attributed through the sum of real numbers, while for average
	response time we use the function of Limit Average. In~\cite{Ch-Do} the
	authors presented decidability results and results concerning the expressive
	power of the model for different functions. Similar results were presented
	in \cite{Ch-Al,Ch-Ex,Ch-Pr} where other kinds of automata that use functions
	for the computation of the weight of a run were presented. In \cite{Dr-De,Dr-Me}, Droste and Meinecke, proposed the structures of valuation
	monoids, and $\omega$-valuation monoids as formalisms capable of describing how these automata work for
	different functions. Valuation monoids (resp. $\omega$-valuation monoids) can be described as monoids (resp. complete monoids) equipped with a valuation function (resp. $\omega$-valuation function) such that additionally the additive operation of the monoid and the valuation function (resp. the $\omega$-valuation function) satisfy specific properties. In \cite{Dr-Me} the authors introduced a weighted MSO over
	product valuation monoids, and over product $\omega$-valuation monoids (i.e., valuation, and $\omega$-valuation monoids further equipped with a product operation) and under the
	consideration of specific algebraic properties for the weight structure proved expressive equivalences of fragments of that logic
	with weighted automata over finite, and infinite words. In \cite{Ma-At} a weighted $LTL$\textit{\ }over
	product $\omega $-valuation monoids, and generalized product $\omega$%
	-valuation monoids was presented. In that paper, an effective translation of formulas of a syntactic fragment of the weighted $LTL$ to weighted generalized B\"{u}chi automata with $\varepsilon $%
	-transitions was presented following the construction of \cite{Ma-We}. We recall that the classes of generalized product $\omega$-valuation monoids, and of product $\omega$-valuation monoids introduced in \cite{Ma-At} are subclasses of the class of product $\omega$-valuation monoids defined in \cite[Definition 5.1]{Dr-Me}. Moreover, the class of generalized product $\omega$-valuation monoids presented in \cite{Ma-At}, properly contains the class of product $\omega$-valuation monoids considered in the same work \cite{Ma-At}.
	
	We now present how the content of this paper is structured. In Section~\ref{Related_work}, we compare the work presented in this paper with related literature. In Section \ref{Preliminaries}, we present some preliminary
	notions. In Section \ref{Totally gen pr omega val mon}, we define the class of totally generalized product $\omega $%
	-valuation monoids that properly contains the class of generalized product $%
	\omega $-valuation monoids. We conclude that idempotent totally generalized
	product $\omega $-valuation monoids preserve basic properties of idempotent
	generalized product $\omega $-valuation monoids. We adopt the definitions of weighted generalized B\"{u}chi automata with $%
	\varepsilon $-transitions, and weighted B\"{u}chi automata with $\varepsilon
	$-transitions introduced in \cite{Ma-At}, whithin the framework of totally generalized product $\omega $%
	-valuation monoids, and also conclude that the result of their expressive
	equivalence with weighted B\"{u}chi automata holds for these weight structures as well. In Section \ref{Desribing ws with wltl}, we present the weighted $%
	LTL$ over idempotent ordered totally generalized product $\omega $-valuation
	monoids, which is defined in the same way as the weighted \textit{LTL }in
	\cite{Ma-At}. We identify a syntactic fragment of the weighted $LTL$, the
	fragment of $\vee$-totally restricted $U$-nesting weighted $LTL$ formulas,
	such that each formula of this fragment can be effectively translated to a weighted B\"{u}chi automaton recognizing its semantics. For this translation we follow
	the construction and proof of the corresponding result presented in~\cite%
	{Ma-At} for totally restricted $U$-nesting $LTL$ formulas over idempotent
	ordered generalized product $\omega $-valuation monoids. If $K$ is an
	idempotent ordered totally generalized product $\omega $-valuation monoid,
	and $k\in K\backslash \left\{ \mathbf{0},\mathbf{1}\right\} ,$ we define the
	notion of $k$-safe infinitary series, and for all $k\in K\backslash \left\{
	\mathbf{0},\mathbf{1}\right\} ,$ we determine two syntactic fragments of
	the weighted $LTL$, namely the fragment of $k$-safe totally restricted $U$%
	-nesting weighted $LTL$ formulas, and the fragment of $k$-safe $\vee $%
	-totally restricted $U$-nesting weighted $LTL$ formulas, with the property
	that the semantics of the formulas in these fragments are $k$-safe
	infinitary series. Moreover, we conclude that for every formula in these
	fragments we can effectively construct a weighted B\"{u}chi automaton
	recognizing its semantics. In Section \ref%
	{A_motivating_example}, we present an example on how weighted $LTL$
	formulas with $k$-safe semantics can be used to argue about aspects of the quantitative behavior of
	systems. In Section \ref{Alg_for_ws}, for two specific idempotent
	ordered totally generalized product $\omega $-valuation monoids (resp. a
	specific idempotent ordered generalized product $\omega $-valuation monoid),
	we present an algorithm that given a weighted B\"{u}chi automaton and a $k$%
	-safe $\vee $-totally restricted $U$-nesting weighted $LTL$ formula (resp. $k
	$-safe totally restricted $U$-nesting weighted $LTL$ formula), decides if
	the semantics of the formula coincides with the behavior of the automaton.
	The core of the algorithm is based on the fact that the weighted $LTL$
	formula can be effectively translated to a weighted B\"{u}chi automaton, and
	the quantitative language equivalence problem is decidable for weighted B%
	\"{u}chi automata over these specific structures. For the proof of this last
	decidability result we follow the approach of the constructive proof of the
	corresponding result in \cite{Ch-Do}. We obtain as a corollary that the results of this section can be generalized for specific families of totally generalized, and generalized product $\omega$-valuation monoids. In Section \ref{Conclusion}, we
	present conclusions and future work directions.
	
	\section{Related Work\label{Related_work}}
	
	As it is mentioned in the introduction, the notion of weighted safety was
	introduced for finitary series over the set of rational numbers in \cite%
	{Ha-Ku}. As the authors proved, it follows by that definition that for every
	rational number $v,$ a series $s$ is $v$-safe iff the language consisting of
	all words whose coefficient on $s$ is less than $v$ is a safety language. In
	the current work, given an idempotent ordered totally generalized product $%
	\omega $-valuation monoid $K,$ and a $k\in K\backslash \left\{ \mathbf{0,1}%
	\right\} ,$ we obtain that an infinitary series $r$ is $k$-safe iff the
	language consisting of all infinite words whose coefficient on $r$ is greater or
	equal to $k$ (according to the order relation induced by the idempotency of
	the monoid) is an infinitary safety language. We see this as a
	generalization of the notion of safety introduced in \cite{Ha-Ku}, since a
	corresponding result for infinitary series, as the one that is proved in
	\cite{Ha-Ku}, could be obtained by our definition if the domain of $K$ is $%
	\mathbb{Q}
	\cup \left\{ \infty ,-\infty \right\} ,$ and the natural order $\geq _{K}$
	induced by the idempotency of $K$ coincides with the following order
	relation: for every $k_{1},k_{2}\in K,$ $k_{1}\geq _{K}k_{2}$ iff $\min
	\left( k_{1},k_{2}\right) =k_{1}.$
	Moreover, if we adopt the definition of threshold safety introduced in \cite{Sa-Bo} in the framework of our weight structure, then by Example \ref{example_infinite_image2} we conclude that if an infinitary series is $k$-safe for some specific $k$, then it is not necessarily threshold safe. However, the converse statement is true, i.e., threshold safety implies
	$k$-safety for every element $k$ of the value domain for which $k$-safety is defined. It is the author's understanding that in the framework of weighted $LTL$ the definition of $k$-safety allows the description of more refined quantitative~properties.
	
	A weighted \textit{LTL} over infinite words over the max-plus semiring with
	discounting parameters was defined in \cite{Ma-Co,Ma-We}. In that work, the
	author presented an effective translation of a syntactic fragment of the
	formulas of that logic to equivalent weighted B\"{u}chi automata. For this translation the author followed the
	approach of \cite{Va-Wo}, i.e., the weighted $LTL$ formula was firstly effectively translated to a generalized weighted B\"{u}chi automaton with $\varepsilon$-transitions that simulates the inductive
	computation of the semantics of the formula by relating the states of the
	automaton with the syntax of the given formula. Then, the generalized weighted B\"{u}chi automaton with $\varepsilon$-transitions is effectively translated to an equivalent weighted B\"{u}chi automaton. We note that generalized weighted B\"{u}chi automata with $\varepsilon$-transitions, and weighted B\"{u}chi automata with $\varepsilon$-transitions over the max-plus semiring with discounting were introduced in \cite{Ma-Co,Ma-We}, and their expressive equivalence with weighted B\"{u}chi automata was proved. This result, and the result of the effective translation of weighted $LTL$ formulas to equivalent weighted B\"{u}chi automata have been
	generalized for a subclass of idempotent (ordered generalized) product $%
	\omega $-valuation monoids in \cite{Ma-At}, and in the current work we
	further generalize it for a subclass of idempotent ordered totally
	generalized product $\omega $-valuation monoids following the arguments of
	the proof of \cite{Ma-At}. To our knowledge a quantitative \textit{%
		LTL} was for the first time introduced in \cite{Ku-Lu} where the authors
	defined a weighted\textit{\ LTL} and weighted automata over de Morgan
	Algebras and an effective translation of the formulas of that logic to
	weighted automata was presented. We note that the core of the definition of
	the weighted \textit{LTL }introduced in \cite{Ma-Co,Ma-We,Ma-At} goes back to the definition of weighted MSO logics over semirings by
	Droste, and Gastin in \cite{Dr-We}.
	
	In \cite{Ch-Do} the authors deal with weighted automata over infinite words
	that use as valuation functions the classical $\lim \sup ,$ $\lim \inf ,$
	and $\sup $ of infinitary sequences of rational numbers, and prove that the
	language equivalence problem is decidable for these types of automata. In
	Section \ref{Alg_for_ws}, we prove a corresponding result for weighted B%
	\"{u}chi automata over $K_{2}$,$\ K_{1},$ $K_{3}$, where $K_{1},$ $K_{2},$ $%
	K_{3}$ are specific totally generalized product $\omega $-valuation monoids
	over the domain $%
	\mathbb{Q}
	\cup \left\{ \infty ,-\infty \right\} $ that incorporate the functions $\lim
	\inf ,$ $\lim \sup ,$ $\sup $ respectively in the definition of their $%
	\omega $-valuation functions. As it is stated in the introduction, for the
	proof of this result, we follow the constructive approach of the proof of the
	corresponding result in \cite{Ch-Do}. Core steps of the algorithms presented in Section \ref{Alg_for_ws} are based on basic constructive steps of that proof.
	
	Finally, in \cite{Sm-Op} the authors presented an algorithm that given a
	weighted transition system that models the motion of a robot, and an $LTL$
	formula that specifies the mission of the robot, determines a path of the
	robot that satisfies the specification, and at the same time minimizes the
	value of a given cost function. Weights on the transitions of the weighted
	transition system represent time, and the cost function is defined so that
	it expresses, for every possible route of the robot, the maximum time needed
	between instances of its movement that satisfy a specific optimizing atomic
	proposition. The $LTL$ formula is defined in a way, so that every path
	that satisfies it, satisfies the optimizing atomic proposition infinitely
	often. The example that we present in Section \ref{A_motivating_example}
	is motivated by the need to argue about other quantitative aspects of the
	behavior of systems as the one presented in \cite{Sm-Op}, and to answer questions related to correctness, and trustworthiness \cite{Si-Me}.
	
	\section{Preliminaries\label{Preliminaries}}
	
	Let $A$ be an alphabet, i.e., a finite non-empty set. As usually, we denote
	by~$A^{\ast }$ the set of all finite words over $A$ and we let $A^{+}=A^{\ast
	}\left\backslash \left\{ \varepsilon \right\} \right. ,$ where $\varepsilon $
	is the empty word. The set of all infinite sequences with elements in $A$,
	i.e., the set of all infinite words over $A,$ is denoted by $A^{\omega }.$
	Let $w\in A^{\omega }.$ A word $v\in A^{\omega }$ is called a suffix of $w$,
	if $w=uv$ for some $u\in A^{\ast }.$ Similarly, a word $u\in A^{\ast }$ is a
	called a prefix of $w$ if $w=uv$ for some $v\in A^{\omega }.$ \ Every
	infinite word $w=a_{0}a_{1}\ldots $ with $a_{i}\in A\left( i\geq 0\right) $
	is written also as $w=w\left( 0\right) w\left( 1\right) \ldots $ where $%
	w\left( i\right) =a_{i}\left( i\geq 0\right) $. The word $w_{\geq i}$
	denotes the suffix of $w$ that starts at position $i,$ i.e., $w_{\geq
		i}=w\left( i\right) w\left( i+1\right) \ldots ,$ while the word $w_{<i}$
	denotes the prefix of $w$, $w_{<i}=w\left( 0\right) \ldots w\left(
	i-1\right) .$ Clearly, $w_{<0}=\varepsilon .$
	
	Let $C,K$ be sets. We denote by $B\subseteq
	_{fin}C$ the fact that $B$ is a finite subset of $C$ and we let $%
	\left(C_{fin}\right)^{\omega }=\underset{B\subseteq _{fin}C}{\bigcup }%
	B^{\omega }.$ We denote by $\mathcal{P}\left( C\right) $ the powerset of $C$%
	. An index set $I$ of $C$ is a subset of $C$. A family of elements of $K$
	over the index set $I$, denoted by $\left( k_{i}\right) _{i\in I}$, is a
	mapping $f$ \ from $I$ to $K $ where $k_{i}=f\left( i\right) $ for all $i\in
	I$. We shall denote by $\mathbb{N}$ the set of non-negative integers.
	
	We recall now the definition of B\"{u}chi, Muller, and Rabin automata over $%
	A.$ A B\"{u}chi automaton (Ba for short) over the alphabet $A$ is a tuple $%
	\mathcal{A=}\left( Q,A,I,\Delta ,F\right) $ where $Q$ is a finite state set,
	$\Delta \subseteq Q\times A\times Q$ is the set of transitions, $I\subseteq Q$ is the
	set of initial states, and $F\subseteq Q$ is the set of final states. Let $%
	w=a_{0}a_{1}\ldots \in A^{\omega }$ with $a_{i}\in A\left( i\geq 0\right) .$
	A path $P_{w}$ of $\mathcal{A}$ over $w$ is an infinite sequence of
	transitions $P_{w}=\left( q_{j},a_{j},q_{j+1}\right) _{j\geq 0}$. We let $%
	In^{Q}\left( P_{w}\right) $ be the set of states that appear infinitely
	often along~$P_{w}.$ We will call $P_{w}$ successful if $q_{0}\in I,$ and at
	least one final state appears infinitely often along $P_{w},$ i.e., if $%
	q_{0}\in I,$ and $In^{Q}\left( P_{w}\right) \cap F\neq \emptyset .$ The
	behavior of $\mathcal{A}$ is defined to be the infinitary language $%
	\left\Vert \mathcal{A}\right\Vert =\left\{ w\in A^{\omega }\mid \text{there
		exists a successful path of }\mathcal{A}\text{ over }w\right\} .$ A
	(nondeterministic) Muller automaton (Ma for short) over the alphabet $A$ is
	a tuple $\mathcal{A}=\left( Q,A,p_{0},\Delta ,\mathcal{F}\right) $, where $Q$
	is a finite set of states, $p_{0}\in Q$ is the initial state, $\Delta
	\subseteq Q\times A\times Q$ is the set of transitions, and $\mathcal{F}\subseteq \mathcal{P}\left(Q\right) $ is
	the set of final subsets$.$ A path $P_{w}$ of $\mathcal{A}$ over $w$, and
	the set $In^{Q}\left( P_{w}\right) $ are defined as in the case of B\"{u}chi
	automata. We say that the path $P_{w}$ is successful, if $q_{0}=p_{0}$ and $In^{Q}\left(
	P_{w}\right) \in \mathcal{F}.$ The behavior of $\mathcal{A}$ is again
	defined to be the infinitary language $\left\Vert \mathcal{A}\right\Vert
	=\left\{ w\in A^{\omega }\mid \text{there exists a successful path of }%
	\mathcal{A}\text{ over }w\right\} .$ \ Finally, a (nondeterministic) Rabin
	automaton (Ra for short) is a tuple $\mathcal{A}=\left( Q,A,p_{0},\Delta ,%
	\mathcal{R}\right) $, where $Q$ is a finite set of states, $p_{0}\in Q$ is
	the initial state, $\Delta \subseteq Q\times A\times Q$ is the set of
	transitions, and $\mathcal{R=}\left( L_{j},U_{j}\right) _{j\in J}$ is a
	finite family of pairs of sets of states. Again, a path $P_{w}$ of $\mathcal{A}$
	over $w$, and the set $In^{Q}\left( P_{w}\right) $ are defined as in the
	case of B\"{u}chi automata. We will say that the path $P_{w}$ is successful
	if $q_{0}=p_{0}$, and there exists an index $j\in J$ such that $In^{Q}\left( P_{w}\right) \cap
	L_{j}=\emptyset ,$ and $In^{Q}\left( P_{w}\right) \cap U_{j}\neq \emptyset .$
	The behavior of $\mathcal{A}$ is defined as in the previous two cases of B%
	\"{u}chi, and Muller automata. We recall the following theorem that states
	the expressive equivalence of Muller, B\"{u}chi and Rabin automata.
	
	\begin{theorem} \label{B-M-R}
		\begin{enumerate}[label=(\roman*)]	
			\item For every Muller automaton over $A$ we can construct a B%
			\"{u}chi automaton over $A$ with the same behavior, and vice-versa.
			
			\item For every Rabin automaton over $A$ we can construct a B\"{u}chi
			automaton over $A$ with the same behavior, and vice-versa.
		\end{enumerate}
	\end{theorem}
	
	\noindent For a detailed reference on the theory of B\"{u}chi, Muller, and Rabin
	automata we refer to \cite{Pe-In}.
	
	\bigskip{} Next, we recall the definition of classical $LTL$, and the
	definition of safety properties (also equivalently called safety languages) and safety formulas (see \cite{Si-Sa}).
	
	\begin{definition}[\cite{Al - Re}]
		\label{Def safe}Let $A$ be an alphabet and $L\subseteq
		A^{\omega }.$ $L$ is a safety property iff for every $w\in A^{\omega }$ the
		following hold: If $\forall i>0,\exists u\in A^{\omega }$ such that $%
		w_{<i} u\in L,$ then $w\in L.$
	\end{definition}

	\begin{example}
		Let $A=\left\{ a,b\right\} $ be an alphabet, $L=a^{+}b^{+}A^{\omega }\subseteq A^{\omega},$
		and $F=a^{+}A^{\omega }\subseteq A^{\omega}$. We consider
		the infinite word $w=a^{\omega }.$ Then, for every $i>0,$ $w_{<
			i} ba^{\omega }\in L,$ however $w\notin L.$ This implies that $L$ is not a
		safety property. We can easily conclude that $F$ is a safety
		property. More precisely, let $\overline{w}\in A^{\omega }$ such
		that for every $i>0,\exists u\in A^{\omega }$ such that $\overline{w}_{<i} u\in F.$
		Then, $w\left( 0\right) =a,$ which implies that $\exists$$j>0,\overline{u}%
		\in A^{\omega }$ such that $\overline{w}=a^{j}\overline{u},$ i.e., $%
		\overline{w}\in F.$
	\end{example}
	
	Let $AP$ be a finite set of atomic propositions. We shall denote the
	elements of $AP$ by $a,b,c,\ldots.$ The syntax of classical $LTL$ over $AP$
	is given by the grammar
	\begin{equation*}
		\varphi ::= true\mid a\mid \lnot \varphi \mid \varphi \wedge \varphi \mid \varphi \vee
		\varphi \mid \bigcirc \varphi \mid \varphi U\varphi \mid \square \varphi
	\end{equation*}%
	where $a\in AP.$ We shall denote by $LTL\left( AP\right) $ the class of all $%
	LTL$ formulas over $AP.$ As usual we identify $\lnot \lnot \varphi $ with $%
	\varphi $ for every $\varphi \in LTL(AP).$ Let now $\varphi \in LTL\left(
	AP\right) ,$ and $w\in \left( \mathcal{P}\left( AP\right) \right) ^{\omega
	}.$ Then, for every $i\geq 0,$ it is defined what it means that $\varphi $
	holds in $w$ at position $i,$ which is denoted by $w,i\models \varphi ,$ inductively
	in the following way:
	
	\begin{itemize}
		\item[-] $w,i\models true$ \ \ \ \ \ \ \ \ \ \ \ \ \ \ \ \ \ \ \ \ \ \ \
		\ \ - $w,i\models \varphi \vee \psi ,$ if $w,i\models \varphi $ or $%
		w,i\models \psi $
		
		\item[-] $w,i\models a,$ if $a\in w\left( i\right) $ \ \ \ \ \ \ \ \ \ \ \ \
		\ \ \ - $w,i\models \varphi \wedge \psi ,$ if $w,i\models \varphi $ and $%
		w,i\models \psi $
		
		\item[-] $w,i\models \lnot \varphi ,$ if $w,i\nvDash \varphi $ \ \ \ \ \ \ \
		\ \ \ \ \ \ \ - $w,i\models \square \varphi ,$ if for every $j\geq i,$ $%
		w,j\models \varphi $
		
		\item[-] $w,i\models \bigcirc \varphi ,$ if $w,i+1\models \varphi $ \ \ \ \
		\ \
		
		\item[-] $w,i\models \varphi U\psi ,$ if there exists $j\geq i$ such that $%
		w,i^{\prime }\models \varphi $ for all $i\leq i^{\prime }<j$, and $%
		w,j\models \psi .$
	\end{itemize}
	
	For an infinite word $w\in \left( \mathcal{P}\left( AP\right) \right) ^{\omega },$ we
	will say that $\varphi $ holds in $w,$ and we will denote by $w\models
	\varphi ,$ iff $w,0\models \varphi.$ We say that two $LTL$ formulas $%
	\varphi ,\psi $ are equivalent and denote by $\varphi \equiv \psi $ iff for
	every $w\in \left( \mathcal{P}\left( AP\right) \right) ^{\omega },$ $%
	w\models\varphi $ iff $w\models \psi $. We also recall that the \textit{weak
		until} operator\textit{\ }$\widetilde{U}$ is defined in the following way $%
	\varphi \widetilde{U}\psi :=\square \varphi \vee \left( \varphi U\psi
	\right) .$ For every $\varphi \in LTL\left( AP\right) $ we let the\textit{\
		language of} $\varphi $, denoted by $\mathcal{L}\left( \varphi \right) ,$ be
	defined as follows: $\mathcal{L}\left( \varphi \right) =\left\{ w\in \left(
	\mathcal{P}\left( AP\right) \right) ^{\omega }\mid w\models \varphi \right\}
	.$
	
	A formula $\varphi$ is a safety formula iff $\mathcal{L}\left( \varphi
	\right) $ is a safety property over $\mathcal{P}\left( AP\right) $. We let $%
	sLTL\left( AP\right) $ be the fragment of $LTL\left( AP\right) $ given by
	the following grammar%
	\begin{equation*}
		\varphi ::= true\mid a\mid \lnot a\mid \varphi \wedge \varphi \mid \varphi \vee \varphi
		\mid \bigcirc \varphi \mid \varphi \widetilde{U}\varphi \mid \square \varphi
	\end{equation*}%
	where $a\in AP.$
	
	\begin{theorem}[\cite{Si-Sa}] For every $a \in AP, a, \lnot a$ are safety formulas and if $%
		\varphi ,\psi $ are safety formulas, then so are $\varphi \wedge \psi
		,\varphi \vee \psi ,\bigcirc \varphi ,\varphi \widetilde{U}\psi ,$ and $%
		\square \varphi .$
	\end{theorem}
	
	It follows by the previous lemma that every formula in $sLTL\left(AP\right)$ is a safety formula.
	
	We turn now to monoids and infinitary series. A monoid $\left( K,+,\mathbf{0}%
	\right) $ is an algebraic structure equipped with a non-empty set $K$ and an
	associative additive operation $+$ with a zero element\textbf{\ }$\mathbf{0}%
	, $ i.e., $\mathbf{0}+k=k+\mathbf{0}=k$ for every $k\in K.$ The monoid $K$
	is called commutative if $+$ is commutative.
	
	A monoid $\left( K,+,\mathbf{0}\right) $ is called complete if it is
	equipped, for every index set $I$, with an infinitary sum operation $%
	\sum_{I}:K^{I}\rightarrow K$ such that for every family $\left( k_{i}\right)
	_{i\in I}$ of elements of $K$ we have
	\begin{equation*}
		\underset{i\in\emptyset}{\sum}k_{i}=\mathbf{0},\underset{i\in\left\{
			j\right\} }{\sum}k_{i}=k_{j},\underset{i\in\left\{ j,l\right\} }{\sum }%
		k_{i}=k_{j}+k_{l}\text{ for }j\neq l
	\end{equation*}
	and
	\begin{equation*}
		\underset{j\in J}{\sum }\left( \underset{i\in I_{j}}{\sum}k_{i}\right) =%
		\underset{i\in I}{\sum }k_{i},
	\end{equation*}
	if $\bigcup _{j\in J}I_{j}=I$ and $I_{j}\bigcap I_{j^{\prime}}=\emptyset$
	for $j\neq j^{\prime}.$ We note that every complete monoid is commutative.
	
	Let $K$ be a complete monoid. $K$ is called additively idempotent (or simply
	idempotent), if $k+k=k$ for every $k\in K.$ Furthermore, $K$ is zero-sum
	free if $k+k^{\prime }=\mathbf{0}$ implies $k=k^{\prime }=\mathbf{0}$. It is
	well known that if $K$ is idempotent, then $K$ is necessarily zero-sum free
	\cite{Ak-Ga}. We recall (cf. \cite{Gu-An}) that idempotency gives rise
	to a natural partial order in $K$ defined in the following way. Let $%
	k,k^{\prime }\in K,$ then $k\leq k^{\prime }$ iff $k^{\prime }=k^{\prime
	}+k. $ Equivalently, it holds $k\leq k^{\prime }$ iff $k^{\prime }=k^{\prime
		\prime }+k$ for some $k^{\prime \prime }\in K$ (cf.  {\cite[Chapter 5]{Dr-Ku}}). We shall denote by $k<k^{\prime }$ the fact that $k\leq k^{\prime
	}$, and $k\neq k^{\prime }$. We shall call the natural order of $K$ a total
	order, if $k\leq k^{\prime }$, or $k^{\prime }\leq k$ for all $k,k^{\prime
	}\in K.$ In this case we will call $K$ an idempotent ordered complete
	monoid. Clearly, if~$K$ is an idempotent ordered complete monoid it holds:
	For every $k,k^{\prime }\in K$, if $k\leq k^{\prime }$ is not true, then $%
	k^{\prime }<k$.
	Let $A$ be an alphabet and\textit{\ }$K$ be an idempotent ordered complete
	monoid. An infinitary series over $A$ and $K$ is a mapping $s:A^{\omega
	}\rightarrow K.$ For every $w\in A^{\omega }$ we write $\left( s,w\right) $
	for the value $s\left( w\right) $ and refer to it as the coefficient of $s$
	on $w.$ We denote by $K\left\langle \left\langle A^{\omega }\right\rangle
	\right\rangle $ the class of all infinitary series over $A$ and $%
	K. $ For a series $s\in K\left\langle \left\langle A^{\omega }\right\rangle
	\right\rangle $ we let $Im\left( s\right) =\left\{ \left( s,w\right) \mid
	w\in A^{\omega }\right\} ,$ $supp\left( s\right) =\left\{ w\in A^{\omega
	}\mid \left( s,w\right) \neq \mathbf{0}\right\} $, and $L_{\geq k}\left(
	s\right) =\left\{ w\in A^{\omega }\mid \left( s,w\right) \geq k\right\} $
	for every $k\in K.$ Moreover, for every $k\in K,$ we shall denote by $k,$
	the infinitary series over $A$ and $K$ defined by $\left( k,w\right) =k$ for
	every $w\in A^{\omega }.$ Given two series $s_{1},s_{2}\in K\left\langle
	\left\langle A^{\omega }\right\rangle \right\rangle $ we write $%
	s_{1}\sqsubseteq s_{2}$ if $\left( s_{1},w\right) \leq \left( s_{2},w\right)
	$ for all $w\in A^{\omega },$ and $s_{1}=s_{2}$ if $\left( s_{1},w\right)
	=\left( s_{2},w\right) $ for all $w\in A^{\omega }.$ The quantitative
	inclusion problem, for the series $s_{1},s_{2},$ is to decide whether $%
	s_{1}\sqsubseteq s_{2}$, while the quantitative language-equivalence problem
	decides whether $s_{1}=s_{2}.$
	
	Next, we recall from \cite{Sm-Op} the definition of weighted transition
	systems. $%
	\mathbb{R}
	_{>0}$ stands for the set of positive real numbers.
	
	\begin{definition}
		A weighted transition system is a tuple $T=\left( Q,Q_{0},R,AP,L,we\right) ,$
		where $Q$ is finite set of states, $Q_{0}\subseteq Q$ is the set of initial
		states, $R\subseteq Q\times Q$ is the transition relation, $AP$ is a finite
		set of atomic propositions, $L:Q\rightarrow \mathcal{P}\left( AP\right) $ is
		a labeling function, and $we:R\rightarrow
		\mathbb{R}
		_{>0}$ is a weight function.
	\end{definition}
	
	A run $r$ of $T$ is an infinite sequence of states $r=q_{0}q_{1}q_{2}\ldots $
	with $q_{0}\in Q_{0},$ and $\left( q_{i},q_{i+1}\right) \in R$ for every $%
	i\geq 0.$ Every run $r=q_{0}q_{1}\ldots $ defines a unique word $%
	w_{r}=L\left( q_{0}\right) L\left( q_{1}\right) L\left( q_{2}\right) \ldots
	\in \left( \mathcal{P}\left( AP\right) \right) ^{\omega },$ and a unique
	sequence of weights $\omega _{r}=\left( we\left( q_{i},q_{i+1}\right)
	\right) _{i\geq 0}.$
	
	In the rest of the paper $A$ will stand for an alphabet, and $AP$ for a
	finite set of atomic propositions.
	
	\section{Totally Generalized Product $\protect\omega $-valuation \\Monoids
		\label{Totally gen pr omega val mon}}
	
	We aim to identify structures equipped with $\omega$-valuation functions that are well studied and can be used for practical applications. To this end, we consider a larger subclass of product $\omega$-valuation monoids, than this of generalized product $\omega$-valuation monoids (for short GP-$\omega$-valuation monoids) introduced in \cite[Definition 2]{Ma-At}. More precisely, we define totally generalized product $\omega$-valuation monoids. These are defined as GP-$\omega$-valuation monoids with the only difference that we restrict the requirement of the validity of the distributivity of the $\omega$-valuation function over finite sums to a smaller range of families of elements of the value domain. This allows the identification of examples of totally generalized product $\omega$-valuation monoids that incorporate classical $\omega$-valuation functions like $sup$ and $\text{limsup}$ of countably infinite families of rational numbers. In \cite{Ma-At}, the author considered GP-$\omega$-valuation monoids that satisfy additional properties, sufficient to establish the effective translation of formulas of fragments of the weighted $LTL$, with weights over GP-$\omega$-valuation monoids, to weighted B\"{u}chi automata. In this work, we will identify sufficient properties that totally generalized product $\omega$-valuation monoids should satisfy, so that we can determine fragments of the weighted $LTL$ whose formulas' can be effectively translated to equivalent weighted B\"uchi automata and their semantics are $k$-safe infinitary series ($k \in K\backslash \left\{ \mathbf{0},\mathbf{1}\right\}$). For the translation of our weighted $LTL$-formulas to equivalent weighted B\"uchi automata the properties considered in \cite{Ma-At} are sufficient in the case of totally generalized product $\omega$-valuation monoids as well. We further equip the structure with properties related to the total ordering of the monoid and the monotononicity of the $\omega$-valuation function, so that we can identify fragments of the weighted $LTL$ whose formulas' semantics are $k$-safe infinitary series. Whether the properties considered to obtain the aforementioned results are necessary it remains open.
	\begin{definition}
		\bigskip An $\omega$-valuation monoid $\left( K,+,Val^{\omega },\mathbf{0}\right) $
		is a complete monoid $\left(K,+,\mathbf{0}\right)$ equipped with an $\omega $%
		-valuation function $Val^{\omega }:\left( K_{fin}\right) ^{\omega
		}\rightarrow K$ such that $Val^{\omega }\left( k_{i}\right) _{i\in
			\mathbb{N}
		}=\mathbf{0}$ whenever $k_{i}=\mathbf{0}$ for some $i\geq 0.$ A totally generalized product $\omega $-valuation monoid (for short TGP-$\omega$-valuation monoid) $\left(
		K,+,\cdot ,Val^{\omega },\mathbf{0},\mathbf{1}\right) $ is an $\omega $%
		-valuation monoid $\left( K,+,Val^{\omega },\mathbf{0}\right) $ further
		equipped with a product operation $\cdot :K^{2}\rightarrow K,$ with $\mathbf{%
			1}\in K$ such that $Val^{\omega }\left( \mathbf{1}^{\omega }\right) =\mathbf{%
			1}$ and $\mathbf{0}\cdot k=k\cdot \mathbf{0=0,}$ $\mathbf{1\cdot }k=k\cdot
		\mathbf{1}=k$ for all $k\in K.$ Additionally, for every index set $I$ and $%
		k\in K,$ $\underset{I}{\sum }\left( k\cdot \mathbf{1}\right) =k\cdot $ $%
		\underset{I}{\sum }\mathbf{1}$, and for every $L\subseteq _{fin}K,$ finite
		index sets $I_{j}\left( j\geq 0\right) $ such that for all $j\geq 0,$ it
		holds $k_{i_{j}}\in L\backslash \left\{ \mathbf{0,1}\right\} $ for all $%
		i_{j}\in I_{j},$ or $k_{i_{j}}\in \left\{ \mathbf{0,1}\right\} $ for all $%
		i_{j}\in I_{j},$ we have
		\begin{equation}
			Val^{\omega }\left( \underset{i_{j}\in I_{j}}{\sum }k_{i_{j}}\right) _{j\in
				\mathbb{N}
			}=\underset{\left( i_{j}\right) _{j}\in I_{0}\times I_{1}\times \ldots }{%
				\sum }Val^{\omega }\left( k_{i_{j}}\right) _{j\in
				\mathbb{N}
			} \label{Property Dist}
		\end{equation}
	\end{definition}
	
	Equation 1 (in the sequel called Property 1) states the distributivity of $Val^{\omega}$ over finite sums that are taken over finite families of elements of $K$ that satisfy the following property: For each family, either every element of the family belongs to $\left\{\mathbf{0},\mathbf{1}\right\}$, or every element of the family belongs to $L\backslash \left\{\mathbf{0},\mathbf{1}\right\}$
	(where $L$ is a finite subset of $K$). We recall that Property
	\ref{Property Dist} has also been considered
	in \cite{Me-Do} (for a larger range of families of elements of $K$) for the definition of Cauchy $\omega$-indexed valuation monoids. Every GP-$\omega$-valuation monoid is also a
	TGP-$\omega $-valuation monoid. A GP-$\omega$-valuation monoid (resp. TGP-$%
	\omega $-valuation monoid) $\left( K,+,\cdot ,Val^{\omega },\mathbf{0},%
	\mathbf{1}\right)$ is called idempotent (resp. ordered) if $\left( K,+,%
	\mathbf{0}\right)$ is an idempotent monoid (resp. an ordered monoid). The results in \cite[Lemmas 1 and 2]{Ma-At} can be proved for idempotent TGP-$\omega$-valuation monoids by adopting the arguments of the corresponding proofs. In fact, in \cite[Lemma 1]{Ma-At} we prove properties of the infinitary sum operation~$\sum$ of an idempotent GP-$\omega$-valuation monoid $K$. More specifically, we prove complete idempotence under the condition that the cardinality of the index set, over which $\sum$ is applied, is at most continuum. We further prove that when $\sum$ is applied over all elements of $K$, and over all elements of any subset~$K^{\prime}$ of $K$, resulting sums $\underset{k\in K}\sum k$, $\underset{k\in K^{\prime}}\sum k$ preserve the order indicated by the subset relation, i.e., $\underset{k\in K^{\prime}}\sum k$$\leq$$\underset{k\in K}\sum k$. Lastly, we prove that $\sum$ is compatible with the addition of subsets of $K$.\footnote{For any two subsets $K^{\prime}$,$K^{\prime \prime}$ of $K$ we let the sum of $K^{\prime}$ with $K^{\prime \prime}$ be the set $K^{\prime}+K^{\prime \prime}=\left\{k^{\prime}+k^{\prime \prime} \mid k^{\prime} \in K^{\prime}, k^{\prime \prime} \in K^{\prime \prime}\right\}$.} In \cite[Lemma 2]{Ma-At}, (resp. \cite[Lemma 3ii]{Ma-At}) we identify conditions the elements of two subsets $K^{\prime}, K^{\prime \prime}$ of $K$ (resp. two countably infinite families $\left( k_{i}^{1}\right) _{i\geq 0}$, $\left(
	k_{i}^{2}\right) _{i\geq 0}$ of elements of $K$) should satisfy, so that we can conclude the order relation between $\underset{k\in K^{\prime}}\sum k$ and $\underset{k \in K^{\prime \prime}}\sum k$ (resp. between $Val^{\omega }\left(
	k_{i}^{1}\right) _{i\geq 0}$ and $Val^{\omega }\left( k_{i}^{2}\right) _{i\geq
		0}$). Following the arguments presented in the proof of
	\cite[Lemma 3ii]{Ma-At} we obtain a corresponding result for idempotent TGP-$\omega$-valuation monoids presented in the following lemma.
	
	\begin{lemma}
		\label{Valuation inequality}
		Let $\left( K,+,\cdot ,Val^{\omega },\mathbf{0},\mathbf{1}\right) $ be
		an idempotent TGP-$\omega$-valuation monoid, and $%
		L\subseteq _{fin}K.$ If $\left( k_{i}^{1}\right) _{i\geq 0}$ and $\left(
		k_{i}^{2}\right) _{i\geq 0}$ are families of elements of $L$ such that for
		every $i\geq 0,k_{i}^{1}\leq k_{i}^{2},$ and for all $i\geq 0,$ it holds $%
		\left\{ k_{i}^{1},k_{i}^{2}\right\} \subseteq L\backslash \left\{ \mathbf{0},%
		\mathbf{1}\right\} ,$ or $\left\{ k_{i}^{1},k_{i}^{2}\right\} \subseteq
		\left\{ \mathbf{0},\mathbf{1}\right\} ,$ then $Val^{\omega }\left(
		k_{i}^{1}\right) _{i\geq 0}\leq Val^{\omega }\left( k_{i}^{2}\right) _{i\geq
			0}$.
	\end{lemma}

	\begin{lemma}
		\label{Lemma-strict-inequality_finite} Let $\left( K,+,\cdot ,Val^{\omega },%
		\mathbf{0},\mathbf{1}\right) $ be an idempotent ordered TGP-$\omega $-valuation monoid, and $k,k_{1},k_{2}\in K$. If $k_{1}<k$,
		and $k_{2}<k$, then $k_{1}+k_{2}<k$.
	\end{lemma}
	
	\begin{proof}
		It holds $k_{1}<k$, and $k_{2}<k$. Thus, $k_{1}\leq k$, and $k_{2}\leq k$,
		Then, $k_{1}\leq k$, implies $k_{1}+k=k$ and since $k_{2}\leq k$, we get $k_{1}+\left(k_{2}+k\right)=\left(k_{1}+k_{2}\right)+k=k$, i.e., $k_{1}+k_{2}\leq k$.
		Assume now that $k_{1}+k_{2}=k$. Since~$K$ is ordered, it holds $k_{1}\leq
		k_{2}$, or $k_{2}\leq k_{1}$. If the first case is true, then by definition
		of the natural order induced by idempotency we get $k_{2}=k_{1}+k_{2}$, and
		thus $k_{2}=k$, which is a contradiction. If the second case is true, i.e.,
		if $k_{1}\leq k$, then following the arguments of the previous case we
		conclude $k_{1}=k$ which is a contradiction. Thus, $k_{1}+k_{2}\leq k$, and $%
		k_{1}+k_{2}\neq k$, i.e., $k_{1}+k_{2}< k$, as desired.
	\end{proof}
	
	In the rest of this paper we will consider idempotent ordered TGP-$\omega$-valuation monoids $\left( K,+,\cdot
	,Val^{\omega },\mathbf{0},\mathbf{1}\right)$ that further satisfy the
	following properties. For all $k,k_{i}\in K$ $(i\geq 1)$ it holds
	\begin{equation}
		Val^{\omega }\left( \mathbf{1,}k_{1},k_{2},k_{3},\ldots \right) =Val^{\omega
		}\left( k_{i}\right) _{i\geq 1},  \label{Property 1}
	\end{equation}
	\begin{equation}
		k=Val^{\omega }\left( k,\mathbf{1},\mathbf{1},\mathbf{1},\ldots \right) ,
		\label{Property 3}
	\end{equation}
	\begin{equation}
		k\leq \mathbf{1.}  \label{Property 5}
	\end{equation}%
	Moreover, for all $k\in K,$ and all $\left(k_{i}\right)_{i\geq 0}\in \left(K_{fin}\right)^{\omega}$ with $k_{i}\geq k$ $(i\geq 0)$ it~holds
	\begin{equation}
		Val^{\omega }\left( k_{0},k_{1},k_{2},\ldots \right) \geq k.
		\label{Property 6}
	\end{equation}
	
	Properties $\ref{Property 1}$, $\ref{Property 3}$ have also been considered in~$\cite{Ma-At}$. As it was stated in~$\cite{Ma-At}$ they express a notion of neutrality of $\mathbf{1}$ with respect to the $\omega$-valuation function, and are sufficient to establish the effective translation of formulas of fragments of the weighted $LTL$ over GP-$\omega$-valuation monoids introduced in that work to weighted B\"uchi automata. Property \ref{Property 5} states the fact that~$\mathbf{1}$ is the maximal element of the total order, while Property \ref{Property 6} expresses a restricted notion of monotonicity.
	
	\begin{example}\cite[Example 2]{Ma-At} \label{Example 1}
		Let $K_{1}=\left( \overline{%
			\mathbb{Q}
		},\sup ,\inf ,\text{liminf,}-\infty ,\infty \right) $ where $\overline{%
			\mathbb{Q}
		}=%
		\mathbb{Q}
		\cup \left\{ \infty ,-\infty \right\} $ and liminf is an $\omega$-valuation
		function from $\left( \overline{%
			\mathbb{Q}
		}_{fin}\right) ^{\omega }$ to $\overline{%
			\mathbb{Q}
		}$ defined by
		\begin{equation*}
			\text{liminf}\left( \left( d_{i}\right) _{i\geq 0}\right) =\left\{
			\begin{array}{ll}
				-\infty &
				\begin{array}{l}
					\text{if there exists }i\geq 0\\\text{with }d_{i}=-\infty%
				\end{array}
				\\
				&  \\
				\infty &
				\begin{array}{l}
					\text{if for all }i\geq 0\text{, }\\d_{i}=\infty \\
				\end{array}
				\\
				&  \\
				\underset{i\geq 0}{\sup}\left( \inf \left\{ d_{k}\mid k\geq
				i,d_{k}\neq \infty \right\} \right) &
				\begin{array}{l}
					\text{if }d_{j}\neq -\infty \text{ for all }\\j\geq 0,
					\text{and there exist}\\\text{infinitely }
					\text{many }i\geq 0\\\text{with }d_{i}\neq \infty%
				\end{array}
				\\
				&  \\
				\inf \left\{ d_{i}\mid i\geq 0\text{ with }d_{i}\neq \infty \right\} & \text{
					otherwise}%
			\end{array}%
			\right.
		\end{equation*}%
		$K_{1} $ is an idempotent ordered  GP-$\omega $%
		-valuation monoid that satisfies Properties \ref{Property 1}, \ref{Property 3}%
		, \ref{Property 5}, and \ref{Property 6}.
	\end{example}

	\begin{example}\label{Example 2}
		We consider the structure $K_{2}=\left( \overline{%
			\mathbb{Q}
		},\sup ,\inf ,\text{limsup,}-\infty ,\infty \right) $ where $\overline{%
			\mathbb{Q}
		}=%
		\mathbb{Q}
		\cup \left\{ \infty ,-\infty \right\} $ and limsup is an $\omega$-valuation
		function from $\left( \overline{%
			\mathbb{Q}
		}_{fin}\right) ^{\omega }$ to $\overline{%
			\mathbb{Q}
		}$ defined by
		\begin{equation*}
			\text{limsup}\left( \left( d_{i}\right) _{i\geq 0}\right) =\left\{
			\begin{array}{ll}
				-\infty &
				\begin{array}{l}
					\text{if there exists }i\geq 0\\\text{with }d_{i}=-\infty%
				\end{array}
				\\
				&  \\
				\infty &
				\begin{array}{l}
					\text{if for all }i\geq 0\text{, }\\d_{i}=\infty \\
				\end{array}
				\\
				&  \\
				\underset{i\geq 0}{\inf}\left( \sup \left\{ d_{k}\mid k\geq
				i,d_{k}\neq \infty \right\} \right) &
				\begin{array}{l}
					\text{if }d_{j}\neq -\infty \text{ for all }\\j\geq 0,
					\text{and there exist}\\\text{infinitely }
					\text{many }i\geq 0\\\text{with }d_{i}\neq \infty%
				\end{array}
				\\
				&  \\
				\sup \left\{ d_{i}\mid i\geq 0\text{ with }d_{i}\neq \infty \right\} & \text{
					otherwise}%
			\end{array}%
			\right.
		\end{equation*}%
		$K_{2} $ is an idempotent ordered  TGP-$\omega $%
		-valuation monoid that satisfies Properties \ref{Property 1}, \ref{Property 3}%
		, \ref{Property 5}, and \ref{Property 6}.
	\end{example}
	
	\begin{example}\label{Example 3}
		We consider the structure $K_{3}=\left( \overline{%
			\mathbb{Q}
		},\sup ,\inf ,\sup_{-\infty },-\infty ,\infty \right) $ where $\overline{%
			\mathbb{Q}
		}=%
		\mathbb{Q}
		\cup \left\{ \infty ,-\infty \right\} $ and $\sup_{-\infty }$ is an $\omega $%
		-valuation function from $\left( \overline{%
			\mathbb{Q}
		}_{fin}\right) ^{\omega }$ to $\overline{%
			\mathbb{Q}
		}$ defined by
		\begin{equation*}
			\sup_{-\infty }\left( \left( d_{i}\right) _{i\geq 0}\right) =\left\{
			\begin{array}{ll}
				-\infty & \text{ if }\exists i\geq 0\text{ with }d_{i}=-\infty \\
				&  \\
				\infty & \text{ if }d_{i}=\infty \text{ for all }i\geq 0 \\
				&  \\
				\sup \left\{ d_{k}\mid k\geq 0,d_{k}\neq \infty \right\} &
				\begin{array}{l}
					\text{if }d_{i}\neq -\infty \text{ for all }i\geq 0,\\\text{and }
					\exists i\geq 0\text{ with }d_{i}\neq \infty .%
				\end{array}%
			\end{array}%
			\right.
		\end{equation*}%
		$K_{3}$ is an idempotent ordered  TGP-$\omega $%
		-valuation monoid that satisfies Properties \ref{Property 1}, \ref{Property 3}%
		, \ref{Property 5}, and \ref{Property 6}.
	\end{example}
	
	The proofs that $K_{2},$ $K_{3}$ are idempotent ordered TGP-$\omega $-valuation monoids that satisfy Properties \ref{Property 1}%
	, \ref{Property 3}, \ref{Property 5}, and \ref{Property 6}
	can be found in the Appendix. We also note that $\omega$-valuation functions over the reals that incorporate classical $sup$, limsup in their definitions have also appeared in
	\cite{Me-Do,Dr-De,Dr-Me}. As it is stated in $\cite{Ch-Do}$, for $ \left( d_{i}\right) _{i\geq 0}$$\in \left( \overline{%
		\mathbb{Q}
	}_{fin}\right) ^{\omega }$ the value $\underset{i\geq 0}{\inf}\left(\sup\left\{ d_{k} \mid k \geq i \right\}\right)=\underset{i\geq 0}{\text {lim}}\left(\sup\left\{ d_{k} \mid k \geq i \right\}\right)$ corresponds to the greatest element of $ \overline{%
		\mathbb{Q}
	}$ that appears infinitely often along the sequence $ \left( d_{i}\right) _{i\geq 0}$, while sup${\left( \left( d_{i}\right) _{i\geq 0}\right)}$ equals the greatest element of $ \overline{%
		\mathbb{Q}
	}$ that appears in $ \left( d_{i}\right) _{i\geq 0}$. Finally, the value $\underset{i\geq 0}{\sup}$$\left(\inf\left\{ d_{k} \mid k \geq i \right\}\right)=\underset{i\geq 0}{\text {lim}}\left(\inf\left\{ d_{k} \mid k \geq i \right\}\right)$ corresponds to the least element of $ \overline{%
		\mathbb{Q}
	}$ that appears infinitely often along the sequence~$ \left( d_{i}\right) _{i\geq 0}$.
	
	\begin{remark}
		\label{Remark tos}
		Let $\left( K,\leq \right) $ be a totally ordered set with a minimum, and a
		maximum element which are respectively denoted by $\mathbf{0},$ and $\mathbf{%
			1}$. Then, $\left( K,\sup_{\leq },\mathbf{0}\right) ,$ $\left( K,\inf_{\leq
		},\mathbf{1}\right) $ are idempotent complete monoids. $\sup_{\leq },$ $%
		\inf_{\leq }$ stand for the supremun, and infimum respectively of possibly
		infinite families of elements of $K,$ and are defined in the classical way
		with respect to $\leq$-ordering of $K$. We define the $\omega$-valuation functions liminf$_{\leq }$, limsup$%
		_{\leq }$, $\sup_{0}$ from $\left( K_{fin}\right) ^{\omega }$ to $K$ in a
		corresponding way with the way the $\omega $-valuation functions liminf,
		limsup, and sup$_{-\infty }$ are defined in Examples \ref{Example 1}, \ref{Example 2}%
		, and \ref{Example 3} respectively. More precisely, the definitions of liminf%
		$_{\leq }$, limsup$_{\leq }$, $\sup_{0}$ are obtained by the definitions of
		liminf, limsup, and sup$_{-\infty }$ by substituting $\inf ,$ $\sup ,$ -$%
		\infty ,$ $\infty $ by $\inf_{\leq},$ $\sup_{\leq},$ $\mathbf{0},$ $\mathbf{1}$
		respectively. We consider the structures $M_{1}=\left( K,\sup_{\leq
		},\inf_{\leq },\text{liminf}_{\leq },\mathbf{0},\mathbf{1}\right) $, $%
		M_{2}=\left( K,\sup_{\leq },\inf_{\leq },\text{limsup}_{\leq },\mathbf{0,1}%
		\right) ,$ $M_{3}=\left( K,\sup_{\leq },\inf_{\leq },\sup_{\mathbf{0}},%
		\mathbf{0,1}\right).$ Then, by adopting the arguments of the proof of \cite[Example 2]{Ma-At}, and of Examples \ref{Example 2}, \ref{Example 3}
		presented in the Appendix we get that $M_{1}$ (resp.  $M_{2},$ $M_{3}$) is an (resp. are)
		idempotent ordered GP-$\omega $-valuation monoid (resp. idempotent ordered TGP-$\omega$-valuation monoids) that satisfies (resp. satisfy) Properties \ref%
		{Property 1}, \ref{Property 3}, \ref{Property 5}, and \ref%
		{Property 6}.
	\end{remark}
	
	In the following example we present totally ordered sets with maximum, and minimum element.
	
	\begin{example}
		Let $N=%
		\mathbb{Z}
		_{>}\cup \left\{ 0\mathbf{,+\infty }\right\} ,$ and $K=N\times N$. We consider the pair $\left(
		K,\preceq \right) $ where $\preceq $ stands for the order relation
		defined in the following way: For every $\left( a_{1},a_{2}\right) ,\left(
		b_{1},b_{2}\right) \in K$ we let $\left( a_{1},a_{2}\right) \preceq
		\left( b_{1},b_{2}\right) $ if $a_{1}<b_{1},$ or if $a_{1}=b_{1}$ and $%
		a_{2}\leq b_{2}.$ Clearly, $\left( K,\preceq \right) $ is totally
		ordered with minimum element $\left( 0,0\right) ,$ and maximum element $\left(
		+\infty ,+\infty \right) .$
		
		We consider now the pair $\left(
		K,\leqq \right) $ where $\leqq $ stands for the order relation
		defined in the following way: For every $\left( a_{1},a_{2}\right) ,\left(
		b_{1},b_{2}\right) \in K$ we let $\left( a_{1},a_{2}\right) \leqq
		\left( b_{1},b_{2}\right) $ if $a_{1}+a_{2}<b_{1}+b_{2},$ or if $%
		a_{1}+a_{2}=b_{1}+b_{2}$ and $a_{1}\leq b_{1}.$ Clearly, $\left( K,\leqq \right) $ is totally ordered with minimum element $\left( 0,0\right) ,$
		and maximum element $\left( +\infty ,+\infty \right) .$
	\end{example}

	In the previous example elements of $K$ could represent two distinct indications of energy
	consumption of a system. If $\left( d_{i}\right) _{i\geq 0}$ is a
	sequence of measurements of these indications at distinct time moments, $\sup_{\leqq }\left( \left( d_{i}\right) _{i\geq 0}\right) $
	corresponds to the least pair $\left( a_{1},a_{2}\right) $ (according to $\leqq$-ordering) with the property that
	at every time moment the overall energy consumption is bounded by the sum $%
	a_{1}+a_{2}$, and whenever it reaches the value $a_{1}+a_{2}$ the first
	indication of energy consumption is bounded by $a_{1}.$ Alternatively,
	elements of $K$ could represent two independent function indications
	of a device. In this case, if $\left( d_{i}\right) _{i\geq 0}$ is a sequence
	of measurements of these indications at distinct time moments, then $%
	\sup_{\preceq }\left( \left( d_{i}\right) _{i\geq 0}\right) $ corresponds to
	the least pair $\left( b_{1},b_{2}\right) $ (according to $\preceq$-ordering) with the property that at every time
	moment the first indication is bounded by $b_{1},$ and whenever it reaches $%
	b_{1},$ the second indication is bounded by $b_{2}.$ Clearly, $\preceq $%
	-ordering, and $\leqq$-ordering of $K$ can be used to argue about quantitative properties
	of different aspects of the behavior of a system whenever priority is of
	interest.
	
	\bigskip
	
	In the rest of the paper, we will refer to
	idempotent ordered GP-$\omega $-valuation monoids (resp. idempotent ordered TGP-$\omega $-valuation monoids) that satisfy Properties \ref%
	{Property 1}, \ref{Property 3}, \ref{Property 5}, and \ref%
	{Property 6} as idempotent ordered GP-$\omega$-valuation monoids (resp. idempotent ordered TGP-$\omega$-valuation monoids).
	
	\bigskip
	
	In \cite{Ma-At} weighted generalized B\"{u}chi automata with $%
	\varepsilon $-transitions ($\varepsilon $-wgBa for short), and weighted B\"{u}chi automata with $%
	\varepsilon $-transitions ($\varepsilon$-wBa for short) over idempotent ordered GP-$%
	\omega$-valuation monoids were introduced, and their expressive equivalence with weighted B\"{u}chi automata (wBa for short) over idempotent ordered GP-$\omega$-valuation monoids was proved (see \cite[Lemmas~4 and 5ii]{Ma-At}).
	The definitions of $\varepsilon$-wgBa, $\varepsilon $-wBa, and wBa over GP-$\omega$-valuation monoids in \cite[Definition 3]{Ma-At} can be generalized for idempotent ordered TGP-$\omega$-valuation monoids in a straightforward way. The expressive equivalence of $\varepsilon$-wgBa, $\varepsilon$-wBa, and wBa over idempotent ordered TGP-$\omega$-valuation monoids can be proved with the same arguments that are used in the proofs of \cite[Lemmas 4 and 5ii]{Ma-At}. We recall the definition of wBa. Let $K$ be an idempotent ordered TGP-$\omega$-valuation monoid.
	
	\begin{definition}
		\label{definition2} A weighted B\"{u}chi automaton over $A$ and $K$
		is a quadruple $\mathcal{A}=\left( Q,wt,I,F\right)$, where $Q$ is
		the finite set of states, $wt:Q\times A \times Q\rightarrow K$ is a mapping assigning weights to
		the transitions of the automaton, $I\subseteq Q$ is the set of initial states, and $F\subseteq Q$ is the set of final states.
	\end{definition}
	
	Let $%
	w=a_{0}a_{1}\ldots \in A^{\omega }$ with $a_{i}\in A\left( i\geq 0\right) .$
	A path $P_{w}$ of $\mathcal{A}$ over $w$ is an infinite sequence of
	transitions $P_{w}=\left( q_{j},a_{j},q_{j+1}\right) _{j\geq 0}$. We let the weight of $%
	P_{w}$ be the value
	\begin{equation*}
		weight_{\mathcal{A}}\left( P_{w}\right) =Val^{\omega }\left( wt\left( q_{j},a_j,q_{j}\right) \right)_{j\geq0}
	\end{equation*}
	The set $In^{Q}\left( P_{w}\right)$ is defined as
	for Ba. The path $P_{w} $ is called \textit{successful} if $q_{0}\in I,$ and $%
	In^{Q}\left( P_{w}\right) \cap F\neq \emptyset $. We shall denote by $succ_{\mathcal{A}}\left(
	w\right)$ the set of all successful paths of $\mathcal{A}$ over $w.$ The
	behavior of $\mathcal{A}$ is the infinitary series $\left\Vert \mathcal{A}%
	\right\Vert :A^{\omega }\rightarrow K$ with coefficients specified, for
	every $w\in A^{\omega }$,
	\begin{equation*}
		\left( \left\Vert \mathcal{A}\right\Vert ,w\right) =\underset{P_{w}\in succ_{%
				\mathcal{A}}\left( w\right) }{\sum }weight_{\mathcal{A}}\left( P_{w}\right).
	\end{equation*}
	
	\begin{example}
		A wBa automaton $\mathcal{M}$ over $K_{2}$ and $A=\left\{a,b\right\}$ is presented in Figure 1. The behavior of $\mathcal{M}$ is defined as follows. $\left( \left\Vert
		\mathcal{M}\right\Vert ,w\right) =\infty $ if $w=b^{\omega },$ $\left(
		\left\Vert \mathcal{M}\right\Vert ,w\right) =2$ if $w=a\overline{w}$ for
		some $\overline{w}\in A^{\omega }$ with a finite number of appearances of $b,$
		$\left( \left\Vert \mathcal{M}\right\Vert ,w\right) =3$ if $w=a\overline{w}$
		for some $\overline{w}\in A^{\omega }$ with an infinite number of appearances
		of $b,$ and $\left( \left\Vert \mathcal{M}\right\Vert ,w\right) =-\infty $
		otherwise. We note that in the graphical representation of the $\mathcal{M}$ presented in Figure 1, we only draw transitions with weight different from $-\infty$.
	\end{example}

	Two wBa are called equivalent if they have the same behavior. A wBa $\mathcal{M=}\left( Q,wt,I,F\right)$ will be called normalized if
	the set of initial states $I$ is a singleton, i.e., if $I=\left\{q_{0}\right%
	\}$. In this case we denote it by $\mathcal{M=}\left(Q,wt,q_{0},F\right)$.
	
	\begin{lemma}
		For every wBa over $A$ and $K$, we can effectively construct an equivalent
		normalized wBa over $A$ and $K$.
	\end{lemma}
	
	\begin{sketch}
		Let $\mathcal{M=}\left( Q,wt,I,F\right)$ be a wBa over $A$ and $K$. We
		define the normalized wBa $\mathcal{A=}\left(
		Q^{\prime},wt^{\prime},q_{0},F^{\prime}\right)$ by setting $%
		Q^{\prime}=Q\cup\left\{q_{0}\right\}\cup S$ where $q_{0}$ is a new state,
		and $S=\left\{s_{q}\mid q\in Q\right\}$ is a set of copies of $Q$, $%
		F^{\prime}=F\cup\left\{s_{q}\in S\mid q\in F\right\}$, and $%
		wt^{\prime}:Q^{\prime}\times A \times Q^{\prime}\rightarrow K$ is defined in
		the following way:
		\begin{itemize}
			\item[]$\begin{array}{ll}
				\bullet wt^{\prime}\left(q_{0},a,\overline{p}\right)=\underset{%
					p^{\prime}\in I, wt(p^{\prime},a,\overline{p})\neq \mathbf{1}, \mathbf{0}}{%
					\sum} wt(p^{\prime},a,\overline{p}) & {\text{  if  } } \overline{p}\in Q,
				\\
				&
				\\
				\bullet wt^{\prime}\left( q_{0},a,s_{q}\right)=\mathbf{1} &
				\begin{array}{l}
					\text{if }%
					s_{q}\in S,\\\text{and there exists }\widetilde{%
						q}\in I,\\\text{with } wt\left( \widetilde{q},a,q \right) =\mathbf{1%
					}
				\end{array}
				\\
				&\\
				\bullet wt^{\prime}\left(q,a,\overline{q}\right)=wt\left(q,a,\overline{q}%
				\right)& \text{  if  } q, \overline{q} \in Q,
				\\&\\
				\bullet wt^{\prime}\left(s_{q},a,s_{%
					\overline{q}}
				\right)=wt\left(q,a,\overline{q}\right)& \text{  if  } s_{q},  s_{%
					\overline{q}} \in S,
				\\
				&\\
				\bullet wt^{\prime}\left(p,a,\overline{p}\right)=\mathbf{0} & \text{ otherwise}
			\end{array}$
		\end{itemize}
		Then, for every $w=a_{0}a_{1}\ldots \in A^{\omega}$, using the completeness axioms of the monoid we get that $\left( \left\Vert \mathcal{A}\right\Vert ,w\right)$ is equal to the sum of all successful paths of $\mathcal{A}$ over $w$ with non-zero weight, which, by the definition of $wt^{\prime}$ and the completeness axioms of the monoid, is in turn equal to the sum $\overline{\sum}+\widetilde{\sum}$~where
		\begin{gather*}
			\overline{\sum}=\underset{weight_{%
					\mathcal{A}}\left( P_{w}\right) \neq \mathbf{0}}{\underset{q_{i}\in Q,\text{
					}i\geq 1}{\underset{\in
						succ_{\mathcal{A}}\left( w\right)}{\underset{P_{w}=\left( q_{i},a_{i},q_{i+1}\right) _{i\geq 0} }{\sum }}}}weight_{\mathcal{A}}\left(
			P_{w}\right), \text{   and }
		\end{gather*}
		\begin{gather*}
			\widetilde{\sum}=\underset{weight_{\mathcal{A}}\left( P_{w}\right) \neq \mathbf{%
					0}}{\underset{q_{i}\in S,\text{ }i\geq 1}{\underset{\in succ_{\mathcal{A}}\left( w\right) }{\underset{P_{w}=\left(
							q_{i},a_{i},q_{i+1}\right) _{i\geq 0}}{%
							\sum }}}}weight_{\mathcal{A}}\left( P_{w}\right).
		\end{gather*}
		For every non-zero weight successful path $P_{w}=\left( q_{i},a_{i},q_{i+1}\right) _{i\geq 0}$ of $\mathcal{A}$ over~$w$, with $q_{i}\in Q,i\geq 1$, we obtain, using Property 1 and the definition of $wt^{\prime}$, the~following
		\begin{gather*}
			weight_{\mathcal{A}}\left(P_{w}\right)\\=Val^{\omega }\left( \underset{wt\left(
				p^{\prime },a_{0},q_{1}\right) \neq \mathbf{0,1}}{\sum\limits_{p^{\prime }\in I}}wt\left( p^{\prime
			},a_{0},q_{1}\right),wt^{\prime }\left( q_{1},a_{1},q_{2}\right)
			,wt^{\prime }\left( q_{2},a_{2},q_{3}\right) ,\ldots \right)
			\\= \sum\limits_{p^{\prime }\in I,wt\left( p^{\prime
				},a_{0},q_{1}\right) \neq \mathbf{0,1}}Val^{\omega }\left( wt\left(
			p^{\prime },a_{0},q_{1}\right) ,wt\left( q_{1},a_{1},q_{2}\right)
			,wt\left( q_{2},a_{2},q_{3}\right) ,\ldots \right).
		\end{gather*}
		Moreover, for every non-zero weight successful path $P_{w}=\left( q_{i},a_{i},q_{i+1}\right) _{i\geq 0}$ of~$\mathcal{A}$ over $w$ with $q_{i}\in S$ for every $i\geq 1$, there exist $p\in I, p_{i}\in Q\left(i\geq 0\right)$ such that $q_{i}=s_{p_{i}}$ for every $i\geq 1$ and
		\begin{gather*}
			weight_{\mathcal{A}}\left(P_{w}\right)=Val^{\omega }\left(\mathbf{1}, wt\left( p_{1},a_{1},p_{2}\right)
			,wt\left( p_{2},a_{2},p_{3}\right) ,\ldots \right)\\=Val^{\omega }\left(wt\left(p,a_{0},p_{1}\right), wt\left( p_{1},a_{1},p_{2}\right)
			,wt\left( p_{2},a_{2},p_{3}\right) ,\ldots \right).
		\end{gather*}
		We conclude that each term of $\overline{\sum}$ and $\widetilde{\sum}$ is equal to the weight of a non-zero weight successful path of $\mathcal{M}$ over $w$. Also, the weight of each non-zero weight successful path $P_{w}$ of $\mathcal{M}$ over $w$ appears as a term in at least one of $\overline{\sum}$, $\widetilde{\sum}$. More precisely, if $P_{w}=\left( p_{i},a_{i},p_{i+1}\right) _{i\geq 0}$ is a non-zero weight successful path of $\mathcal{M}$ over $w$ with $wt\left(p_{0},a_{0},p_{1}\right)=\mathbf{1}$, then $weight_{\mathcal{M}}\left(P_{w}\right)$ equals the weight of the successful path $\left(q_{0},a_{0},s_{p_{1}}\right)\left(s_{p_{i}},a_{i},s_{p_{i+1}}\right)_{i\geq 1}$ of $\mathcal{A}$ over $w$. Now, if $wt\left(p_{0},a_{0},p_{1}\right)\neq \mathbf{1}$, then the weight of the successful path $\left(q_{0},a_{0},p_{1}\right)\left( p_{i},a_{i},p_{i+1}\right) _{i\geq 1}$ of $\mathcal{A}$ over $w$ is equal to a finite sum with one of its terms being $weight_{\mathcal{M}}\left(P_{w}\right)$. We thus, by the completeness axioms of the monoid and the generalization of \cite[ Lemma 1(ii)]{Ma-At} for idempotent TGP-$\omega$-valuation monoids, can conclude that $\left( \left\Vert \mathcal{A}\right\Vert ,w\right)=\left( \left\Vert \mathcal{M}\right\Vert ,w\right)$ as~desired.
		
	\end{sketch}
	
	\setlength{\unitlength}{3mm}
	\thinlines
	\begin{picture}(15,15)
		\put(0.5,0.5){\circle{3}}
		\put(8,1){\circle{3}}
		\put(8,1){\circle{3.3}}
		\put(8.9,8.7){\circle{3}}
		\put(8.9,8.7){\circle{3.3}}
		\put(0.5,0.5){\makebox(0,0)[cc]{$q_{0}$}}
		\put(8,1){\makebox(0,0)[cc]{$q_{1}$}}
		\put(8.9,8.7){\makebox(0,0)[cc]{$q_{2}$}}
		\put(2,1){\vector(1,1){6}}
		\put(4,5){\makebox(0,0)[cc]{$b,\infty$}}
		\put(2,1){\vector(1,0){4}}
		\put(4,0.3){\makebox(0,0)[cc]{$a,2$}}
		\put(1,4){\vector(0,-1){2.3}}
		\put(10,1){\oval(5,5)[rt]}
		\put(10,10){\oval(5,5)[rt]}
		\put(10,12.4){\vector(-2,-3){1.5}}
		\put(10,3.4){\vector(-2,-3){1}}
		\put(13,3){\makebox(0,0)[cc]{$a,2$}}
		\put(13.1,2){\makebox(0,0)[cc]{$b,3$}}
		\put(12.5,13){\makebox(0,0)[cc]{$b,\infty$}}
		\put(12.6,10){\line(-1,0){2.8}}
		\put(12.6,1){\line(-1,0){2.9}}
	\end{picture}
	\begin{picture}(15,5)
		\put(11,5){\makebox(0,0)[cc]{\shortstack{
					\textbf{Figure 1}}}}
	\end{picture}
	
	\bigskip
	\bigskip
	In all following sections, $K$ will stand for an idempotent ordered TGP-$\omega$-valuation monoid unless specified otherwise.

	\section{Describing Weighted Safety with Weighted LTL over Idempotent
		Ordered TGP-$\protect\omega $-valuation \\Monoids\label%
		{Desribing ws with wltl}}
	
	In this section we firstly aim to introduce a notion of weighted safety for infinitary
	series over idempotent ordered TGP-$\omega $%
	-valuation monoids. More precisely, given an idempotent ordered TGP-$\omega $-valuation monoid $K,$ we will define for every
	$k\in K\backslash \left\{ \mathbf{0,1}\right\} $ the notion of $k$-safe
	infinitary series. We prove that an infinitary series $s$
	over $K$ and an alphabet $A,$ is $k$-safe iff the infinitary language that
	contains all $w\in A^{\omega }$ whose coefficient on $s$ is greater or equal
	to $k$ is an infinitary safety language.
	We provide fragments of a weighted specification
	language whose formulas' semantics are infinitary series that express this
	notion of weighted safety. More specifically, we will extend the definition of the weighted \textit{LTL}
	over idempotent ordered GP-$\omega $-valuation monoids
	presented in \cite{Ma-At} to idempotent ordered TGP-$%
	\omega $-valuation monoids, and determine for every $k\backslash \left\{
	\mathbf{0,1}\right\} $ syntactic fragments of this logic such that the
	semantics of each formula in these fragments is a $k$-safe infinitary series. By presenting a counterexample we conclude that
	for arbitrary idempotent ordered TGP-$\omega $%
	-valuation monoids $K,$ and $k\in K\backslash \left\{ \mathbf{0},\mathbf{1}%
	\right\} ,$ the class of semantics of formulas in the aforementioned fragments does not
	coincide\ with the class of $k$-safe infinitary series.
	
	\begin{definition}
		\label{Def k-safe}\bigskip Let $k\in K\backslash \left\{ \mathbf{0},\mathbf{1%
		}\right\} $ and $s\in K\left\langle \left\langle A^{\omega }\right\rangle
		\right\rangle .$ Series $s$ will be called $k$-safe if for every $w\in A^{\omega }$
		the following hold: If $\forall i>0,\exists u\in A^{\omega }$ such that $%
		\left( s,w_{<i} u\right) \geq k,$ then $\left( s,w\right) \geq k.$
	\end{definition}
	
	\begin{example}
		\label{example_infinite_image} Let $A=\left\{ a,b\right\} $. We consider the
		infinitary series $S\in K_{2}\left\langle \left\langle A^{\omega
		}\right\rangle \right\rangle $ defined in the following way. For every $w\in
		A^{\omega }$, we let $\left( S,w\right) =3^{k}$ if $w\in a^{k}b^{+}A^{\omega
		} $ where $k>0,$ or $\left( S,w\right) =3$ if $w=a^{\omega }$, or $\left(
		S,w\right) =-\infty $ otherwise. Clearly,
		\begin{gather*}
			\left( S,w\right) \geq 3 \text{  iff  }
			w\in a^{\omega }\cup \left( \underset{k>0}{\bigcup }%
			a^{k}b^{+}A^{\omega }\right)
		\end{gather*}
		where $a^{\omega }\cup \left( \underset{k>0}{\bigcup }%
		a^{k}b^{+}A^{\omega }\right) =a^{+}b^{+}A^{\omega }\cup a^{\omega }=aA^{\omega}$. We prove that $S$ is $3$-safe. Let $w\in A^{\omega }$ such that $\forall
		i>0,\exists u\in A^{\omega }$ with $\left( S,w_{<i} u\right) \geq 3$. Then, $%
		\forall i>0,$ $w_{<i} u=aw^{\prime }$ for some $%
		w^{\prime }\in A^{\omega }$. Then, $w\left(0\right)=a$, which implies $w\in aA^{\omega}.$ Thus, $\left(S,w\right) \geq 3$, which implies that $S$ is $3$-safe as desired.
	\end{example}
	
	\begin{example}
		\label{example_infinite_image2} Let $A=\left\{ a,b\right\} $. We consider
		three infinitary series $R\in K_{2}\left\langle \left\langle A^{\omega
		}\right\rangle \right\rangle $ defined in the following way. For every $w\in
		A^{\omega }$ we let $\left( R,w\right) =3^{k+2}$ if $w\in
		a^{k}b^{+}A^{\omega }$ where $k>0,$ or $\left( R,w\right) =3$ if $%
		w=a^{\omega }$, or $\left( R,w\right) =-\infty $ otherwise. Clearly,
		\begin{gather*}
			\left(
			R,w\right) \neq -\infty
			\text{  iff  }
			w\in aA^{\omega }.
		\end{gather*}
		We can prove that $R$ is $3$-safe, with the same arguments we used in the
		previous example. We prove that $R$ is not $4$-safe. To this end, let $%
		w=a^{\omega }$. Then, $\forall i>0,$ $\left( R,w_{<i} b^{\omega }\right)
		=3^{i+1}\geq 3^{2}>4$, however $\left( R,w\right) =3<4,$ and this concludes
		our claim.
	\end{example}

	\begin{lemma}
		\label{Lemma k}Let $k\in K\backslash \left\{ \mathbf{0},\mathbf{1}\right\}, k^{\prime}\in K. $
		Then, $k^{\prime }\in K\left\langle \left\langle A^{\omega }\right\rangle
		\right\rangle $ is $k$-safe.
	\end{lemma}
	
	\begin{proof}
		If $k^{\prime }<k,$ then our claim trivially holds. If $k^{\prime }\geq k,$
		then for every $w\in A^{\omega },$ $\left( k^{\prime},w\right) \geq k,$ and it holds:
		$\forall i>0,\exists u\in A^{\omega }$ such that $\left( k^{\prime},w_{<i} u\right)
		\geq k,$ and $\left( k^{\prime},w\right) \geq k$ which proves our claim.$\ $
	\end{proof}
	\begin{theorem}
		\label{Theorem weighted safety - boolean safety}Let $k\in K\backslash
		\left\{ \mathbf{0},\mathbf{1}\right\} $ and $s\in K\left\langle \left\langle
		A^{\omega }\right\rangle \right\rangle .$ Then, $s$ is $k$-safe iff $L_{\geq
			k}\left( s\right) $ is a safety language.
	\end{theorem}

	\begin{proof}
		Let $s$ be $k$-safe. Then, for every $w\in A^{\omega },$ it holds: If $%
		\forall i>0,\exists u\in A^{\omega }$ such that $\left( s,w_{<i} u\right)
		\geq k,$ then $\left( s,w\right) \geq k$. This implies that for every $w\in
		A^{\omega },$ it holds: If $\forall i>0,\exists u\in A^{\omega }$ such that $%
		w_{<i} u\in L_{\geq k}\left( s\right) ,$ then $w\in L_{\geq k}\left( s\right)
		$, which in turn implies that $L_{\geq k}\left( s\right) $ is a safety
		language.
		
		Now, we let $L_{\geq k}\left( s\right)$ be a safety language. Then, for
		every $w\in A^{\omega },$ it holds: If $\forall i>0,\exists u\in A^{\omega }$
		such that $w_{<i} u\in L_{\geq k}\left( s\right) ,$ then $w\in L_{\geq
			k}\left( s\right) .$ Thus, for every $w\in A^{\omega },$ it holds: If $%
		\forall i>0,\exists u\in A^{\omega }$ such that $\left( s,w_{<i} u\right)
		\geq k,$ then $\left( s,w\right) \geq k,$ and so $s$ is $k$-safe as desired.
	\end{proof}

	\begin{remark}
		We can obtain a to Definition \ref{Def k-safe} equivalent definition for $k$-safe infinitary series by
		considering the contrapositive condition. More specifically, for $k\in
		K\backslash \left\{ \mathbf{0,1}\right\} $, $s\in K\left\langle \left\langle
		A^{\omega }\right\rangle \right\rangle $ is $k$-safe if for every $w\in
		A^{\omega }$ it holds: if $\left( s,w\right) <k,$ then there exists an $i>0$ such that
		for all $u\in A^{\omega },$ $\left( s,w_{<i} u\right) <k.$ The equivalence of
		the two definitions can be obtained in a straightforward way using Theorem $\ref{Theorem weighted safety - boolean safety}$.
	\end{remark}

	We recall now from \cite{Ma-At} that the syntax of the weighted $LTL$ over $%
	AP$ and $K$ is given by the grammar
	\begin{equation*}
		\varphi ::=k\mid a\mid \lnot a\mid \varphi \vee \varphi \mid \varphi \wedge
		\varphi \mid \bigcirc \varphi \mid \varphi U\varphi \mid \square \varphi
	\end{equation*}%
	where $k\in K$ and $a\in AP.$ As usual we identify $\lnot \lnot a$ with $a$
	for every $a\in AP.$ We shall denote by $LTL\left( K,AP\right)$ the class
	of all weighted $LTL$ formulas over $AP$ and $K.$ We note that formulas in $LTL\left( K,AP\right)$ are by construction in Negation Normal Form, i.e., negation is applied only to atomic propositions.
	
	\begin{definition}
		The semantics $\Vert \varphi \Vert$ of formulas $\varphi \in LTL\left(
		K,AP\right) $ are represented as infinitary series in $K\left\langle
		\left\langle \left( \mathcal{P}\left( AP\right) \right) ^{\omega
		}\right\rangle \right\rangle $ inductively defined in the following way. \
		
		\begin{itemize}
			\item[-] $\left( \left\Vert k\right\Vert ,w\right) =k$ \ \ \ \ \ \ \ \ \ \ \
			\ \ \ \ \ \ \ \ \ \ \ \ \ \ \ - $\left( \left\Vert \varphi \wedge \psi
			\right\Vert ,w\right) =\left( \left\Vert \varphi \right\Vert ,w\right) \cdot
			\left( \left\Vert \psi \right\Vert ,w\right) $
			
			\item[-] $\left( \left\Vert a\right\Vert ,w\right) =\left\{
			\begin{array}{cc}
				\mathbf{1} & \text{if }a\in w\left( 0\right) \\
				\mathbf{0} & \text{otherwise}%
			\end{array}%
			\right. $\ \ \ \ \ - $\left( \left\Vert \varphi \vee \psi \right\Vert
			,w\right) =\left( \left\Vert \varphi \right\Vert ,w\right) +\left(
			\left\Vert \psi \right\Vert ,w\right) $
			
			\item[-] $\left( \left\Vert \lnot a\right\Vert ,w\right) =\left\{
			\begin{array}{ll}
				\mathbf{1} & \text{if }a\notin w\left( 0\right) \\
				\mathbf{0} & \text{otherwise}%
			\end{array}%
			\right. $ \ -\ $\left( \left\Vert \bigcirc \varphi \right\Vert ,w\right)
			=\left( \left\Vert \varphi \right\Vert ,w_{\geq 1}\right) $
			
			\item[-] $\left( \left\Vert \varphi U\psi \right\Vert ,w\right) $
			
			$=\sum\limits_{i\geq 0}Val^{\omega }\left( \left( \left\Vert \varphi
			\right\Vert ,w_{\geq 0}\right) ,\ldots ,\left( \left\Vert \varphi
			\right\Vert ,w_{\geq i-1}\right) ,\left( \left\Vert \psi \right\Vert
			,w_{\geq i}\right) ,\mathbf{1},\mathbf{1},\ldots \right) $
			
			\item[-] $\left( \left\Vert \square \varphi \right\Vert ,w\right)
			=Val^{\omega }\left( \left( \left\Vert \varphi \right\Vert ,w_{\geq
				i}\right) \right) _{i\geq 0}$
		\end{itemize}
	\end{definition}
	
	We shall denote by $true$ the formula $\mathbf{1}\in K.$ Two formulas $%
	\varphi ,\psi \in LTL\left( K,AP\right) $ will be called equivalent, and we
	will denote this by $\varphi \equiv \psi ,$ iff $\left( \left\Vert \varphi
	\right\Vert ,w\right) =\left( \left\Vert \psi \right\Vert ,w\right) $ for
	every $w\in \left( \mathcal{P}\left( AP\right) \right) ^{\omega }.$ We
	define the \textit{weak until} operator $\widetilde{U}$ by letting $\varphi
	\widetilde{U}\psi :=\square \varphi \vee \left( \varphi U\psi \right) .$ The
	syntactic boolean fragment $bLTL\left( K,AP\right) $ of $LTL\left(
	K,AP\right) $ is given by the grammar
	\begin{equation*}
		\varphi ::=\mathbf{0}\mid true\mid a\mid \lnot a\mid \varphi \vee \varphi
		\mid \varphi \wedge \varphi \mid \bigcirc \varphi \mid \varphi \mathit{U}%
		\varphi \mid \square \varphi
	\end{equation*}%
	where $a\in AP.$ Inductively, we can prove that $Im\left( \left\Vert \varphi
	\right\Vert \right) \subseteq \left\{ \mathbf{0,1}\right\} $ for every $%
	\varphi \in bLTL\left( K,AP\right), $ and the semantics of the formulas in $%
	bLTL\left( K,AP\right) $ and corresponding classical $LTL$ formulas
	coincide. The syntactic fragment $sbLTL\left( K,AP\right) $ of $bLTL\left(
	K,AP\right) $ is given by the grammar
	\begin{equation*}
		\varphi ::=true\mid a\mid \lnot a\mid \varphi \vee \varphi \mid \varphi
		\wedge \varphi \mid \bigcirc \varphi \mid \varphi \widetilde{\mathit{U}}%
		\varphi \mid \square \varphi .
	\end{equation*}%
	We observe that for every $\varphi \in sbLTL\left( K,AP\right) $ it holds $%
	\varphi $ in $sLTL\left( AP\right) $, and for every $w\in \left( \mathcal{P}%
	\left( AP\right) \right) ^{\omega },$ $\left( \left\Vert \varphi \right\Vert
	,w\right) =\mathbf{1}$ iff $w\models \varphi .$ We will call every formula
	in $sbLTL\left( K,AP\right) $ a boolean safety formula.
	
	As in {\cite{Ma-At}} we let a restricted $LTL$\textit{-step formula }be an $%
	LTL\left( K,AP\right) $-formula of the form $\underset{1\leq i\leq n}{%
		\bigvee }\left( k_{i}\wedge \varphi _{i}\right) $ with $k_{i}\in K\backslash
	\left\{ \mathbf{0},\mathbf{1}\right\} $, $\varphi _{i}\in bLTL\left(
	K,AP\right) $ for all $1\leq i\leq n.$ We denote by $r$-$stLTL\left(
	K,AP\right) $ the class of all restricted $LTL$-step formulas over $AP$ and $%
	K.$ We now let $L_{k}=\left\{ k^{\prime }\in K\mid k^{\prime }\geq k\right\}
	$ for every $k\in K.$ For a given $k\in K\backslash \left\{ \mathbf{0,1}%
	\right\} ,$ we let a $k$-$LTL$\textit{-step formula }be an $LTL\left(
	K,AP\right) $-formula of the form $\underset{1\leq i\leq n}{\bigvee }\left(
	k_{i}\wedge \varphi _{i}\right) $ with $k_{i}\in L_{k}\backslash \left\{
	\mathbf{0},\mathbf{1}\right\} $ and $\varphi _{i}\in sbLTL\left( K,AP\right)
	$ for all $1\leq i\leq n.$ We denote by $k$-$stLTL\left( K,AP\right) $ the
	class of  all $k$-$LTL$-step formulas over $AP$ and $K.$
	\begin{example}
		Let $AP=\left\{ a,b\right\} ,$ $\varphi =\square \left( \left( a\wedge
		2\right) \vee \left( b\wedge 3\right) \right) \in LTL\left( K_{3},AP\right) .
		$ Then, for every $w\in \left(\mathcal{P}\left( AP\right)\right) ^{\omega }$ the following hold:
		$\left( \left\Vert \varphi \right\Vert ,w\right) =-\infty ,$ if for some
		position $i\geq 0,$ both $a,b$ do not belong to $w\left( i\right) .$
		Moreover, $\left( \left\Vert \varphi \right\Vert ,w\right) =3,$ if for every
		position $i\geq 0,$ at least one of $a$, and $b$ belongs to $w\left( i\right)$, and
		for at least one $i\geq 0,$ $b$ belongs to $w\left( i\right) .$ Finally, $%
		\left( \left\Vert \varphi \right\Vert ,w\right) =2,$ if for every position $%
		i\geq 0,$ $a$ belongs to $w\left( i\right) ,$ but $b$ does not.
	\end{example}
	
	\begin{example}
		Let $AP=\left\{ a,b\right\} ,$ $\varphi =\square \left( \left( a\wedge
		2\right) \vee \left( b\wedge 3\right) \right) \in LTL\left( K_{2},AP\right) .
		$ Then, for every $w\in \left(\mathcal{P} \left( AP\right)\right)^{\omega }$ the following hold:
		$\left( \left\Vert \varphi \right\Vert ,w\right) =-\infty ,$ if for some
		position $i\geq 0,$ both $a,b$ do not belong to $w\left( i\right) .$
		Moreover, $\left( \left\Vert \varphi \right\Vert ,w\right) =3,$ if for every
		position $i\geq 0,$ $a \in w\left( i\right) ,$ or $b\in w\left( i\right)$, and
		there exist infinitely many $i\geq 0,$ such that $b\in w\left( i\right)$. Finally, $\left( \left\Vert \varphi \right\Vert ,w\right) =2,$
		if for every position $i\geq 0,$ $a\in w\left( i\right),$ or $b\in w\left( i\right)$, and there exist only finitely many $i\geq 0,$ such that $b\in w\left( i\right)$.
	\end{example}

	\begin{remark}
		\label{remark-k-step} Let $k\in K\backslash \left\{ \mathbf{0,1}\right\}$, and
		$\varphi\in k$-$stLTL\left(K,AP\right).$ Then, $\varphi=\\\underset{1\leq
			i\leq n}{\bigvee}\left(k_{i}\land\varphi_{i}\right)$ where $k_{i}\in L_{k}\backslash\left\{\mathbf{0},\mathbf{1}\right\},$
		and $\varphi_{i}\in sbLTL\left(K,AP\right)$ for every $1\leq i\leq n$. Then,
		\begin{equation*}
			Im\left(\parallel\varphi\parallel\right)\subseteq\left\{\mathbf{0}\right\}
			\cup\left\{ \underset{1\leq j\leq l,i_{j}\in\left\{ 1,\ldots,n\right\} }{\sum%
			}k_{i_{j}} \mathrel{\Big|} l\in\left\{ 1,\ldots,n\right\} \right\} .
		\end{equation*}
		
		Clearly, $Im\left(\parallel\varphi\parallel\right)$ is finite. Moreover, it
		holds $k_{i}\geq k$ for every $i\in\left\{ 1,\ldots,n\right\} ,$ i.e., $%
		k_{i}=k+k_{i}$ for every $i\in\left\{ 1,\ldots,n\right\} $. Thus, for every $%
		l\in\left\{ 1,\ldots,n\right\} ,$
		\begin{gather*}
			\underset{1\leq j\leq l,i_{j}\in\left\{
				1,\ldots,n\right\} }{\sum}k_{i_{j}}=k_{i_{1}}+\underset{2\leq j\leq
				l,i_{j}\in\left\{ 1,\ldots,n\right\} }{\sum}k_{i_{j}}\\
			=k+k_{i_{1}}+\underset{%
				2\leq j\leq l,i_{j}\in\left\{ 1,\ldots,n\right\} }{\sum}k_{i_{j}}\\
			=k+\underset%
			{1\leq j\leq l,i_{j}\in\left\{ 1,\ldots,n\right\} }{\sum}k_{i_{j}}\geq k,
		\end{gather*}
		where the first and last equality hold by the completeness axioms of the
		monoid, and the second equality and the last inequality hold by the
		definition of the natural order induced by idempotency. Clearly, for all $%
		w\in \left( \mathcal{P}\left( AP\right) \right) ^{\omega }$ it holds $\left(
		\left\Vert \varphi \right\Vert ,w\right) <k\Leftrightarrow \left( \left\Vert
		\varphi \right\Vert ,w\right) =\mathbf{0}$.
	\end{remark}
	
	\bigskip We recall now from {\cite{Ma-At}} the definition of totally
	restricted $U$-nesting $LTL$-formulas over $AP$ and $K.$ Due to the restricted version of distributivity of $Val^{\omega}$ over finite sums that holds in TGP-$\omega$-valuation monoids, additional syntactic restrictions must be imposed to the formulas of the fragment of totally
	restricted $U$-nesting $LTL$-formulas in order to establish the effective translation of formulas of the weighted $LTL$ over TGP-$\omega$-valuation monoids to wBa. To that end, we introduce
	the fragment of $\vee $-\textit{totally} \textit{restricted} $U$-\textit{nesting} $%
	LTL$-formulas over $AP$ and $K,$ which is defined inductively exactly as the fragment of totally restricted $U$-nesting $LTL$-formulas with the only difference that we require a stronger syntactic restriction in the inductive step that refers to the $\vee$-operator.
	
	\begin{definition}
		The fragment of totally \textit{restricted} $U$-\textit{nesting} $LTL$%
		-formulas over $AP$ and $K$, denoted by $t$-$RULTL\left( K,AP\right) $, is
		the least class of formulas in $LTL\left( K,AP\right) $ which is defined
		inductively in the following way.
		
		\begin{itemize}
			\item[$1. $] $k\in t$-$RULTL\left( K,AP\right) $ for every $k\in K,$
			
			\item[$2. $] $bLTL\left( K,AP\right) \subseteq t$-$RULTL\left(
			K,AP\right) ,$
			
			\item[$3.$] $r$-$stLTL\left( K,AP\right) \subseteq t$-$RULTL\left(
			K,AP\right) ,$
			
			\item[$4.$] If $\varphi \in t$-$RULTL\left( K,AP\right) ,$ then $%
			\bigcirc \varphi \in t$-$RULTL\left( K,AP\right) ,$
			
			\item[$5.$] If $\varphi ,\psi \in t$-$RULTL\left( K,AP\right) ,$ then $%
			\varphi \vee \psi \in t$-$RULTL\left( K,AP\right) ,$
			
			\item[$6.$] If $\varphi \in bLTL\left( K,AP\right) ,$ and $\psi \in r$-$%
			stLTL\left( K,AP\right) ,$
			
			or $\psi =\xi U\lambda ,$ or $\psi =\square \xi $ with $\xi ,\lambda \in r$-$%
			stLTL\left( K,AP\right) ,$
			
			then $\varphi \wedge \psi ,\psi \wedge \varphi \in t$-$RULTL\left(
			K,AP\right) ,$
			
			\item[$7.$] If $\varphi ,\psi \in r$-$stLTL\left( K,AP\right) ,$ then $%
			\varphi U\psi \in t$-$RULTL\left( K,AP\right) ,$
			
			\item[$8.$] If $\varphi \in r$-$stLTL\left( K,AP\right) ,$ then $\square
			\varphi \in t$-$RULTL\left( K,AP\right) .$
		\end{itemize}
		
		The fragment of $\vee $-totally \textit{restricted} $U$-\textit{nesting} $%
		LTL $-formulas over $AP$ and $K$, denoted by $\vee $-$t$-$RULTL\left(
		K,AP\right) $, is the least class of formulas in $LTL\left( K,AP\right) $
		which is defined inductively exactly as the fragment of totally \textit{restricted} $U$-\textit{nesting} $LTL$%
		-formulas over $AP$ and $K$ with the only difference that the inductive step that refers to the $\vee$-operator (condition 5) is now defined by the following statement
		
		\begin{itemize}
			\item[$\cdot $] If $\varphi ,\psi \in r$-$stLTL\left( K,AP\right) ,$ or $%
			\varphi =\lambda U\xi ,\psi =\square \xi $ with $\lambda ,\xi \in r$-$%
			stLTL\left( K,AP\right) ,$ then $\varphi \vee \psi ,\psi \vee \varphi \in \vee $-$t$-$RULTL\left(
			K,AP\right) .$
		\end{itemize}
	\end{definition}
	
	\begin{example}
		Let $\varphi =\square \left( a\wedge
		2\right) \wedge \left( trueUc\right) \in t$-$RULTL\left( K_{1},AP\right)$ where $AP=\left\{ a,b,c\right\} .$
		Then, for every $w\in \left( \mathcal{P}\left(AP\right)\right) ^{\omega }$ the following hold: $%
		\left( \left\Vert \varphi \right\Vert ,w\right)$ $ =2,$ if $\exists i\geq 0,$
		such that $c\in w\left( i\right) ,$ and $\forall j\geq 0,$ $a\in w\left(
		i\right) .$ Otherwise,  $\left( \left\Vert \varphi \right\Vert ,w\right)
		=-\infty .$
	\end{example}
	
	\begin{theorem}[\cite{Ma-At}]
		\label{From formulas to automata copy(1)} Let $K$ be an
		idempotent ordered  GP-$\omega $-valuation monoid, and $%
		\varphi \in t$-$RULTL\left( K,AP\right) .$ Then, we can effectively
		construct a wBa over $\mathcal{P}\left( AP\right) $ and $K$ recognizing $%
		\left\Vert \varphi \right\Vert .$
	\end{theorem}
	
	Following the proofs of \cite{Ma-At} we generalize the above theorem for
	idempotent ordered TGP-$\omega $-valuation monoids
	and $\vee $-totally restricted $U$-nesting $LTL$-formulas. A sketch of the
	proof can be found in the Appendix.
	
	\begin{theorem}
		\label{From formulas to automata-tr copy(1)}Let $K$ be an idempotent ordered
		TGP-$\omega $-valuation monoid, and $\varphi \in
		\vee $-$t$-$RULTL\left( K,AP\right) .$ Then, we can effectively construct a
		wBa over $\mathcal{P}\left( AP\right) $ and $K$ recognizing $\left\Vert
		\varphi \right\Vert .$
	\end{theorem}
	
	We aim to introduce syntactic fragments of the weighted $LTL$ over $K$,
	and $AP$ with the property that for specific $k\in K\backslash \left\{
	\mathbf{0},\mathbf{1}\right\}$ the semantics of the formulas in the
	fragments are $k$-safe infinitary series over $K$. We also aim that the
	formulas in these fragments preserve the property of totally restricted
	$U$-nesting $LTL$ formulas, and $\vee$-totally restricted $U$-nesting $LTL$
	formulas that can be effectively translated to an equivalent wBa. To motivate the
	restrictions that we will impose to these fragments we present the following
	examples of $\vee$-$t$-$RULTL\left( K,AP\right)$-formulas whose semantics
	are not $k$-safe for specific $k\in K\backslash \left\{ \mathbf{0},\mathbf{1}%
	\right\} $.
	
	\begin{example}
		We let $AP=\left\{a,b\right\}$. Then, $\varphi=\left(3\wedge a\right)
		U\left(3\wedge b\right) \in \vee $-$t$-$RULTL\left( K_{2},AP\right)$. It
		holds
		\begin{equation*}
			\left(\parallel\varphi\parallel,w\right)=%
			\begin{cases}
				\begin{array}{cl}
					3 & {\textstyle if\:}w\in \left\{ \left\{ a\right\},\left\{ a,b\right\} \right\} ^{*}\left\{
					\left\{b\right\},\left\{ a,b\right\} \right\}\left(\mathcal{P}\left(AP\right)\right)^{\omega}\\
					-\infty & otherwise%
				\end{array}%
			\end{cases}%
		\end{equation*}
		
		Let $w=\left\{ a\right\} ^{\omega}$, then for every $i>0,$ $%
		\left(\parallel\varphi\parallel,w_{<i} \left\{ b\right\}
		^{\omega}\right)=3\geq3.$ However, $\left(\parallel\varphi\parallel,\left\{
		a\right\} ^{\omega}\right)=-\infty<3,$ which implies that $%
		\parallel\varphi\parallel$ is not $3$-safe.
		
		Let now $\psi=\square\left(2\wedge a\right) \vee \left(\left(2\land
		a\right)U\left(3\land b\right)\right)\in\lor$-$t$-$RULTL\left(K_{2},AP%
		\right) $. It~holds
		\begin{equation*}
			\left(\parallel\psi\parallel,w\right)=%
			\begin{cases}
				\begin{array}{cl}
					3 & {\textstyle if\:}w\in \left\{ \left\{ a\right\},\left\{ a,b\right\} \right\} ^{*}\left\{\left\{
					b\right\},\left\{ a,b\right\} \right\}\left(\mathcal{P}\left(AP\right)\right)^{\omega} \\
					2 & {\textstyle if\:}w=\left\{ a\right\} ^{\omega} \\
					-\infty & otherwise%
				\end{array}%
			\end{cases}%
		\end{equation*}
		
		Let $w=\left\{ a\right\} ^{\omega}$, then for every $i>0,$ $%
		\left(\parallel\psi\parallel,w_{<i}\left\{ b\right\}
		^{\omega}\right)=3\geq3. $ However, $\left(\parallel\psi\parallel,\left\{
		a\right\} ^{\omega}\right)=2<3,$ which implies that $\parallel\psi\parallel$
		is not $3$-safe.
	\end{example}
	
	\begin{definition}
		Let $k\in K\backslash \left\{ \mathbf{0},\mathbf{1}\right\} .$ The fragment
		of $k$-safe totally restricted $U$-\textit{nesting} $LTL$-formulas over $AP$
		and $K$, denoted by $k$-$t$-$RULTL\left( K,AP\right) $, is the least class
		of formulas in $LTL\left( K,AP\right)$ which is defined inductively in the
		following way.
		
		\begin{itemize}
			\item[$1.$] $L_{k}\subseteq k$-$t$-$RULTL\left( K,AP\right) ,$
			
			\item[$2.$] $sbLTL\left( K,AP\right) \subseteq k$-$t$-$RULTL\left(
			K,AP\right) ,$
			
			\item[$3.$] $k$-$stLTL\left( K,AP\right) \subseteq k$-$t$-$RULTL\left(
			K,AP\right) ,$
			
			\item[$4.$] If $\varphi \in k$-$t$-$RULTL\left( K,AP\right) ,$ then $%
			\bigcirc \varphi \in k$-$t$-$RULTL\left( K,AP\right) ,$
			
			\item[$5.$] If $\varphi ,\psi \in k$-$t$-$RULTL\left( K,AP\right) ,$
			then $\varphi \vee \psi \in k$-$t$-$RULTL\left( K,AP\right) ,$
			
			\item[$6.$] If $\varphi \in sbLTL\left( K,AP\right) ,$ and $\psi \in k$-$%
			stLTL\left( K,AP\right) $
			
			or $\psi =\xi \widetilde{U}\lambda ,$ or $\psi =\square \xi $ with $\xi
			,\lambda \in k$-$stLTL\left( K,AP\right) ,$
			
			then $\varphi \wedge \psi ,\psi \wedge \varphi \in k$-$t$-$RULTL\left(
			K,AP\right) ,$
			
			\item[$7.$] If $\varphi ,\psi \in k$-$stLTL\left( K,AP\right) ,$ then $%
			\varphi \widetilde{U}\psi \in k$-$t$-$RULTL\left( K,AP\right) ,$
			
			\item[$8.$] If $\varphi \in k$-$stLTL\left( K,AP\right) ,$ then $\square
			\varphi \in k$-$t$-$RULTL\left( K,AP\right) .$
		\end{itemize}
		The fragment of $k$-safe $\vee$-totally \textit{restricted} $U$-\textit{nesting} $LTL$%
		-formulas over $AP$ and $K$, which is denoted by $k$-$\vee $-$t$-$RULTL\left(
		K,AP\right)$, is the least class of formulas in $LTL\left( K,AP\right)$ which is defined inductively exactly as the fragment of $k$-safe totally \textit{restricted} $U$-\textit{nesting} $LTL$-formulas over $AP$ and~$K$ with the only difference that the inductive step that refers to the $\vee$-operator (condition~$5$) is now defined by the following statement
		\begin{itemize}
			\item[$\cdot $] If $\varphi ,\psi \in k$-$stLTL\left( K,AP\right) ,$ then $%
			\varphi \vee \psi \in k$-$\vee $-$t$-$RULTL\left( K,AP\right).$
		\end{itemize}
	\end{definition}
	
	\begin{remark}
		\label{Remark 2} We note that $k$-$\vee$-$t $-$RULTL\left(
		K,AP\right) $$\subseteq$ $k$-$t$-$RULTL\left(
		K,AP\right) $. Moreover, for all $k\in K\backslash \left\{ \mathbf{0},\mathbf{1}%
		\right\} ,$ $\mathbf{0\notin }L_{k},$ and for every $k$-$t$-$RULTL\left(
		K,AP\right) $-formula (resp. for every $k$-$\vee $-$t $-$RULTL\left(
		K,AP\right) $-formula), we can construct, using the definition of $%
		\widetilde{U},$ an equivalent $t$-$RULTL\left( K,AP\right) $-formula (resp.
		an equivalent $\vee $-$t$-$RULTL\left( K,AP\right) $-formula).
	\end{remark}
	Next, we will prove that the semantics of  $k$-$t$-$RULTL\left(
	K,AP\right) $-formulas are $k$-safe infinitary series over $K$ and
	$\mathcal{P}\left( AP\right) $. We shall need the following lemmas.
	\begin{lemma}
		\label{boolean Lemma}If $k\in K\backslash \left\{ \mathbf{0},\mathbf{1}%
		\right\} ,$ and $\varphi \in sbLTL\left( K,AP\right) ,\ $ then $\left\Vert
		\varphi \right\Vert $ is $k$-safe.
	\end{lemma}
	
	\begin{proof}
		Let $w\in \left( \mathcal{P}\left( AP\right) \right) ^{\omega }$ such that $%
		\forall i>0,\exists u\in \left( \mathcal{P}\left( AP\right) \right) ^{\omega
		},$ with $\left( \left\Vert \varphi \right\Vert ,w_{<i} u\right) \geq k.$ By
		Property \ref{Property 5} of TGP-$\omega$-valuation monoids, the facts that $k\neq\mathbf{0}$, and $Im\left(\left\Vert
		\varphi \right\Vert\right)=\left\{\mathbf{0},\mathbf{1}\right\} $, we get
		that $\forall i>0,\exists u\in \left( \mathcal{P}\left( AP\right) \right)
		^{\omega },$ such that $\left( \left\Vert \varphi \right\Vert
		,w_{<i} u\right) =\mathbf{1,}$ i.e., $\forall i>0,\exists u\in \left(
		\mathcal{P}\left( AP\right) \right) ^{\omega },$ with $w_{<i} u\mathbf{%
			\models }\varphi .$ Then, since $\mathcal{L}\left(\phi\right)$ is a safety language, $w\models \varphi ,$ which implies that $\left(
		\left\Vert \varphi \right\Vert ,w\right) =\mathbf{1\geq }k$ as desired.
	\end{proof}
	
	Let now $k\in K\backslash \left\{ \mathbf{0},\mathbf{1}\right\} $, and $%
	\varphi =\underset{1\leq i\leq n}{\bigvee }\left( k_{i}\wedge \varphi
	_{i}\right) \in k$-$stLTL\left( K,AP\right) .$ We assign to $\varphi$ the formula $%
	\varphi _{b}:=$ $\underset{1\leq i\leq n}{\bigvee }\varphi _{i}\in
	sbLTL\left( K,AP\right) .$ In addition, if $\varphi =\psi \widetilde{U}\xi $
	$\left( \text{resp. }\varphi =\square \psi, \varphi=\psi U\xi\right),$ with $%
	\psi ,\xi \in k$-$stLTL\left( K,AP\right) $, we assign to $\varphi$, the formula $\varphi _{b}:=\psi _{b}\widetilde{U}\xi _{b}\in
	sbLTL\left( K,AP\right) $ (resp. $\varphi _{b}:=\square \psi _{b}\in
	sbLTL\left( K,AP\right),$ $ \varphi_{b}:=\psi_{b}U\xi_{b}\in bLTL\left(
	K,AP\right)$ ). In fact, the operator $\varphi_{b}$ assigns to $\varphi$ a boolean formula whose language coincides with $L_{\geq k}\left(\left\Vert \varphi \right\Vert \right)$. With Lemmas \ref{k-LTLstep_boithitiko}, \ref{until-always}, that follow, we prove the latter property.
	
	\begin{lemma}
		\label{k-LTLstep_boithitiko}Let $k\in K\backslash \left\{ \mathbf{0},\mathbf{%
			1}\right\} ,$ and $\varphi \in k$-$stLTL\left( K,AP\right) $. Then, the
		following is true: For every $w\in \left( \mathcal{P}\left( AP\right)
		\right) ^{\omega },$ $\left( \left\Vert \varphi \right\Vert ,w\right) \geq k$
		iff $w\models \varphi _{b}.$
	\end{lemma}
	
	\begin{proof}
		Let $\varphi =\underset{1\leq i\leq n}{\bigvee }\left( k_{i}\wedge \varphi
		_{i}\right) \in k$-$stLTL\left( K,AP\right) .$ Then, for every $w\in \left(
		\mathcal{P}\left( AP\right) \right) ^{\omega }$ it holds:
		
		\begin{gather*}
			\left( \left\Vert
			\varphi \right\Vert ,w\right) \geq k \\
			\text{iff}
		\end{gather*}
		
		\begin{gather*}
			\left( \left\Vert \underset{%
				1\leq i\leq n}{\bigvee }\left( k_{i}\wedge \varphi _{i}\right) \right\Vert
			,w\right) \geq k
			\\
			\text{iff}\\
			\left( \left\Vert \underset{%
				1\leq i\leq n}{\bigvee }\left( k_{i}\wedge \varphi _{i}\right) \right\Vert
			,w\right) \neq \mathbf{0}\\
			\text{iff} \\
			\exists i\in \left\{ 1,\ldots ,n\right\}  \text{ s.t. } %
			\left( \left\Vert \varphi _{i}\right\Vert ,w\right) =\mathbf{1}\\
			\ \text{iff} \\
			w\models \underset{1\leq i\leq n}{\bigvee }\varphi _{i}=\varphi _{b},
		\end{gather*}
		where the second equivalence is obtained by Remark \ref{remark-k-step}. The straight implication of the third equivalence is obtained by assuming the contrary and by the semantics of $LTL\left(K,AP\right)$-formulas. The inverse implication of the third equivalence holds by the semantics of $LTL\left(K,AP\right)$-formulas and the fact that $K$ is zero-sum free. The last equivalence is obtained by the coincidence of the semantics of $sbLTL\left(K,AP\right)$-formulas with the corresponding $sLTL\left(AP\right)$-formulas.
	\end{proof}
	
	\begin{lemma}
		\label{until-always}Let $\varphi =\psi \widetilde{U}\xi $, or $\varphi
		=\square \psi ,$ or $\varphi =\psi U\xi $ with $\psi ,\xi \in k$-$%
		stLTL\left( K,AP\right) $. Then, the following is true: For every $w\in
		\left( \mathcal{P}\left( AP\right) \right) ^{\omega }$, $\left( \left\Vert
		\varphi \right\Vert ,w\right) \geq k$ iff $w\models \varphi _{b}.$
	\end{lemma}
	
	\begin{proof}
		Let $\varphi=\square \psi$, with $\psi \in k$-$stLTL\left( K,AP\right) $.
		Then, the following equivalences are true:
		
		\begin{gather*}
			\left( \left\Vert \varphi \right\Vert ,w\right)=\left( \left\Vert \square
			\psi \right\Vert ,w\right)=Val^{\omega }\left( \left( \left\Vert \psi
			\right\Vert ,w_{\geq j}\right) \right) _{j\geq 0} \geq k \\
			\text{iff}
			\left( \left\Vert \psi \right\Vert ,w_{\geq j}\right) \geq k,\forall j\geq 0
			\\
			\text{iff} \\
			w_{\geq j}\models \psi _{b},\forall j\geq 0 \\
			\ \text{iff} \\
			w\models \square \psi_{b}=\varphi _{b}.
		\end{gather*}
		
		To derive the straight implication of the first equivalence we assume the
		contrary. Then, since $K$ is ordered, $\exists j\geq0$ with $\left(
		\left\Vert \psi \right\Vert ,w_{\geq j}\right) < k $, which by Remark \ref%
		{remark-k-step} implies that $\left( \left\Vert \psi \right\Vert ,w_{\geq
			j}\right)=\mathbf{0}$. Thus, by definition of $Val^{\omega}$, we get $%
		Val^{\omega }\left( \left( \left\Vert \psi \right\Vert ,w_{\geq j}\right)
		\right)_{j\geq0} =\mathbf{0}$, which is a contradiction. The inverse
		implication of the first equivalence is concluded by Property \ref{Property 6}. The second equivalence is concluded by Lemma \ref{k-LTLstep_boithitiko}, and the last one by the semantics of $LTL\left(AP\right)$-formulas.
		
		Let $\varphi =\psi U\xi $ with $\psi ,\xi \in k$-$stLTL\left( K,AP\right) ,$
		and $w\in \left( \mathcal{P}\left( AP\right) \right) ^{\omega }.$ Then, the
		following equivalences are true:%
		\begin{gather*}
			\left( \left\Vert \varphi \right\Vert ,w\right)=\left( \left\Vert \psi U\xi
			\right\Vert ,w\right)\\=\underset{h\geq 0}{\sum }Val^{\omega }\left( \left(
			\left\Vert \psi \right\Vert ,w_{\geq 0}\right) ,\ldots ,\left( \left\Vert
			\psi \right\Vert ,w_{\geq h-1}\right) ,\left( \left\Vert \xi \right\Vert
			,w_{\geq h}\right) ,\mathbf{1,1,\ldots }\right) \geq k \\
			\text{iff} \\
			\exists h\geq 0\text{ such that }\left( \left\Vert \xi \right\Vert ,w_{\geq
				h}\right) \geq k\text{ and }\left( \left\Vert \psi \right\Vert ,w_{\geq
				l}\right) \geq k,0\leq l\leq h-1 \\
			\text{iff} \\
			\exists h\geq 0\text{ such that }w_{\geq h}\models \xi _{b}\text{ and }%
			w_{\geq l}\models \psi _{b},0\leq l\leq h-1 \\
			\text{iff} \\
			w\models \psi _{b}U\xi _{b}=\varphi _{b}.
		\end{gather*}
		In order to prove the straight implication of first equivalence we assume the
		contrary. Then, since $K$ is ordered, we get that for all $h\geq0$ it holds
		\begin{gather*}
			\left( \left\Vert
			\xi \right\Vert ,w_{\geq h}\right) < k,\text{ or }\left( \left\Vert \psi
			\right\Vert ,w_{\geq l}\right) <k \text{ for some } l\in
			\left\{0,\ldots,h-1\right\}.
		\end{gather*}
		This, by Remark \ref{remark-k-step}, implies
		that for all $h\geq0$ it holds
		\begin{gather*}
			\left( \left\Vert \xi \right\Vert ,w_{\geq h}\right)=%
			\mathbf{0},\text{ or }\left( \left\Vert \psi \right\Vert ,w_{\geq l}\right)=%
			\mathbf{0} \text{ for some } l\in \left\{0,\ldots,h-1\right\}.
		\end{gather*}
		
		\noindent Thus,
		\begin{gather*}
			\forall
			h\geq0, Val^{\omega }\left( \left( \left\Vert \psi \right\Vert ,w_{\geq
				0}\right) ,\ldots ,\left( \left\Vert \psi \right\Vert ,w_{\geq h-1}\right)
			,\left( \left\Vert \xi \right\Vert ,w_{\geq h}\right) ,\mathbf{%
				1,1,\ldots }\right)=\mathbf{0},
		\end{gather*}
		and so by the generalization of \cite[Lemma 1(ii)]{Ma-At} for idempotent TGP-$\omega$-valuation monoids it holds:
		\begin{gather*}
			\underset{h\geq0}{{\sum }}Val^{\omega }\left( \left(
			\left\Vert \psi \right\Vert ,w_{\geq 0}\right) ,\ldots ,\left( \left\Vert
			\psi \right\Vert ,w_{\geq h-1}\right) ,\left( \left\Vert \xi \right\Vert
			,w_{\geq h}\right) ,\mathbf{1,1,\ldots }\right)=\mathbf{0}
		\end{gather*}
		which is a contradiction. The inverse implication of the first equivalence is derived by the completeness axioms
		of the monoid and Properties \ref{Property 5}, \ref{Property 6}. As in the previous case, the second equivalence is concluded by Lemma \ref%
		{k-LTLstep_boithitiko}, and the last one by the semantics of $LTL\left(AP\right)$-formulas.
		
		Finally, if $\varphi =\psi \widetilde{U} \xi ,$ $\psi ,\xi \in k$-$%
		stLTL\left( K,AP\right) ,$ and $w\in \left( \mathcal{P}\left( AP\right)
		\right) ^{\omega },$ the following equivalences are true:
		\begin{gather*}
			\left( \left\Vert \varphi \right\Vert ,w\right)=\left( \left\Vert \psi
			\widetilde{U} \xi \right\Vert ,w\right)= \left( \left\Vert \square \psi
			\right\Vert ,w\right)+ \left( \left\Vert \psi U\xi \right\Vert ,w\right)\geq
			k \\
			\text{iff} \\
			\left( \left\Vert \square \psi \right\Vert ,w\right)\geq k \text{ or }
			\left( \left\Vert \psi U\xi \right\Vert ,w\right)\geq k \\
			\text{iff} \\
			w\models \square\psi_{b} \text{ or } w\models \psi_{b}U\xi_{b} \\
			\text{iff} \\
			w\models \psi _{b}\widetilde{U}\xi _{b}=\varphi _{b}.
		\end{gather*}
		The straight implication of the first equivalence is derived by assuming the
		contrary and Lemma \ref{Lemma-strict-inequality_finite}, while the inverse
		implication of the first equivalence is derived by the definition of the
		natural order induced by idempotency, the associativity, and the commutativity of the monoid.
		The third equivalence is derived by the previous cases, and the last one by the semantics of $LTL\left(AP\right)$-formulas.
	\end{proof}
	
	\begin{lemma}
		\label{k-LTLstep}Let $k\in K\backslash \left\{ \mathbf{0},\mathbf{1}\right\}
		,$ and $\varphi \in k$-$stLTL\left( K,AP\right) .$ Then, $\varphi $ is $k$%
		-safe.
	\end{lemma}
	
	\begin{proof}
		Let $\varphi =\underset{1\leq i\leq n}{\bigvee }\left( k_{i}\wedge \varphi
		_{i}\right) \in k$-$stLTL\left( K,AP\right) ,$ and $w\in \left( \mathcal{P}%
		\left( AP\right) \right) ^{\omega }$ such that for every $j>0,\exists
		u_{j}\in \left( \mathcal{P}\left( AP\right) \right) ^{\omega }$ with $%
		\left( \left\Vert \varphi \right\Vert ,w_{<j} u_{j}\right) \geq k.$ Then, by
		Lemma \ref{k-LTLstep_boithitiko} we get that $\forall j>0,$ $\exists
		u_{j}\in \left( \mathcal{P}\left( AP\right) \right) ^{\omega }$ with $%
		w_{<j} u_{j}\models \varphi _{b}=\underset{1\leq i\leq n}{\bigvee }\varphi
		_{i}$. Then, for at least one $i\in \left\{ 1,\ldots ,n\right\} ,$ there
		exist infinitely many $j>0,$ such that $w_{<j} u_{j}\models \varphi _{i}\in
		sbLTL\left( K,AP\right) $ for some $u_{j}\in \left( \mathcal{P}\left(
		AP\right) \right) ^{\omega }.$ Let now $j_{1}<j_{2}<j_{3}<\ldots $ be the
		sequence of positions with $w_{<j_{l}} u_{j_{l}}\models \varphi _{i}$ for
		some $u_{j_{l}}\in \left( \mathcal{P}\left( AP\right) \right) ^{\omega },$ $%
		l\geq 1.$ Let $j<j_{h}$ for some $h\geq 1,$ then $w_{<j_{h}}=w_{<j} w\left(
		j\right) \ldots w\left( j_{h}-1\right) .$ Thus, for $u_{j}=w\left( j\right)
		\ldots w\left( j_{h}-1\right) u_{j_{h}}$ we get $w_{<j} u_{j}\models \varphi
		_{i}$ as desired. Since for all $j>0,$ there exist $h> 0$ such that $j\leq
		j_{h},$ we conclude that for all $j>0,$ there exists $u\in \left( \mathcal{P}%
		\left( AP\right) \right) ^{\omega }$ with $w_{<j} u\models \varphi _{i}$.
		Since $\varphi _{i}\in sLTL\left( AP\right) ,$ we conclude that $w\models
		\varphi _{i}.$ This implies that $\left( \left\Vert \varphi _{i}\right\Vert
		,w\right) =\mathbf{1,}$ and hence $\left( \left\Vert \varphi \right\Vert
		,w\right) =\left( \left\Vert \underset{1\leq i\leq n}{\bigvee }\left(
		k_{i}\wedge \varphi _{i}\right) \right\Vert ,w\right) \geq k$ as desired.
	\end{proof}
	
	We are ready now to prove that the semantics of $k$-$t$-$%
	RULTL\left( K,AP\right) $-formulas are $k$-safe infinitary series.
	
	\begin{theorem}
		\label{Theorem_safety_k_totally restricted formulas} Let $k\in K\backslash
		\left\{ \mathbf{0},\mathbf{1}\right\} ,$ and $\varphi \in k$-$t$-$%
		RULTL\left( K,AP\right) .$ Then, $\left\Vert \varphi \right\Vert \in
		K\left\langle \left\langle \left( \mathcal{P}\left( AP\right) \right)
		^{\omega }\right\rangle \right\rangle $ is $k$-safe.
	\end{theorem}
	
	\begin{proof}
		We prove the Lemma's claim by induction on the structure of $k$-$t$-$%
		RULTL\left( K,AP\right)$-formulas. For $\varphi =k^{\prime }\in L_{k}$, by Lemma \ref{Lemma k} we get that $%
		\left\Vert k^{\prime }\right\Vert =k^{\prime }$ is $k$-safe. For $%
		\varphi \in sbLTL\left( K,AP\right) $, $\varphi \in k$-$stLTL\left(
		K,AP\right) $ we obtain our claim by the Lemmas \ref{boolean Lemma} and \ref%
		{k-LTLstep} respectively.
		
		Assume that $\xi ,\psi \in k$-$t$-$RULTL\left( K,AP\right) $ such that the
		induction hypothesis holds for $\psi ,\xi .$ Let $\varphi =\xi \vee \psi ,$
		and $w\in \left( \mathcal{P}\left( AP\right) \right) ^{\omega }$ such that $%
		\forall i>0$, $\exists u\in \left( \mathcal{P}\left( AP\right) \right) ^{\omega
		}$ with $\left( \left\Vert \varphi \right\Vert ,w_{<i} u\right) \geq k.$
		Then, $\forall i>0,\exists u\in \left( \mathcal{P}\left( AP\right) \right)
		^{\omega }$ with $\left( \left\Vert \xi \right\Vert ,w_{<i} u\right) +\left(
		\left\Vert \psi \right\Vert ,w_{<i} u\right) \geq k$. By assuming the
		contrary, and using Lemma \ref{Lemma-strict-inequality_finite}, we get that
		there exist infinitely many $i>0$ such that $\left( \left\Vert \psi
		\right\Vert ,w_{<i} u\right) \geq k$ for some $u\in \left( \mathcal{P}\left(
		AP\right) \right) ^{\omega }$, or there exist infinitely many $i>0$ such
		that $\left( \left\Vert \xi \right\Vert ,w_{<i} u\right) \geq k$ for some $%
		u\in \left( \mathcal{P}\left( AP\right) \right) ^{\omega }.$ Assume now that
		the first case is true and let $i_{1}<i_{2}<i_{3}<\ldots $ be the sequence
		of positions with $\left( \left\Vert \psi \right\Vert
		,w_{<i_{j}} u_{i_{j}}\right) \geq k$ for some $u_{i_{j}}\in \left( \mathcal{P}%
		\left( AP\right) \right) ^{\omega },$ $j>0.$ Let $i<i_{h}$ for some $h\geq
		1, $ then $w_{<i_{h}}=w_{<i} w\left( i\right) \ldots w\left( i_{h}-1\right) .$
		Thus, for $u=w\left( i\right) \ldots w\left( i_{h}-1\right) u_{i_{h}}$ we
		get $\left( \left\Vert \psi \right\Vert ,w_{<i} u\right) \geq k$ as desired.
		Since for all $i>0,$ there exist $h>0$ such that $i<i_{h},$ we conclude that
		for all $i>0,$ there exists $u\in \left( \mathcal{P}\left( AP\right) \right)
		^{\omega }$ with $\left( \left\Vert \psi \right\Vert ,w_{<i} u\right) \geq k$%
		. Then, by induction hypothesis $\psi $ is $k$-safe, i.e., $\left(
		\left\Vert \psi \right\Vert ,w\right) \geq k,$ which implies that $\left(
		\left\Vert \psi \right\Vert ,w\right) $=$k+\left(
		\left\Vert \psi \right\Vert ,w\right).$ Thus,
		$\left(
		\left\Vert \varphi \right\Vert ,w\right) =\left( \left\Vert \xi \right\Vert
		,w\right) +\left( \left\Vert \psi \right\Vert ,w\right)$$=k+\left(\left(
		\left\Vert \xi \right\Vert ,w\right)+\left( \left\Vert \psi \right\Vert ,w\right) \right)=k+\left(
		\left\Vert \varphi \right\Vert ,w\right)$, where the second equality is derived by the definition of the natural order induced by idempotency, the associativity, and the commutativity of the monoid. Thus, $\left(
		\left\Vert \varphi \right\Vert ,w\right)\geq k$ as desired.
		Now, if the second case is true, we conclude our claim with the same
		arguments.
		
		Next, we let $\varphi =\bigcirc \xi ,$ and $w\in \left( \mathcal{P}\left(
		AP\right) \right) ^{\omega }$ such that $\forall i>0,\exists u\in \left(
		\mathcal{P}\left( AP\right) \right) ^{\omega }$ with $\left( \left\Vert
		\varphi \right\Vert ,w_{<i} u\right) =\left( \left\Vert \xi \right\Vert
		,\left( w_{<i} u\right) _{\geq 1}\right) \geq k.$ This implies that for every
		$i>0,$ $\exists u\in \left( \mathcal{P}\left( AP\right) \right) ^{\omega }$
		with $\left( \left\Vert \xi \right\Vert ,\left( w_{\geq 1}\right)
		_{<i} u\right) \geq k$ (observe that for all $i>0,$ $\left( w_{\geq 1}\right)
		_{<i} u=\left( w_{<i+1} u\right) _{\geq 1}$)$,$ and since by induction
		hypothesis $\xi $ is $k$-safe, we get $\left( \left\Vert \xi \right\Vert
		,w_{\geq 1}\right) \geq k,$ i.e., $\left( \left\Vert \bigcirc \xi
		\right\Vert ,w\right) \geq k$ as desired.
		
		Let now $\varphi =\psi \wedge \xi $ with $\psi \in sbLTL\left( K,AP\right) $%
		, and $\xi \in k$-$stLTL\left( K,AP\right) ,$ or $\xi =\zeta \widetilde{U}%
		\lambda ,$ or $\xi =\square \zeta $ with $\zeta ,\lambda \in k$-$%
		stLTL\left( K,AP\right) ,$ such that the induction hypothesis holds for $%
		\psi ,\xi .$ Let $w\in \left( \mathcal{P}\left( AP\right) \right) ^{\omega }$
		such that $\forall i\geq 0$, $\exists u\in \left( \mathcal{P}\left( AP\right)
		\right) ^{\omega }$ with $\left( \left\Vert \varphi \right\Vert
		,w_{<i} u\right) \geq k.$ Then, $\forall i\geq 0,\exists u\in \left( \mathcal{%
			P}\left( AP\right) \right) ^{\omega }$ with $\left( \left\Vert \psi
		\right\Vert ,w_{<i} u\right) \cdot \left( \left\Vert \xi \right\Vert
		,w_{<i}u\right) \geq k$,\ and since $k\neq \mathbf{0}$ and $\psi $ is
		boolean, we get that $\forall i\geq 0,\exists u\in \left( \mathcal{P}\left(
		AP\right) \right) ^{\omega }$ with $\left( \left\Vert \psi \right\Vert
		,w_{<i} u\right) =\mathbf{1},$\textbf{\ }and $\left( \left\Vert \psi
		\right\Vert ,w_{<i} u\right) \cdot \left( \left\Vert \xi \right\Vert
		,w_{<i}u\right) =\left( \left\Vert \xi
		\right\Vert ,w_{<i} u\right) \geq k.$ Thus, since by Property \ref{Property 5} we get that $\mathbf{1}$ is the maximal element of $K$, it holds that
		$\forall i\geq 0,\exists u\in \left( \mathcal{P}\left(
		AP\right) \right) ^{\omega }$ with $\left( \left\Vert \psi \right\Vert
		,w_{<i} u\right) \geq k,$ and $\left( \left\Vert \xi
		\right\Vert ,w_{<i} u\right) \geq k.$ Then, by induction hypothesis we get $%
		\left( \left\Vert \psi \right\Vert ,w\right) \geq k\mathbf{,}$
		and $\left( \left\Vert \xi \right\Vert ,w\right) \geq k,$ which, taking into account that $\psi$ is boolean and $k\neq\mathbf{0}$, implies that $\left( \left\Vert \psi \right\Vert ,w\right)=\mathbf{1}$
		and $\left( \left\Vert \xi \right\Vert ,w\right) \geq k,$ and thus
		$\left( \left\Vert \varphi \right\Vert ,w\right)=\left(\left\Vert \psi \right\Vert ,w\right) \cdot \left( \left\Vert \xi \right\Vert ,w\right) \geq k$ as desired.
		
		Let $\varphi =\psi \widetilde{U}\xi $ with $\psi ,\xi \in k$-$stLTL\left(
		K,AP\right) ,$ and $w\in \left( \mathcal{P}\left( AP\right) \right) ^{\omega
		}$ such that $\forall i>0,\exists u\in \left( \mathcal{P}\left( AP\right)
		\right) ^{\omega }$ such that $\left( \left\Vert \varphi
		\right\Vert ,w_{<i}u\right) \geq k.$ Then, by Lemma~\ref{until-always}, we
		get that for all $i>0,$ $\exists u\in \left( \mathcal{P}\left( AP\right)
		\right) ^{\omega }$ with $w_{<i} u\models \varphi_{b}$ $,$
		and since $\varphi _{b}\in sLTL\left( AP\right) ,$ we
		conclude that $w\models \varphi _{b}.$ By Lemma \ref%
		{until-always}, this implies that $\left( \left\Vert \varphi \right\Vert ,w\right) \geq k$ as desired.
		
		We now let $\varphi =\square \psi  \text{ with } \psi \in k$-$stLTL\left(K,AP\right)$, and $w\in \left( \mathcal{P}\left(
		AP\right) \right) ^{\omega }$ such that $\forall i>0,\exists u\in \left(
		\mathcal{P}\left( AP\right) \right) ^{\omega }$ with $\left( \left\Vert
		\varphi \right\Vert ,w_{<i} u\right) \geq k.$ Then, by Lemma~\ref%
		{until-always}, for all $i>0,\exists u\in \left( \mathcal{P}\left( AP\right)
		\right) ^{\omega },$ $w_{<i} u\models \varphi _{b}\in sLTL\left(
		AP\right) .$ Hence, $w\models \varphi _{b},$ and again by Lemma \ref{until-always}
		we get that $\left( \left\Vert \varphi \right\Vert ,w\right) \geq k$ as
		desired.
	\end{proof}
	
	\begin{corollary}
		\label{Corollary k-safety of V-restricted formulas } Let $k\in K\backslash
		\left\{ \mathbf{0},\mathbf{1}\right\} ,$ and $\varphi \in k$-$\vee$-$t$-$%
		RULTL\left( K,AP\right) .$ Then, $\left\Vert \varphi \right\Vert \in
		K\left\langle \left\langle \left( \mathcal{P}\left( AP\right) \right)
		^{\omega }\right\rangle \right\rangle $ is $k$-safe.
	\end{corollary}
	
	\begin{remark}
		\label{remark_non-maximal_fragments} Let $k\in K_{2}\backslash \left\{
		\mathbf{-\infty },\mathbf{\infty }\right\} $. Then, by induction on the
		structure of formulas in $k$-$t$-$RULTL\left( K_{2},AP\right) $, and using
		Remark \ref{remark-k-step}, we can prove that $Im\left( ||\varphi ||\right) $
		is finite, for every $\varphi \in $ $k$-$t$-$RULTL\left( K_{2},AP\right) $.
		In Example \ref{example_infinite_image}, we have presented the $2$-safe
		infinitary series $S\in K_{2}\left\langle \left\langle A^{\omega
		}\right\rangle \right\rangle $, where $A=\left\{ a,b\right\} $, with $%
		Im\left( S\right) $ being an infinite set. We set $AP=\left\{ a,b\right\} $.
		Clearly, $S\in K_{2}\left\langle \left\langle \left( \mathcal{P}\left(
		AP\right) \right) ^{\omega }\right\rangle \right\rangle $. This implies that
		the class of semantics of $2$-$t$-$RULTL\left( K_{2},AP\right) $-formulas is
		a proper subclass of the class of $2$-safe infinitary series over $K_{2}$, and $\mathcal{P%
		}\left( AP\right) $.
	\end{remark}
	
	\begin{corollary}
		Let $AP$ be a finite set of atomic propositions. Then, for the arbitrary
		idempotent ordered TGP-$\omega $-valuation monoid $K,
		$ and $k\in K\backslash \left\{ \mathbf{0,1}\right\} $ the class of
		semantics of $k$-$t$-$RULTL\left( K,AP\right) $-formulas does not coincide\
		with the class of $k$-safe infinitary series over $K,$ and $\mathcal{P}\left( AP\right) .$
	\end{corollary}
	
	\begin{example}
		Let $AP=\left\{ a,b\right\} ,$ $\varphi =\square \left( a\wedge 2\right)
		\in 2$-$\vee $-$t$-$RULTL\left( K_{2},AP\right) .$ For every $w\in \left(
		\mathcal{P}\left( AP\right) \right) ^{\omega },$ it holds $\left( \left\Vert
		\varphi \right\Vert ,w\right) =2$ whenever $a\in w\left( i\right) $ for
		every $i\geq 0,$ and $\left( \left\Vert \varphi \right\Vert ,w\right)
		=-\infty $ otherwise. We verify that $\left\Vert \varphi \right\Vert $ is $2$%
		-safe. More precisely, let $w\in \left( \mathcal{P}\left( AP\right) \right)
		^{\omega }$ such that for every $i>0$ there exists $u\in \left( \mathcal{P}%
		\left( AP\right) \right) ^{\omega }$ with $\left( \left\Vert \varphi
		\right\Vert ,w_{<i} u\right) \geq 2.$ This implies that for every $i>0,$ $%
		a\in w\left( j\right) $ for every $0\leq j\leq i-1,$ which in turn implies
		that $a\in w\left( i\right) $ for every $i\geq 0,$ i.e., $\left( \left\Vert
		\varphi \right\Vert ,w\right) =2.$
	\end{example}
	
	\begin{example}
		Let $\varphi =\left( 3\wedge b\right)
		\widetilde{U}\left( 3\wedge a\right) \in 2$-$\vee $-$t$-$RULTL\left(
		K_{2},AP\right) $ where $AP=\left\{ a,b\right\}.$ It holds $\left( \left\Vert \varphi \right\Vert ,w\right)
		=3$ if $b\in w\left( j\right) $ for all $j\geq 0$, or if there exists an $%
		i\geq 0,$ such that $a\in w\left( i\right) $ and $b\in w\left( j\right) $
		for every $0\leq j\leq i-1.$ Otherwise, $\left( \left\Vert \varphi
		\right\Vert ,w\right) =-\infty .$ Clearly, $\left\Vert \varphi \right\Vert $
		is $2$-safe. More precisely, let $w\in \left( \mathcal{P}\left( AP\right)
		\right) ^{\omega }$ such that for every $i>0$ there exists $u\in \left(
		\mathcal{P}\left( AP\right) \right) ^{\omega }$ with $\left( \left\Vert
		\varphi \right\Vert ,w_{<i} u\right) \geq 2.$ Then, for every $i>0$, $\exists u\in \left( \mathcal{P}\left(
		AP\right) \right) ^{\omega }$ such that at least
		one of the following is true: $\left( \left\Vert \left( 3\wedge
		b\right) U\left( 3\wedge a\right) \right\Vert ,w_{<i} u\right) \geq 2,$ or $%
		\left( \left\Vert\square \left( 3\wedge b \right) \right\Vert ,w_{<i} u\right) \geq 2.$\footnote{%
			This claim is derived by assuming the contrary, and Lemma \ref%
			{Lemma-strict-inequality_finite}.} This implies that at least one of the
		following is true: There exist infinitely many $i>0$ such that $\exists u\in
		\left( \mathcal{P}\left( AP\right) \right) ^{\omega }$ with $\left(
		\left\Vert \left( 3\wedge b\right) U\left( 3\wedge a\right) \right\Vert
		,w_{<i} u\right) \geq 2$, or there exist infinitely many $i>0$ such that $%
		\exists u\in \left( \mathcal{P}\left( AP\right) \right) ^{\omega }$ with $\left(
		\left\Vert \square \left( 3\wedge b \right)\right\Vert ,w_{<i} u\right) \geq 2$. If the first case
		is true we derive that $\forall i>0,$ $\exists u\in \left( \mathcal{P}\left(
		AP\right) \right) ^{\omega },$ such that $\left( \left\Vert \left( 3\wedge
		b\right) U\left( 3\wedge a\right) \right\Vert ,w_{<i} u\right) \geq 2.$%
		\footnote{%
			This claim is derived with the same arguments used in the proof of Theorem %
			\ref{Theorem_safety_k_totally restricted formulas} in the inductive step of
			the disjunction.} Hence, $\forall i>0,\exists u\in \left( \mathcal{P}\left(
		AP\right) \right) ^{\omega },j_{i}\geq 0$ such that $\left( \left\Vert
		3\wedge b\right\Vert ,\left( w_{<i} u\right) _{\geq h}\right) \geq 2\left(
		0\leq h\leq j_{i}-1\right) ,$ and $\left( \left\Vert 3\wedge a\right\Vert
		,\left( w_{<i} u\right) _{\geq j_{i}}\right) \geq 2.$\footnote{%
			This claim is derived by assuming the contrary, and using the definition of $%
			Val^{\omega }$, and Remark \ref{remark-k-step}.} Assume that there exist an $%
		i>0,$ such that $\exists u\in \left( \mathcal{P}\left( AP\right) \right)
		^{\omega },j_{i}\in \left\{ 0,\ldots ,i-1\right\} $ such that $\left(
		\left\Vert 3\wedge b\right\Vert ,\left( w_{<i} u\right) _{\geq h}\right) \geq
		2\left( 0\leq h\leq j_{i}-1\right) $ and $\left( \left\Vert 3\wedge
		a\right\Vert ,\left( w_{<i} u\right) _{\geq j_{i}}\right) \geq 2.$ This
		implies that $\exists j\geq 0$ such that $\left( \left\Vert 3\wedge
		b\right\Vert ,w_{\geq h}\right) \geq 2\left( 0\leq h\leq j-1\right) $ and $%
		\left( \left\Vert 3\wedge a\right\Vert ,w_{\geq j}\right) \geq 2,$ i.e., $%
		\left( \left\Vert \psi U\xi \right\Vert ,w\right) \geq 2,$ which in turn
		implies that $\left( \left\Vert \psi \widetilde{U}\xi \right\Vert ,w\right)
		\geq 2$ as desired. Assume now that for all $i>0,\nexists j_{i}<i$ such $%
		\left( \left\Vert 3\wedge a\right\Vert ,\left( w_{<i} u\right) _{\geq
			j_{i}}\right) \geq 2.$ Then, we get that for all $i>0,\left( \left\Vert
		3\wedge b\right\Vert ,\left( w_{<i} u\right) _{\geq h}\right) \\\geq 2\left(
		0\leq i\leq j-1\right) ,$ which implies that $\left( \left\Vert 3\wedge
		b\right\Vert ,w_{\geq i}\right) \geq 2$ for all $i\geq 0,$ i.e., $\left(
		\left\Vert \square \left( 3\wedge b\right) \right\Vert ,w\right) \geq 2,$
		and thus $\left( \left\Vert \left( 3\wedge b\right) \widetilde{U}\left(
		3\wedge a\right) \right\Vert ,w\right) \geq 2$ as desired.
		
		If the second case is true we conclude that $\left( \left\Vert \left(
		3\wedge b\right) \widetilde{U}\left( 3\wedge a\right) \right\Vert ,w\right)
		\geq 2$ with similar arguments.
	\end{example}
	
	By Remark \ref{Remark 2}, Theorems \ref{From formulas to automata copy(1)}, %
	\ref{From formulas to automata-tr copy(1)}, and Corollary \ref{Corollary
		k-safety of V-restricted formulas } we conclude the following.
	
	\begin{corollary}
		\label{Corollary}
		\begin{enumerate}[label=(\roman*)]
			\item Let $K$ be an idempotent ordered GP-$%
			\omega $-valuation monoid, $k\in K\backslash \left\{ \mathbf{0},\mathbf{1}%
			\right\} $ and $\varphi \in k$-$t$-$RULTL\left( K,AP\right) .$ Then, we can
			effectively construct a wBa over $\mathcal{P}\left( AP\right) $ and $K$
			recognizing $\left\Vert \varphi \right\Vert .$
			
			\item Let $K$ be an idempotent ordered  TGP-$\omega $%
			-valuation monoid, $k\in K\backslash \left\{ \mathbf{0},\mathbf{1}\right\} $
			and $\varphi \in k$-$\vee $-$t$-$RULTL\left( K,AP\right) .$ Then, we can
			effectively construct a wBa over $\mathcal{P}\left( AP\right) $ and $K$
			recognizing $\left\Vert \varphi \right\Vert .$
		\end{enumerate}
	\end{corollary}
	
	\section{A Motivating Example\label{A_motivating_example}}
	
	In this section we present an example on how $k$-safe weighted $LTL$-formulas can be
	used to draw conclusions related to the quantitative behavior of systems.
	More specifically, motivated by the work in \cite{Sm-Op}, we consider a
	weighted transition system that models the routes of a robot that moves
	between states realizing in some of them gathering of data, or control, or
	uploading of data. The weights on the transitions of the system represent
	the amount of energy that is produced by the robot's movement in each
	transition. In fact the weighted transition system that we consider is a
	variation of the corresponding one in \cite{Sm-Op}. We determine quantitative properties that could be related to the preservation
	of the movement of the robot, and for a specific rational number $%
	k,$ we consider a $k$-safe weighted $LTL$ formula over~$K_{2}$ and $\mathcal{P}\left(AP\right)$ that expresses the
	aforementioned properties, where $AP$ is the set of propositions used to label the states of the
	weighted transition system. We then translate the weighted transition system
	to a weighted B\"{u}chi automaton over $K_{2}$ and $\mathcal{P}\left( AP\right) ,
	$ with the following property: each path of the
	automaton, with weight different from $-\infty,$ simulates a
	run of the transition system, and vice-versa. We finally prove that the
	behavior of the automaton is equal to the semantics of the formula,
	verifying in this way that the weighted transition system has the desirable quantitative
	properties.
	
	\medskip
	
	We let $%
	T=\left( Q,q_{0},R,AP,L,we\right) $ where
	
	\begin{itemize}
		\item $Q=\left\{ q_{0},q_{1},q_{2},q_{3}\right\} ,$
		
		\item $AP=\left\{ gather,upload,control\right\} ,$
		
		\item $R=\left\{ \left( q_{0},q_{2}\right) ,\left( q_{2},q_{3}\right)
		,\left( q_{3},q_{0}\right) ,\left( q_{2},q_{1}\right) ,\left(
		q_{1},q_{0}\right) \right\} $,
	\end{itemize}
	
	\begin{itemize}
		\item $L:Q\rightarrow \mathcal{P}\left( AP\right) $ is given by $L\left(
		q_{0}\right) =\emptyset ,$ $L\left( q_{1}\right) =\left\{ upload\right\} $, $%
		L\left( q_{2}\right) =\left\{ gather\right\} ,$ $L\left( q_{3}\right)
		=\left\{ control,upload\right\} $,
		
		\item $we:R\rightarrow Q_{\geq 0}$ is defined by $we\left(
		q_{0},q_{2}\right) =8,$ $we\left( q_{2},q_{3}\right) =5,$ $we\left(
		q_{2},q_{1}\right) =7,$ $we\left( q_{3},q_{0}\right) =3,$ and $we\left(
		q_{1},q_{0}\right) =10$.
	\end{itemize}
	
	\setlength{\unitlength}{3mm}
	\thinlines
	\begin{picture}(15,15)(-3,-3)
		\put(0.5,0.5){\circle{3}}
		\put(8,1){\circle{3}}
		\put(8.9,8.7){\circle{3}}
		\put(0.5,8){\circle{3}}
		\put(0.5,0.5){\makebox(0,0)[cc]{$q_{2}$}}
		\put(8,1){\makebox(0,0)[cc]{$q_{3}$}}
		\put(8.9,8.7){\makebox(0,0)[cc]{$q_{1}$}}
		\put(0.5,8){\makebox(0,0)[cc]{$q_{0}$}}
		\put(2,1){\vector(1,1){6}}
		\put(4,4){\makebox(0,0)[cc]{$7$}}
		\put(4.5,8.9){\makebox(0,0)[cc]{$10$}}
		\put(4.8,4.8){\makebox(0,0)[cc]{$3$}}
		\put(2,1){\vector(1,0){4}}
		\put(0.2,4){\makebox(0,0)[cc]{$8$}}
		\put(4,0.3){\makebox(0,0)[cc]{$5$}}
		\put(0.8,6.5){\vector(0,-1){4.7}}
		\put(8,7.7){\vector(-1,0){6}}
		\put(7,2){\vector(-1,1){5}}
		\put(8.9,8.7){\makebox(0,0)[cc]{$q_{1}$}}
		\put(13.2,3){\makebox(0,0)[cc]{$\left\{control,upload\right\}$}}
		\put(12.2,10){\makebox(0,0)[cc]{$\left\{upload\right\}$}}
		\put(-2,-2){\makebox(0,0)[cc]{$\left\{gather\right\}$}}
		\put(-1.5,10){\makebox(0,0)[cc]{$\left\{\emptyset \right\}$}}
		\put(0.8,11.5){\vector(0,-1){2}}
	\end{picture}
	\begin{picture}(15,5)
		\put(12,2){\makebox(0,0)[cc]{\shortstack{
					\textbf{Figure 2 }\\The weighted transition system T}}}
	\end{picture}

	\begin{remark}
		\label{Remark _T}Let $r_{1}=q_{0}q_{1}q_{2}\ldots ,r_{2}=q_{0}^{\prime
		}q_{1}^{\prime }q_{2}^{\prime }\ldots $ be runs of $T$ such that $r_{1}\neq
		r_{2},$ then $\exists i\geq 0$ with $q_{i}\neq q_{i}^{\prime }$, which by
		the definition of $T$ implies that $\exists i\geq 0$ such that $L\left(
		q_{i}\right) \neq L\left( q_{i}^{\prime }\right) ,$ i.e., $w_{r_{1}}\neq
		w_{r_{2}}.$
	\end{remark}
	
	We are interested in verifying that $T$ has the following properties that
	could be related with the preservation of the movement of the robot.
	
	\textbf{Property 1}\emph{:}\textit{\ The greatest amount of energy that is
		produced infinitely often along any run of }$T$ \textit{is greater than }$8.$
	
	\textbf{Property 2}\emph{: }\textit{Along any run of }$T,$ \textit{whenever
		a state }$q,$\textit{\ where uploading and control are both executed, is
		visited, then an amount of energy equal to }$8$\textit{\ should be
		produced during a transition that appears in the run after a finite number of transitions after }$q$\textit{\ is
		visited, and whenever a state }$q^{\prime }$\textit{, where uploading is
		executed but control is not, is visited, then an amount of energy equal to
	}$10$\textit{\ should be produced during a transition that appears in the
		run after a finite number of transitions after }$q^{\prime }$\textit{\ is visited, and whenever a state}$~q^{\prime \prime },$\textit{\ where uploading is not executed, is visited,
		then an amount of energy equal to }$8$\textit{\ should be produced during
		a transition that appears in the run after a finite number of transitions after }$q^{\prime \prime }$\textit{\ is
		visited.}
	
	To this end we work as follows. We translate $T$ to a wBa over $K_{2}$ and $%
	AP, $ such that the sequence of weights of the transitions of each path of
	the automaton, with weight different from $-\infty $, over a word $w$
	simulates the sequence of weights of\ a run of $T$ defining the word $w,$
	and vice-versa. Moreover, we consider an $8$-$\vee $-totally restricted $U$%
	-nesting weighted $LTL$-formula $\varphi $ whose semantics expresses \textbf{%
		Properties 1, and 2. }We finally prove that the behavior of $\varphi $
	coincides with the behavior of the automaton.
	
	More specifically, we consider the normalized wBa $\mathcal{\mathcal{M}}%
	=\left( Q,wt,q_{0},Q\right) $ over $K_{2}$ and $\mathcal{P}\left( AP\right) $
	where $wt:Q\times \mathcal{P}\left( AP\right) \times Q$ is given by
	\begin{equation*}
		wt\left( q_{i},\pi ,q_{j}\right) =\left\{
		\begin{array}{ll}
			we\left( q_{i},q_{j}\right) & \text{if }L\left( q_{i}\right) =\pi ,\text{
				and }\left( q_{i},q_{j}\right) \in R \\
			-\infty & \text{otherwise}%
		\end{array}%
		\right.
	\end{equation*}%
	for every $i,j\in \left\{ 0,1,2,3\right\} .$ Clearly, every path of $%
	\mathcal{M}$ over a word $w\in \left( \mathcal{P}\left( AP\right) \right)
	^{\omega }$ is successful.
	
	Next, let $r=q_{0}q_{1}q_{2}\ldots $ be a run of $T$ defining the word $%
	w_{r}.$ Then $P_{w_{r}}^{r}=\left( q_{i},L\left( q_{i}\right)
	,q_{i+1}\right) _{i\geqslant 0}$ is a path of $\mathcal{M}$ over $w_{r}$ with%
	\begin{equation*}
		weight_{\mathcal{M}}\left( P_{w_{r}}^{r}\right) =\mathrm{limsup}\left( \left( wt\left(
		q_{i},L\left(q _{i}\right),q_{i+1}\right) \right) _{i\geq 0}\right)
	\end{equation*}
	\begin{equation*}
		=\mathrm{limsup}%
		\left( \left( we\left( q_{i},q_{i+1}\right) \right) _{i\geq 0}\right) \neq
		-\infty .
	\end{equation*}
	Moreover, if $P_{w}=\left( q_{i},\pi _{i},q_{i+1}\right) _{i\geqslant 0}$ is
	a path of $\mathcal{M}$ over $w=\pi _{0}\pi _{1}\pi _{2}\ldots $ with $%
	weight_{\mathcal{M}}\left( P_{w}\right) \neq -\infty ,$ then $r=q_{0}q_{1}q_{2}\ldots $ is
	a run of $T$ defining the word $w_{r}=w$, and the weight sequence $\omega
	_{r}=\left( wt\left( q_{i},\pi _{i},q_{i+1}\right) \right) _{i\geq 0}=\left(
	we\left( q_{i,}q_{i+1}\right) \right) _{i\geq 0}.$ Thus, there is a
	one-to-one correspondence between the runs of $T$ defining the word $w,$ and
	paths of $\mathcal{M}$ over $w$ with weight$\neq -\infty .$ Then, by Remark~\ref{Remark _T}, we conclude that there exists at most one path $P_{w}$ of $%
	\mathcal{M}$ over $w$ with $weight_{\mathcal{M}%
	}\left( P_{w}\right) \neq -\infty ,$ and $weight_{\mathcal{M}}\left(
	P_{w}\right) $ represents the maximum amount of energy that is produced
	infinitely often by the movement of the robot along the corresponding run $r$
	of $T.$ Thus,
	\begin{equation*}
		\left( \left\Vert \mathcal{M}\right\Vert ,w\right) =\left\{
		\begin{array}{ll}
			weight_{\mathcal{M}}\left( P_{w}^{r}\right) & \text{if there exists a run }r%
			\text{ of }T\text{ with }w_{r}=w \\
			-\infty & \text{otherwise}%
		\end{array}%
		\right.
	\end{equation*}
	
	We consider now the formula $\varphi =\psi \wedge \psi ^{\prime }\in 8${-}$%
	\vee ${-}$t${-}$RULTL\left( K_{2},AP\right) $ where
	\begin{equation*}
		\psi ^{\prime }=\square \left( \left( \lnot upload\wedge 8\right) \vee
		\left( upload\wedge control\wedge 8\right) \vee \left( upload\wedge \lnot
		control\wedge 10\right) \right) ,
	\end{equation*}%
	and%
	\begin{equation*}
		\psi =\varphi _{1}\wedge \square \left( \varphi _{1}\vee \varphi _{2}\vee
		\varphi _{3}\right)
	\end{equation*}%
	where
	
	\begin{itemize}[label=$-$]
		\item $\varphi _{1}=\lnot gather\wedge \lnot upload\wedge \lnot
		control\wedge \bigcirc gather,$ and
		
		\item $\varphi _{2}=gather\wedge \lnot upload\wedge \lnot control\wedge
		\bigcirc upload,$ and
		
		\item $\varphi _{3}=\left( \left( control\wedge upload\wedge \lnot
		gather\right) \vee \left( \lnot control\wedge upload\wedge \lnot
		gather\right) \right) $
		
		$\wedge \bigcirc \left( \lnot gather\wedge \lnot upload\wedge \lnot
		control\right) .$
	\end{itemize}
	
	We note that $\varphi _{1}$ expresses the property that the robot is at a
	state where no action is realized, and at the next state that is visited
	gathering of data should be realized. $\varphi _{2}$ describes the property
	that the robot is at a state where gathering of data is realized, and no
	other action is executed by the robot while it is in that state, and at the
	next state that is visited, uploading of data should be realized. $\varphi
	_{3}$ expresses the property the robot is at a state, where uploading of
	data and control are realized, but gathering of data is not, or at a state where
	only uploading of data is realized, and at the next step a state is visited
	where no action is realized. In fact $\psi$ specifies the requirements that an infinite word $w\in \left( \mathcal{P}\left(
	AP\right) \right) ^{\omega } $ should satisfy so that it is obtained by a run of $T$. More specifically, it holds that if $w\in \left( \mathcal{P}\left(
	AP\right) \right) ^{\omega }$ is a word defined by a run in $T$, then $%
	\left( \left\Vert \psi \right\Vert ,w\right) =\mathbf{1,}$ and vice-versa.
	Moreover,\ $\left( \left\Vert \psi \right\Vert ,w\right) =\mathbf{1},$
	implies $w\models \psi _{b}^{\prime }.$ Finally, $\varphi $ expresses
	\textbf{Property 2}, and since $\varphi $ is $8$-safe we get that $\varphi $
	expresses \textbf{Property 1 }as well.
	
	Then, for every $w=\pi _{0}\pi _{1}\pi _{2}\ldots \in \left( \mathcal{P}%
	\left( AP\right) \right) ^{\omega }$ we get
	\begin{equation*}
		\left( \left\Vert \varphi \right\Vert ,w\right) =\left\{
		\begin{array}{ll}
			8 &
			\begin{array}{l}
				\text{if there exists a run }r\text{ of }T\text{ such that }w_{r}=w, \\
				\text{and }q_{1}\text{ does not appear infinitely often along }r%
			\end{array}%
			\text{ } \\
			10 &
			\begin{array}{l}
				\text{if there exists a run }r\text{ of }T\text{ such that }w_{r}=w, \\
				\text{and }q_{1}\text{ appears infinitely often along }r%
			\end{array}
			\\
			-\infty & \text{ otherwise}%
		\end{array}%
		\right. .
	\end{equation*}
	
	It holds
	\begin{equation*}
		\parallel \varphi \parallel =\Vert \mathcal{M}\Vert .
	\end{equation*}
	
	Then, for every $w\in supp\left( \left\Vert \mathcal{M}\right\Vert \right) ,$
	it holds
	\begin{equation*}
		\left( \parallel \mathcal{M}\parallel ,w\right) =\underset{P_{w}\in succ_{%
				\mathcal{M}}\left( w\right) }{sup}\left( weight_{\mathcal{M}}\left(
		P_{w}\right) \right) =weight_{\mathcal{M}}\left( P_{w}^{r}\right) \geq 8,
	\end{equation*}%
	where $r$ is the unique run of $T$ with $w_{r}=w$. We thus get that the
	maximum amount of energy that is produced infinitely often along any run of $%
	T$ is greater than $8,$ and \textbf{Property 2} is satisfied for any run of $%
	T.$
	
	\section{Algorithms for $k$-safety
		\label{Alg_for_ws}}
	
	In this section, firstly for $K_{2},K_{3}$ (resp. for $K_{1}$), we present an
	algorithm that given a wBa and a $k$-safe $\vee $-totally generalized $U$%
	-nesting weighted $LTL$ formula (resp. $k$-safe totally generalized $U$%
	-nesting weighted $LTL$ formula), decides if the semantics of the formula
	coincides with the behavior of the automaton. The core of the algorithm is
	based on the fact that the weighted $LTL$ formula can be effectively
	translated to a wBa, and the quantitative language equivalence problem is
	decidable for wBa over these specific structures. For the proof of this last
	decidability result we follow the approach of the constructive proof of the
	corresponding result in \cite{Ch-Do}. More precisely, for the proof of the
	following lemmas we follow the constructive approach of Lemmas 1, and 2 in
	\cite{Ch-Do}. We finally conclude that the algorithms presented can be adopted for
	larger families of TGP-$\omega$-valuation monoids, and GP-$\omega$-valuation monoids.
	
	\begin{lemma}
		\label{Lemma 1}Let $\mathcal{M=}\left( Q,wt,q_{0},F\right) $ be a wBa over $%
		K_{3}$ and $A.$ We can effectively construct a wBa $\widehat{\mathcal{M}}%
		\mathcal{=}\left( \widehat{Q},\widehat{wt},\widehat{q_{0}},\widehat{F}%
		\right)$ over $K_{2}$ and $A$ such that $\left\Vert \mathcal{M}\right\Vert
		=\left\Vert \widehat{\mathcal{M}}\right\Vert.$
	\end{lemma}
	
	\begin{proof}
		Let $C=Im\left( wt\right) \diagdown \left\{ \infty \right\} $. We define the
		wBa $\widehat{\mathcal{M}}\mathcal{=}\left( \widehat{Q},\widehat{wt},%
		\widehat{q_{0}},\widehat{F}\right)$ over $K_{2}$ as follows . We set
		
		\begin{itemize}
			\item $\widehat{Q}=Q\cup \left( Q\times C\right) $
			
			\item $\widehat{q_{0}}=q_{0}$
			
			\item $\widehat{wt}\left( p,a,\widetilde{p}\right) =\left\{
			\begin{array}{ll}
				wt\left( p,a,\widetilde{p}\right) & \text{ \ if }\left( p,a,\widetilde{p%
				}\right) \in Q\times A\times Q\text{ } \\
				&  \\
				wt\left( p,a,\widetilde{q}\right) & \text{ \ if }p\in Q,\widetilde{p}=\left(
				\widetilde{q},wt\left( p,a,\widetilde{q}\right) \right) \in Q\times C \\
				&  \\
				v &
				\begin{array}{l}
					\text{if }p=\left( q,v\right) ,\widetilde{p}=\left( \widetilde{q},v\right)
					\in Q\times C,q\neq \widetilde{q},\text{ } \\
					\text{and }wt\left( q,a,\widetilde{q}\right) \leq v%
				\end{array}
				\\
				&  \\
				v & \text{ \ if }p=\widetilde{p}=\left( q,v\right) \in Q\times C \\
				&  \\
				-\infty & \text{ \ otherwise}%
			\end{array}%
			\right. $
			
			\item $\widehat{F}=F\cup \left( F\times C\right) $
		\end{itemize}
		
		Let $w=a_{0}a_{1}\ldots \in A^{\omega }.$ We show that for every $\widehat{P}%
		_{w}\in succ_{\widehat{\mathcal{M}}}\left( w\right) $ with $weight_{\widehat{%
				\mathcal{M}}}\left( \widehat{P}_{w}\right) \neq -\infty ,$ there exists a
		successful path $P_{w}$ of $\mathcal{M}$ over $w$ with $weight_{\widehat{%
				\mathcal{M}}}\left( \widehat{P}_{w}\right) \leq weight_{\mathcal{M}}\left(
		P_{w}\right) .$ We let $\widehat{P}_{w}=t_{0}t_{1}t_{2}\ldots \in succ_{%
			\widehat{\mathcal{M}}}\left( w\right) $ with $weight_{\widehat{\mathcal{M}}%
		}\left( \widehat{P}_{w}\right) \neq -\infty $. Let $In^{\widehat{Q}}\left(
		\widehat{P}_{w}\right) \subseteq Q$. Then, $\widehat{P}_{w}$ is also a
		successful path of $\mathcal{M}$ over $w$, and we point out the following
		cases. (a) If $weight_{\widehat{\mathcal{M}}}\left( \widehat{P}_{w}\right)
		\neq \infty ,$ then $weight_{\widehat{\mathcal{M}}}\left( \widehat{P}%
		_{w}\right) =\mathit{limsup}\left( \widehat{wt}\left( t_{i}\right) \right)
		_{i\geq 0}\leq weight_{\mathcal{M}}\left( P_{w}\right)
		=\sup_{-\infty }\left( wt\left( t_{i}\right) \right) _{i\geq 0}$
		as desired. (b) If $weight_{\widehat{\mathcal{M}}}\left( \widehat{P}%
		_{w}\right) =\infty ,$ then the weight of each transition is equal to $%
		\infty $, and thus $\widehat{P}_{w}$ is also a successful path of $\mathcal{M%
		}$ with the same weight. Assume now that there exists $q\in In^{\widehat{Q}%
		}\left( \widehat{P}_{w}\right) ,$ such that $q\notin Q.$ Then, $weight_{%
			\widehat{\mathcal{M}}}\left( \widehat{P}_{w}\right) \neq \infty ,$ and $%
		\widehat{P}_{w}=\left( q_{0},a_{0},q_{1}\right) \ldots \left(
		q_{j},a_{j},q_{j+1}\right) \left( q_{j+1},a_{i+1,}\left( q_{i+2},v\right)
		\right) \left( \left( q_{i+j},v\right) ,a_{i+j},\left( q_{i+j+1},v\right)
		\right) _{j\geq 2}$\\ with $weight_{\widehat{\mathcal{M}}}\left( \widehat{P}%
		_{w}\right) =v=\widehat{wt}\left( q_{j+1},a_{j+1},\left( q_{j+2},v\right)
		\right) =wt\left( q_{j+1},a_{j+1},q_{j+2}\right) $. Then, the path $%
		P_{w}=\left( q_{i},a_{i},q_{i+1}\right) _{i\geq 0}$ is a successful path of $%
		\mathcal{M}$ over $w$ with $weight_{\widehat{\mathcal{M}}}\left( \widehat{P}%
		_{w}\right) \leq weight_{\mathcal{M}}\left( P_{w}\right) $ as desired.
		
		Thus,%
		\begin{eqnarray}
			\left( \left\Vert \widehat{\mathcal{M}}\right\Vert ,w\right) &=&\underset{%
				\widehat{P}_{w}\in succ_{\widehat{\mathcal{M}}}\left( w\right) }{\sup }%
			\left( weight_{\widehat{\mathcal{M}}}\left( \widehat{P}_{w}\right) \right)
			\label{Lemma1,1} \\
			&\leq &\underset{P_{w}\in succ_{\mathcal{M}}\left( w\right) }{\sup }\left(
			weight_{\mathcal{M}}\left( P_{w}\right) \right) =\left( \left\Vert \mathcal{M%
			}\right\Vert ,w\right) .  \notag
		\end{eqnarray}
		
		We now prove that for every $P_{w}\in succ_{\mathcal{M}}\left( w\right) $
		with $weight_{\mathcal{M}}\left( P_{w}\right) \neq -\infty ,$ there exists a
		$\widehat{P}_{w}\in succ_{\widehat{\mathcal{M}}}\left( w\right) $ such that $%
		weight_{\mathcal{M}}\left( P_{w}\right) \leq weight_{\widehat{\mathcal{M}}%
		}\left( \widehat{P}_{w}\right) .$ We let $P_{w}=\left(t_{i}\right)_{i\geq0}  \in
		succ_{\mathcal{M}}\left( w\right) $ with $weight_{\mathcal{M}}\left(
		P_{w}\right) \neq -\infty .$ If $weight_{\mathcal{M}}\left( P_{w}\right)
		\\\neq \infty $, then we obtain a successful path $\widehat{P}_{w}$ of $%
		\widehat{\mathcal{M}}$ over $w$ with $weight_{\widehat{\mathcal{M}}}\left(
		\widehat{P}_{w}\right) =weight_{\mathcal{M}}\left( P_{w}\right) $ in the
		following way: We let $t_{i}=\left( q_{i},a_{i},q_{i+1}\right) $ for all $%
		i\geq 0$, and we construct $\widehat{P}_{w}$ by copying the finite prefix of
		$P_{w}$ until we reach a transition $t_{i}$ with $wt\left( t_{i}\right)
		=v=weight_{\mathcal{M}}\left( P_{w}\right) $, and then we act as follows: if
		$q_{i+1}\in F$ we loop through $\left( q_{i+1},v\right) $, i.e.,
		\begin{equation*}
			\widehat{P}_{w}=\left( q_{0},a_{0},q_{1}\right) \ldots \left(
			q_{i},a_{i},\left( q_{i+1},v\right) \right) \left( \left( q_{i+1},v\right)
			,a_{i+j},\left( q_{i+1},v\right) \right) _{j\geq 1}
		\end{equation*}
		otherwise we reach a final state $\left( q_{i+k},v\right) ,k\geq
		2,q_{i+k}\in F,$ with the sequence of transitions \\ \\ $\left(
		q_{i+1},a_{i+1},\left( q_{i+2},v\right) \right) \left( \left(
		q_{i+2},v\right) ,a_{i+2},\left( q_{i+3},v\right) \right) \ldots \left(
		\left( q_{i+k-1},v\right) ,a_{i+k-1},\left( q_{i+k},v\right) \right) $\ \\\\ and
		finally loop through $\left( q_{i+k},v\right) ,$ i.e.,
		\begin{eqnarray*}
			\widehat{P}_{w} &=&\left( q_{0},a_{0},q_{1}\right) \ldots \left(
			q_{i},a_{i},\left( q_{i+1} ,v\right)\right) \left( q_{i+1},a_{i+1},\left(
			q_{i+2},v\right) \right) \\
			&&\left( \left( q_{i+2},v\right) ,a_{i+2},\left( q_{i+3},v\right) \right)
			\ldots \left( \left( q_{i+k-1},v\right) ,a_{i+k-1},\left( q_{i+k},v\right)
			\right) \\&&
			\left( \left( q_{i+k},v\right) ,a_{i+j},\left( q_{i+k},v\right)
			\right) _{j\geq k}
		\end{eqnarray*}
		If $weight_{\mathcal{M}}\left( P_{w}\right) =\infty ,$ then $P_{w}$ is also
		a successful path of $\widehat{\mathcal{M}}$ over $w$ with the same weight. Thus,%
		\begin{eqnarray}
			\left( \left\Vert \mathcal{M}\right\Vert ,w\right) &=&\underset{P_{w}\in
				succ_{\mathcal{M}}\left( w\right) }{\sup }\left( weight_{\mathcal{M}}\left(
			P_{w}\right) \right)  \label{Lemma1,2} \\
			&\leq &\underset{\widehat{P}_{w}\in succ_{\widehat{\mathcal{M}}}\left(
				w\right) }{\sup }\left( weight_{\widehat{\mathcal{M}}}\left( \widehat{P}%
			_{w}\right) \right) =\left( \left\Vert \widehat{\mathcal{M}}\right\Vert
			,w\right) .  \notag
		\end{eqnarray}
		
		Hence, by (\ref{Lemma1,1}), (\ref{Lemma1,2}) we get that $\left( \left\Vert
		\widehat{\mathcal{M}}\right\Vert ,w\right) =\left( \left\Vert \mathcal{M}%
		\right\Vert ,w\right) $ for every $w\in A^{\omega },$ and this concludes our
		proof.
	\end{proof}
	
	\begin{lemma}
		\label{Lemma 2}
		Let $\mathcal{N=}\left( Q,wt,q_{0},F\right) $ be a normalized wBa over $%
		K_{2}$ and $A$, and $v\in \overline{%
			\mathbb{Q}
		}.$ We can effectively construct a Ba $\mathcal{N}^{\geq v}$ such that $%
		\left\Vert \mathcal{N}^{\geq v}\right\Vert =\left\{ w\in A^{\omega }\mid
		\left( \left\Vert \mathcal{N}\right\Vert ,w\right) \geq v\right\} .$
	\end{lemma}
	
	\begin{proof}
		We first construct a Ma $\mathcal{M}^{\geq v}$ with $%
		\left\Vert \mathcal{M}^{\geq v}\right\Vert =\left\{ w\in A^{\omega }\mid
		\left( \left\Vert \mathcal{N}\right\Vert ,w\right) \geq v\right\} .$ Then, we obtain the claim by Theorem \ref{B-M-R}. We assume first that $v\in Q.$ We let $\mathcal{M}^{\geq v}=\left(
		Q^{\geq v},A,q_{0}^{\geq v},\Delta ^{\geq v},F^{\geq v}\right) $ where
		
		\begin{itemize}
			\item $Q^{\geq v}=\left( Q\times \left\{ 0,1,3,4\right\} \times \left\{
			B,C\right\} \right) \cup \left( \widehat{Q}\times \left\{ 2,5,6\right\}
			\times \left\{ B,C\right\} \right) $
			
			where $\widehat{Q}=\left\{ s_{q}\mid q\in Q\right\} $ is a set of copies of
			the elements of $Q$
			
			\item $q_{0}^{\geq v}=\left( q_{0},0,C\right)$
			\item The set of final sets is defined as follows
			
			$F^{\geq v}=\left\{ \widehat{F}\mid \widehat{F}\subseteq F_{1}\cup F_{2}\cup
			F_{3},F_{1}\cap \widehat{F}\neq \emptyset ,F_{2}\cap \widehat{F}\neq
			\emptyset \right\} $ where
			
			$F_{1}=\left( Q\times \left\{ 3\right\} \times \left\{ C,B\right\} \right)
			\cup \left( \widehat{Q}\times \left\{ 2\right\} \times \left\{ C,B\right\}
			\right) \cup \left( \widehat{Q}\times \left\{ 6\right\} \times \left\{
			C,B\right\} \right), $
			
			$F_{2}=\left( Q\times \left\{ 0,1,3,4\right\} \times \left\{ B\right\}
			\right) \cup \left( \widehat{Q}\times \left\{ 2,5,6\right\} \times \left\{
			B\right\} \right), $ and
			
			$F_{3}=Q\times \left\{ 1,4\right\} \times \left\{ B,C\right\}.$
		\end{itemize}
		
		The definition of $\Delta ^{\geq v}$, and the proof of the equality
		\begin{equation*}
			\left\Vert \mathcal{M}^{\geq v}\right\Vert =\left\{ w\in A^{\omega }\mid
			\left( \left\Vert \mathcal{N}\right\Vert ,w\right) \geq v\right\}
		\end{equation*}
		is presented in an analytical way in the Appendix. We only here describe intuitively the operation of $\mathcal{M}^{\geq v}$. The set of states of $\mathcal{M}^{\geq v}$ are triples in which the first component is either a state of $\mathcal{N}$ or a copy of a state of $\mathcal{N}$. A transition $\left( \left( q,i,k\right) ,a,\left(\overline{q},j,l\right) \right)$ of $\mathcal{M}^{\geq v}$ between states whose first component is an element of $Q$ simulates the transition $\left(q,a,\overline{q}\right)$ of $\mathcal{N}$. At the same time the information on whether the weight of the simulated transition is greater or equal to $v$, or smaller to $v$, or equal to $\infty$ is indicated by the number that appears in second component of the arriving state, i.e., $j$ is equal to $1, 3,$ or~$4$ respectively. The third component of the arriving state is $B$ if the arriving state of the original transition is final, and $C$ otherwise.
		Copies of states of~$\mathcal{N}$ are used to simulate paths of $\mathcal{N}$ where there exists a position after which the weight of each transition, appearing in the path, is $\infty$. More precisely, a path of the original automaton where the weight of each transition is equal to $\infty$ is simulated by a path of $\mathcal{M}^{\geq v}$ where the first component of every state, but the first one, is a copy of the corresponding state of the original path. The second component of every state, but the first one, of the path of $\mathcal{M}^{\geq v}$ is equal to $2$, indicating that the weight of the transition, that is being simulated, is equal to $\infty$. The third components of the states that appear in the new path are either $B$, or $C$, depending on whether the first component of the state is a copy of a final state of the original automaton or not. Finally, paths of the original automaton (with weight different from $\neq -\infty$) in which there are finitely many transitions with a rational weight, and at least one transition with weight greater, or equal to $v$, are simulated in the new automaton by paths where the automaton transits to a state where the first component is a copy of the corresponding state of $\mathcal{N}$ when the weight of the transition that is being simulated is a rational number greater, or equal to $v$. The second component of this state equal to $5$. Then, the rest of the original original path is simulated by transitions between states whose first component is a copy of an element of $Q$, and second component is equal to $5$, up until the point where some transition with weight $\infty$ is simulated. Then, $\mathcal{M}^{\geq v}$ enters a state whose second component is equal to $6$. Finally, the new automaton can move from a state with second component~$6$, only to states whose second component is also $6$, and only to simulate a transition of the original automaton whose weight is equal to $\infty$. The Muller acceptance condition is used to verify that successful paths of the new automaton simulate successful paths of $\mathcal{N}$ in which the weights of the transitions are equal to the desired values infinitely often.
		
		If $v=\infty ,$ we prove our claim similarly. Finally, if $v=-\infty ,$ then
		$\left( \left\Vert \mathcal{N}\right\Vert ,w\right) \geq v$ for every $w\in
		A^{\omega },$ which implies that $\left\{ w\in A^{\omega }\mid \left(
		\left\Vert \mathcal{N}\right\Vert ,w\right) \geq v\right\} \\=A^{\omega },$
		and it is trivial to construct the desired Ba.
	\end{proof}
	
	\begin{lemma}
		\label{Lemma 2 inf}
		Let $\mathcal{N=}\left( Q,wt,q_{0},F\right)$ be a normalized wBa over $%
		K_{1}$ and $A$, and $v\in
		\mathbb{Q}
		.$ We can effectively construct a Ba $\mathcal{N}^{\geq v}$ over $A$ such
		that $\left\Vert \mathcal{N}^{\geq v}\right\Vert =\left\{ w\in A^{\omega}\mid \left( \left\Vert \mathcal{N}\right\Vert ,w\right) \geq v\right\} .$
	\end{lemma}
	
	\begin{proof}
		To prove the Lemma's claim we first construct a Ra $\mathcal{M}^{\geq v}$ with the property that $\left\Vert \mathcal{M}^{\geq v}\right\Vert =\left\{ w\in A^{\omega}\mid \left( \left\Vert \mathcal{N}\right\Vert ,w\right) \geq v\right\} .$ Then, the claim is obtained by Theorem \ref{B-M-R}. The construction of $\mathcal{M}^{\geq v}$ is presented in the Appendix.
	\end{proof}
	
	\begin{theorem}
		\label{quant incl}The quantitative language-inclusion problem is decidable
		for wBa over $K_{1}$, or $K_{2}$, or $K_{3}$ and $A.$
	\end{theorem}
	
	\begin{proof}
		Let $\mathcal{N=}\left( Q,wt,q_{0},F\right) $, $\widehat{\mathcal{N}}%
		\mathcal{=}\left( \widehat{Q},\widehat{wt},\widehat{q_{0}},\widehat{F}%
		\right) $ be two wBa over $K_{2} $ and $A.$ Let $C=Im\left( wt\right) ,$ and
		$\widehat{C}=Im\left( \widehat{wt}\right) $. Clearly, $C, \widehat{C}$ are computable. Since $\left( \left\Vert
		\mathcal{N}\right\Vert ,w\right) \in C$ for all $w\in A^{\omega },$ it holds
		$\left\Vert \mathcal{N}\right\Vert \sqsubseteq $ $\left\Vert \widehat{%
			\mathcal{N}}\right\Vert $ iff $\left\Vert \mathcal{N}^{\geq v}\right\Vert
		\subseteq \left\Vert \widehat{\mathcal{N}}^{\geq v}\right\Vert $ for all $%
		v\in C,$ where $\mathcal{N}^{\geq v}$, $\widehat{\mathcal{N}}^{\geq v}$ are
		the Ba obtained by $\mathcal{N}$, $\widehat{\mathcal{N}}$ following the
		proof of Lemma \ref{Lemma 2}. We thus conclude that the quantitative
		inclusion problem is decidable for wBa over $K_{2}$. The proof for wBa over $%
		K_{1}$ is treated similarly. We conclude our claim for wBa over $K_{3}$ by
		Lemma \ref{Lemma 1}$.$
	\end{proof}
	
	\begin{theorem}
		\label{Theorem quant equal}The quantitative language equivalence problem is
		decidable for wBa over $K_{1}$, or $K_{2}$, or $K_{3}$ and $A.$
	\end{theorem}
	
	Let now $AP$ be a finite set of atomic propositions. The following Theorem
	is concluded by Corollary \ref{Corollary}, and Theorem \ref{Theorem quant equal}.
	
	\begin{theorem}
		\begin{enumerate}[label=(\roman*)]
			\item  Let $k\in K_{1}\backslash \left\{ -\infty ,\infty \right\} .$ If $		\varphi \in k$-$t$-$RULTL\left( K_{1},AP\right) ,$ and $\mathcal{M}$ is a wBa over $K_{1}$ and $\mathcal{P}\left( AP\right) ,$ it is decidable whether $\left\Vert \varphi \right\Vert =\left\Vert \mathcal{M}\right\Vert .$
		
		 	\item Let $k\in K_{2}\backslash \left\{ -\infty ,\infty \right\} .$ If $		\varphi \in k$-$\vee $-$t$-$RULTL\left( K_{2},AP\right) ,$ and $\mathcal{M}$ 		is a wBa over $K_{2}$ and $\mathcal{P}\left( AP\right) ,$ it is decidable 		whether $\left\Vert \varphi \right\Vert =\left\Vert \mathcal{M}\right\Vert .$
		
			\item Let $k\in K_{3}\backslash \left\{ -\infty ,\infty \right\} .$ If $\varphi \in k$-$\vee $-$t$-$RULTL\left( K_{3},AP\right) ,$ and $\mathcal{M}$ is a wBa over $K_{3}$ and $\mathcal{P}\left( AP\right) ,$ it is decidable	whether $\left\Vert \varphi \right\Vert =\left\Vert \mathcal{M}\right\Vert .$
	\end{enumerate}
	\end{theorem}

	\noindent Let now $K$ be a totally ordered set with a minimum, and a maximum element, $\mathbf{0},\mathbf{1}$ respectively, and $M_{1}, M_{2}, M_{3}$ be the TGP-$\omega$-valuation monoids as defined in Remark $\ref{Remark tos} $. Following the arguments of the proofs of Lemmas $\ref{Lemma 1}$, $\ref{Lemma 2}$, $\ref{Lemma 2 inf}$, and Theorems $\ref{quant incl}$, $\ref{Theorem quant equal}$ we obtain corresponding results for $M_{1}, M_{2}, M_{3}$. By these and Corollary \ref{Corollary} we obtain the following corollary.
	
	\begin{corollary}
		\begin{enumerate}[label=(\roman*)]
		 	\item Let $k\in M_{1}\backslash \left\{ \mathbf{0} ,\mathbf{1} \right\} .$ If $		\varphi \in k$-$t$-$RULTL\left( M_{1},AP\right) ,$ and $\mathcal{M}$ is a wBa over $M_{1}$ and $\mathcal{P}\left( AP\right) ,$ it is decidable whether 		$\left\Vert \varphi \right\Vert =\left\Vert \mathcal{M}\right\Vert .$
		
			\item Let $k\in M_{2}\backslash \left\{ \mathbf{0} ,\mathbf{1}\right\} .$ If $		\varphi \in k$-$\vee $-$t$-$RULTL\left( M_{2},AP\right) ,$ and $\mathcal{M}$		is a wBa over $M_{2}$ and $\mathcal{P}\left( AP\right) ,$ it is decidable 		whether $\left\Vert \varphi \right\Vert =\left\Vert \mathcal{M}\right\Vert .$
		
			\item Let $k\in M_{3}\backslash \left\{ \mathbf{0} ,\mathbf{1}\right\} .$ If $		\varphi \in k$-$\vee $-$t$-$RULTL\left( M_{3},AP\right) ,$ and $\mathcal{M}$		is a wBa over $M_{3}$ and $\mathcal{P}\left( AP\right) ,$ it is decidable		whether $\left\Vert \varphi \right\Vert =\left\Vert \mathcal{M}\right\Vert .$
		\end{enumerate}
	\end{corollary}
	
	We are ready now to present the following algorithms. These can be adopted in a straightforward way for $ M_{2}, M_{3}, M_{1}$ respectively.
	
	\begin{algorithm}[ht!]
		\caption{Decision procedure for $K_{2}$}
		\KwIn{wBa $\mathcal{M=}\left( Q,wt_{%
			\mathcal{M}},q_{0},F\right) $ and $k$-safe $\vee$-totally restricted $U$-nesting weighted $LTL$ formula $\varphi$}
		\KwOut{"yes" if $\left\Vert \mathcal{M}\right\Vert =\left\Vert
			\varphi \right\Vert ,$ otherwise "no"}
		
		Determine $\widetilde{\varphi } \in $$\vee$-$t$-$RULTL\left( K_{2},AP\right) $ with $\left\Vert \varphi \right\Vert =\left\Vert \widetilde{%
			\varphi }\right\Vert $
		
		Construct a wBa $\mathcal{A}_{\widetilde{\varphi }}=\left( S,wt_{\mathcal{A}%
			\widetilde{_{\varphi }}},p_{0},P\right) $ such that $\left\Vert \mathcal{A}_{%
			\widetilde{\varphi }}\right\Vert =\left\Vert \varphi \right\Vert $
		
		Compute $Im\left( wt_{\mathcal{A}\widetilde{_{\varphi }}}\right) $
		
		Compute $Im\left( wt_{\mathcal{M}}\right) $

			For every $v\in Im\left( wt_{\mathcal{A}\widetilde{_{\widetilde{\varphi }}}%
		}\right) $
		
		\ \ \ \ \ construct the Ba $\mathcal{A}_{\widetilde{\varphi }}^{\geq v}$
		
		$\ \ \ \ \ $construct the Ba $\mathcal{M}^{\geq v}$
		
		\ \ \ If $\left\Vert \mathcal{A}_{\widetilde{\varphi }}^{\geq v}\right\Vert
		\nsubseteq \left\Vert \mathcal{M}^{\geq v}\right\Vert $
		
		\ \ \ \ \ \ \ \ \ return "no"
		
		end for
		\end{algorithm}

\begin{algorithm}[ht]
		For every $v\in Im\left( wt_{\mathcal{M}}\right) $
		
		\ \ \ \ \ construct the Ba $\mathcal{A}_{\widetilde{\varphi }}^{\geq v}$
		
		$\ \ \ \ \ $construct the Ba $\mathcal{M}^{\geq v}$
		
		\ \ \ If $\left\Vert \mathcal{M}^{\geq v}\right\Vert \nsubseteq \left\Vert
		\mathcal{A}_{\widetilde{\varphi }}^{\geq v}\right\Vert $
		
		\ \ \ \ \ \ \ \ \ return "no"
		
		end for
		
		return "yes"

	\end{algorithm}

	\begin{algorithm}[h!]
	\caption{Decision procedure for $K_{3}$}
	\KwIn{wBa $\mathcal{M=}\left( Q,wt_{%
			\mathcal{M}},q_{0},F\right) $ and $k$-safe $\vee $-totally\ restricted $U$-nesting weighted $LTL$-formula $\varphi $}
	\KwOut{"yes" if $\left\Vert \mathcal{M}\right\Vert =\left\Vert
		\varphi \right\Vert ,$ otherwise "no"}
	
	Determine $\widetilde{%
		\varphi } \in $ $\vee$-$t$-$RULTL\left( K_{3},AP\right) $ with $\left\Vert \varphi \right\Vert =\left\Vert \widetilde{%
		\varphi }\right\Vert $
	
	Construct a wBa $\mathcal{A}_{\widetilde{\varphi }}=\left( S,wt_{\mathcal{A}%
		\widetilde{_{\varphi }}},p_{0},P\right) $ over $K_{3}$ such that $\left\Vert \mathcal{A}_{%
		\widetilde{\varphi }}\right\Vert =\left\Vert \varphi \right\Vert $
	
	Construct the wBa $\widehat{\mathcal{A}_{\widetilde{\varphi }}},$ $\widehat{%
		\mathcal{M}}$ over $K_{2}$ with
	
	$\left\Vert \widehat{\mathcal{A}_{\widetilde{\varphi }}}\right\Vert
	=\left\Vert \mathcal{A}_{\widetilde{\varphi }}\right\Vert ,\left\Vert
	\widehat{\mathcal{M}}\right\Vert =\left\Vert \mathcal{M}\right\Vert $
	
	Compute $Im\left( wt_{\widehat{\mathcal{A}_{\widetilde{\varphi }}}}\right) $
	
	Compute $Im\left( wt_{\widehat{\mathcal{M}}}\right) $
	
	For every $v\in Im\left( wt_{\widehat{\mathcal{A}_{\widetilde{\varphi }}}%
	}\right) $
	
	\ \ \ \ \ construct the Ba $\widehat{\mathcal{A}_{\widetilde{\varphi }}}%
	^{\geq v}$
	
	$\ \ \ \ \ $construct the Ba $\widehat{\mathcal{M}}^{\geq v}$
	
	\ \ \ If $\left\Vert \widehat{\mathcal{A}_{\widetilde{\varphi }}}^{\geq
		v}\right\Vert \nsubseteq \left\Vert \widehat{\mathcal{M}}^{\geq
		v}\right\Vert $
	
	\ \ \ \ \ \ \ \ \ return "no"
	
	end for
	
	For every $v\in Im\left( wt_{\widehat{\mathcal{M}}}\right) $
	
	\ \ \ \ \ construct the Ba $\widehat{\mathcal{A}_{\widetilde{\varphi }}}%
	^{\geq v}$
	
	$\ \ \ \ \ $construct the Ba $\widehat{\mathcal{M}}^{\geq v}$
	
	\ \ \ If $\left\Vert \widehat{\mathcal{M}}^{\geq v}\right\Vert \nsubseteq
	\left\Vert \widehat{\mathcal{A}_{\widetilde{\varphi }}}^{\geq v}\right\Vert $
	
	\ \ \ \ \ \ \ \ \ return "no"
	
	end for
	
	return "yes"
\end{algorithm}

		\begin{algorithm}[h!]
		\caption{Decision procedure for $K_{1}$}
		\KwIn{wBa $\mathcal{M=}\left( Q,wt_{%
				\mathcal{M}},q_{0},F\right) $ and $k$-safe totally restricted $U$-nesting weighted $LTL$-formula $\varphi $}
		\KwOut{"yes" if $\left\Vert \mathcal{M}\right\Vert =\left\Vert
			\varphi \right\Vert ,$ otherwise "no"}
		
		Determine $\widetilde{\varphi } \in$ $t$-$RULTL\left( K_{1},AP\right) $
		with $\left\Vert \varphi \right\Vert =\left\Vert \widetilde{\varphi }%
		\right\Vert $
		
		Construct a wBa $\mathcal{A}_{\widetilde{\varphi }}=\left( S,wt_{\mathcal{A}%
			\widetilde{_{\varphi }}},p_{0},P\right) $ such that $\left\Vert \mathcal{A}_{%
			\widetilde{\varphi }}\right\Vert =\left\Vert \varphi \right\Vert $
		
		Compute $Im\left( wt_{\mathcal{A}_{\widetilde{\varphi }}}\right) $
		
		Compute $Im\left( wt_{\mathcal{M}}\right) $
		
		For every $v\in Im\left( wt_{\mathcal{A}_{\widetilde{\varphi }}}\right) $
		
		\ \ \ \ \ construct the Ba $\mathcal{A}_{\widetilde{\varphi }}^{\geq v}$
		
		$\ \ \ \ \ $construct the Ba $\mathcal{M}^{\geq v}$
		
		\ \ \ If $\left\Vert \mathcal{A}_{\widetilde{\varphi }}^{\geq v}\right\Vert
		\nsubseteq \left\Vert \mathcal{M}^{\geq v}\right\Vert $
		
		\ \ \ \ \ \ \ \ \ return "no"
		
		end for
		
		For every $v\in Im\left( wt_{\mathcal{M}}\right) $
		
		\ \ \ \ \ construct the Ba $\mathcal{A}_{\widetilde{\varphi }}^{\geq v}$
		
		$\ \ \ \ \ $construct the Ba $\mathcal{M}^{\geq v}$
		
		\ \ \ If $\left\Vert \mathcal{M}^{\geq v}\right\Vert \nsubseteq \left\Vert
		\mathcal{A}_{\widetilde{\varphi }}^{\geq v}\right\Vert $
		
		\ \ \ \ \ \ \ \ \ return "no"
		
		end for
		
		return "yes"
	\end{algorithm}

	\section{Conclusion\label{Conclusion}}
	
	In this work we have introduced the notion of $k$-safe infinitary series
	over idempotent ordered TGP-$\omega $-valuation
	monoids $\left( K,+,\cdot ,Val^{\omega },\mathbf{0},\mathbf{1}\right) $ that
	satisfy specific properties, where $k\in K\backslash \left\{ \mathbf{0},%
	\mathbf{1}\right\} $. For every $k\in K\backslash \left\{ \mathbf{0},\mathbf{%
		1}\right\} ,$ we define two syntactic fragments of the weighted $LTL$, the fragment of $k$-safe
	totally restricted $U$-nesting weighted $LTL$ formulas, and the fragment of $%
	k$-safe $\vee $-totally restricted $U$-nesting weighted $LTL$ formulas (the later is a subclass of the former) with
	the property that the semantics of the formulas in these fragments are $k$%
	-safe infinitary series. Whether a $k$-safe infinitary series over a TGP-$\omega$-valuation monoid $K$ definable by the weighted $LTL$ over~$K$ is also definable by a formula in these fragments remains open. Moreover, for every formula in these fragments, we
	can effectively construct a wBa over $K$ recognizing its semantics. We
	present an example of how $k$-safe weighted $LTL$ formulas can be used to
	describe quantitative properties of a weighted transition system, and how we
	can verify that the system has the desired properties by translating it to a
	wBa that simulates its runs, and proving that the behavior of the automaton
	coincides with the semantics of the formula. Finally, for two specific
	idempotent ordered TGP-$\omega $-valuation monoids
	(resp. a specific idempotent ordered GP-$\omega $-valuation
	monoid), we present decision procedures that given a wBa and a $k$-safe $\vee$-totally restricted $U$-nesting weighted $%
	LTL $ formula (resp. $k$-safe totally restricted $U$-nesting weighted $LTL$
	formula), decide if the semantics of the formula coincides with the behavior
	of the automaton. These results are generalized for specific families of idempotent ordered TGP-$\omega$-valuation monoids, and idempotent ordered GP-$\omega$-valuation monoids.
	
	The study of the complexity of the proposed algorithms, and
	the study of decidability questions related to the
	formulas in the proposed fragments is a work in progress. Another possible
	road is the study of the use of $k$-safe weighted $LTL$ formulas introduced,
	in the definition of algorithms that given a weighted transition system, a
	boolean safety formula, and an $\omega $-valuation function, decide whether
	the words defined by runs of the weighted transition system satisfy the
	specification, and at the same time corresponding infinite sequences of
	weights are valuated through the $\omega $-valuation function to a price
	above a given threshold. A natural theoretical next step is the definition
	and study of other kinds of quantitative properties of infinitary series
	over TGP-$\omega $-valuation monoids, like weighted liveness, in the framework of the weighted $LTL$ introduced in this work.
	The use of temporal logic and model checking techniques in the description
	and verification of properties of biological systems has been proposed, and
	is being studied in the literature~\cite{Fi-Pi}. We see the use of the weighted $LTL$ and of $k
	$-safe weighted $LTL$ formulas in the description, and verification of
	quantitative properties of such systems as a challenging perspective.

		\newpage
		\section*{Appendix}\label{appendix}
		
		\begin{proof}[Proof of Example \protect\ref{Example 2}]
			
			We will prove first that $K_{2}$ is an idempotent ordered TGP-$%
			\omega $-valuation monoid that satisfy Properties \ref{Property 1}, \ref{Property 3}, \ref{Property 5}, and \ref{Property 6} . Property \ref{Property 3} follows directly from
			the definition of $K_{2}$. First, we prove Property \ref{Property Dist} for TGP-$\omega $-valuation monoids.
			
			We consider $L\subseteq _{fin}K_{2},$ finite index sets $I_{j}$($j\geq 0$), and $k_{i_{j}}\in
			L\left( i_{j}\in I_{j}\right) $ such that for all $j\geq 0,$ it holds $%
			k_{i_{j}}\in L\backslash \left\{ \infty,-\infty \right\} $ for all $i_{j}\in
			I_{j},$ or $k_{i_{j}}\in \left\{ \infty,-\infty\right\} $ for all $i_{j}\in
			I_{j}$. We will prove that
			\begin{equation*}
				\text{limsup}\left( \underset{i_{j}\in I_{j}}{\sup }k_{i_{j}}\right) _{j\in
					\mathbb{N}
				}=\underset{\left( i_{j}\right) _{j}\in I_{0}\times I_{1}\times \ldots }{%
					\sup }\left( \text{limsup}\left( k_{i_{j}}\right) _{j\in
					\mathbb{N}
				}\right) .
			\end{equation*}%
			We set $A=$limsup$\left( \underset{i_{j}\in I_{j}}{\sup }k_{i_{j}}\right)
			_{j\in
				\mathbb{N}
			},$ and $B=\underset{\left( i_{j}\right) _{j}\in I_{0}\times I_{1}\times
				\ldots }{\sup }\left( \text{limsup}\left( k_{i_{j}}\right) _{j\in
				\mathbb{N}
			}\right) .$
			
			Assume that there exists an $l\geq 0,$ such that $k_{i_{l}}=-\infty $ for all
			$i_{l}\in I_{l}$. Then, for all $\left( i_{j}\right) _{j}\in I_{0}\times
			I_{1}\times \ldots ,$ limsup$\left( k_{i_{j}}\right) _{j\in
				\mathbb{N}
			}=-\infty ,$ i.e., $B=-\infty .$ Moreover, $\underset{i_{l}\in I_{l}}{\sup }%
			k_{i_{l}}=-\infty ,$ and thus $A=-\infty $ as desired.
			\\Otherwise, we point out the following cases:
			
			\begin{enumerate}[label=\textnormal{(\Roman*)}]
			\item Assume that for all $j\geq 0$ there exists $i_{j}\in I_{j}$ such that $k_{i_{j}}=\infty .$ Then,
			there exist $\left( i_{j}\right) _{j}\in I_{0}\times I_{1}\times \ldots $
			such that limsup$\left( k_{i_{j}}\right) _{j\in
				\mathbb{N}
			}=\infty $ which implies that $B=\infty .$ In addition, we get that $%
			\underset{i_{j}\in I_{j}}{\sup }k_{i_{j}}=\infty $ for all $j\geq 0,$ i.e., $%
			A=\infty $ as well.
			
			\item Assume that there exists finitely many $j\geq 0,$ such that $%
			k_{i_{j}}\neq \infty $ for all $i_{j}\in I_{j},$ then $\underset{i_{j}\in
				I_{j}}{\sup }k_{i_{j}}\neq \infty $ for only a finite number of $j\geq 0,$
			which implies that
			
			\begin{equation*}
				A=\sup \left\{ \underset{i_{j}\in I_{j}}{\sup }k_{i_{j}}\mid j\geq 0\text{
					with}\underset{i_{j}\in I_{j}}{\sup }k_{i_{j}}\neq \infty \right\} .
			\end{equation*}
			
			Moreover,
			\\
			
			$B=\underset{k_{i_{j}}\neq -\infty ,j\geq 0}{\underset{\left( i_{j}\right)
					_{j}\in I_{0}\times I_{1}\times \ldots }{\sup }}\left( \text{limsup}\left(
			k_{i_{j}}\right) _{j\in
				\mathbb{N}
			}\right)\\ =\underset{k_{i_{j}}\neq -\infty ,j\geq 0}{\underset{\left(
					i_{j}\right) _{j}\in I_{0}\times I_{1}\times \ldots }{\sup }}\left( \sup
			\left\{ k_{i_{j}}\mid j\geq 0\text{ with }k_{i_{j}}\neq \infty \text{ }%
			\right\} \right) \newline
			=\sup \left\{ \underset{i_{j}\in I_{j}}{\sup }k_{i_{j}}\mid j\geq 0\text{
				with}\underset{i_{j}\in I_{j}}{\sup }k_{i_{j}}\neq \infty \right\} =A$ \\\\where
			the second equality holds since for all $j\geq 0,$ it holds $k_{i_{j}}\in
			L\backslash \left\{ \infty ,-\infty \right\} $ for all $i_{j}\in I_{j},$ or $%
			k_{i_{j}}\in \left\{ \infty ,-\infty \right\} $ for all $i_{j}\in I_{j}$.
			
			\item Assume that there exists infinitely many $j\geq 0$ such that $%
			k_{i_{j}}\neq \infty $ for all $i_{j}\in I_{j}.$ Since $I_{k}$ are finite
			for every $h\geq 0,$ it holds $\underset{i_{h}\in I_{h}}{\sup }\left(
			k_{i_{h}}\right) \in \left\{ k_{i_{h}}\mid i_{h}\in I_{h}\right\} $ for
			every $h\geq 0,$ hence there exist a sequence $\left( i_{j}\right) _{j}\in
			I_{0}\times I_{1}\times I_{2}\times \ldots $ such that $A=$limsup$\left(
			k_{i_{j}}\right) _{j\geq 0},$ and thus $A\leq B.$
			
			Let now $\left( i_{j}\right) _{j}\in I_{0}\times I_{1}\times I_{2}\times
			\ldots $,\footnote{Without any loss we assume that $k_{i_{j}}\neq -\infty$ for all $j\geq 0$.} then
			\begin{eqnarray*}
				\text{limsup}\left( k_{i_{j}}\right) _{j\in
					\mathbb{N}
				} &=&\underset{j\geq 0}{\text{inf}}\left( \sup \left\{ k_{i_{l}}\mid
				l\geq j,k_{i_{l}}\neq \infty \right\} \right) \\
				&\leq &\underset{j\geq 0}{\text{inf}}\left( \sup \left\{ \underset{i_{j}\in I_{j}}{\sup }%
				k_{i_{l}}\mid l\geq j, \underset{i_{l}\in I_{l}}{\sup }%
				k_{i_{l}}\neq \infty \right\} \right) \\
				&=&\text{limsup}\left( \underset{i_{j}\in I_{j}}{\sup }k_{i_{j}}\right)
				_{j\in
					\mathbb{N}
				}=A
			\end{eqnarray*}%
			where the inequality holds due to the following: for every $j\geq 0,$ it holds $k_{i_{j}}\in
			L\backslash \left\{ \infty ,-\infty \right\} $ for all $i_{j}\in I_{j},$ or $%
			k_{i_{j}}\in \left\{ \infty ,-\infty \right\} $ for all $i_{j}\in I_{j}$ , thus for
			every $j\geq 0,$ $k_{i_{j}}\neq \infty ,$ implies
			$\underset{i_{j}\in I_{j}}{\sup }k_{i_{j}}\neq \infty ,$ and vice-versa. So, $A=B$ as desired.
		\end{enumerate}
			
		\noindent 	We prove now Property \ref{Property 1}. We will prove that for all $k,k_{i}\left( i\geq
			1\right) \in \overline{%
				\mathbb{Q}
			},$%
			\begin{equation*}
				\text{limsup}\left( \infty \mathbf{,}k_{1},k_{2},\ldots \right) =\text{limsup%
				}\left( k_{1},k_{2},\ldots \right) .
			\end{equation*}
			First we assume that $\exists i\geq 1,$ such that $k_{i}=-\infty. $ Then,
			
			\begin{equation*}
				\text{ limsup}%
				\left( \infty \mathbf{,}k_{1},k_{2},\ldots \right) =\text{limsup}\left(
				k_{1},k_{2},\ldots \right) =-\infty .
			\end{equation*}
			Otherwise we point out the following
			cases.
			\begin{enumerate}[label=$-$]
			\item If $\forall i\geq 1,k_{i}=\infty ,$ limsup$\left( \infty \mathbf{,}%
			k_{1},k_{2},\ldots \right) =$limsup$\left( k_{1},k_{2},\ldots \right)
			=\infty .$
			\item If there exist infinitely many $i\geq 1$ such that $k_{i}\neq \infty ,$
			then \linebreak limsup$\left( \infty \mathbf{,}k_{1},\ldots \right) =\underset{%
				j\geq 1}{\text{inf}}\left( \sup \left\{ k_{i_{l}}\mid l\geq
			j,k_{i_{l}}\neq \infty \right\} \right) =$limsup$\left( k_{1},\ldots
			\right) $.
			\item Finally, if there exist finitely many $i\geq 1$ such that $k_{i}\neq
			\infty ,$ then limsup$\left( \infty \mathbf{,}k_{1},k_{2},\ldots \right)
			=\sup \left\{ k_{i_{j}}\mid j\geq 1,k_{i_{j}}\neq \infty \right\} =$limsup$%
			\left( k_{1},k_{2},\ldots \right) .$
		\end{enumerate}

			\bigskip We prove now Property \ref{Property 5}. By definition it holds $k\leq \infty $ for
			every $k\in \overline{%
				\mathbb{Q}
			}.$
			
			\bigskip
			We prove now Property \ref{Property 6}. Let $k,$ $k_{i}\in  \overline{%
				\mathbb{Q}
			}$ ($i\geq 0$) such that for every $ i\geq 0, k_{i}\geq
			k .$ We point out the cases $k=\infty ,k=-\infty $ or $%
			k\neq \infty ,-\infty .$ If $k=\infty ,$ then for all $i\geq 0,$ it holds $%
			k_{i}=\infty ,$ which implies that
			limsup$\left( k_{0},k_{1},k_{2},\ldots \right) =\infty \geq k$ as
			wanted. Assume now that $k=-\infty ,$ then it trivially holds limsup$\left(
			k_{0},k_{1},k_{2},\ldots \right) \geq -\infty .$ Finally, if $k\neq \infty
			,-\infty ,$ then $k_{i}\neq -\infty $ for all $i\geq 0,$ and thus limsup$%
			\left( k_{0},k_{1},k_{2},\ldots \right) =k_{i}\geq k$ for some $i\geq 0.$%
			\footnote{Observe that limsup is defined over $\left(\overline{\mathbb{Q}}_{fin}\right)^{\omega}$, i.e., for some finite subset $B$ of $\overline{\mathbb{Q}},$ $k_{i} \in B,$ for every $i\geq 0$.}
			
		\end{proof}
		
	\bigskip
		
		\begin{proof} [Proof of Example \protect\ref{Example 3}]
			\bigskip
			
			We will prove now that $K_{3}$ is an idempotent TGP-$%
			\omega $-valuation monoid that satisfy Properties \ref{Property 1}, \ref{Property 3}, \ref{Property 5}, and \ref{Property 6}. Property \ref{Property 3} follows directly from the definition of $K_{3}$.
			
			We will first prove Property \ref{Property Dist} for TGP-$\omega $%
			-valuation monoids. Let $L\subseteq
			_{fin}\overline{%
				\mathbb{Q}
			},$ finite index sets $I_{j},$ and $k_{i_{j}}\in L\left( i_{j}\in
			I_{j}\right) $ such that for all $j\geq 0,$ it holds $k_{i_{j}}\in
			L\backslash \left\{ \infty ,\mathbf{-\infty }\right\} $ for all $i_{j}\in
			I_{j},$ or $k_{i_{j}}\in \left\{ \mathbf{\infty },\mathbf{-\infty }\right\} $
			for all $i_{j}\in I_{j}$. We will prove that
			\begin{equation*}
				\sup_{-\infty }\left( \underset{i_{j}\in I_{j}}{\sup }k_{i_{j}}\right) =%
				\underset{\left( i_{j}\right) _{j}\in I_{0}\times I_{1}\times \ldots }{\sup }%
				\left( \sup_{-\infty }\left( k_{i_{j}}\right) _{j\in
					\mathbb{N}
				}\right) .
			\end{equation*}%
			We set $A=\sup_{-\infty }\left( \underset{i_{j}\in I_{j}}{\sup }%
			k_{i_{j}}\right) ,$ and $B=\underset{\left( i_{j}\right) _{j}\in I_{0}\times
				I_{1}\times \ldots }{\sup }\left( \sup_{-\infty }\left( k_{i_{j}}\right)
			_{j\in
				\mathbb{N}
			}\right) .$ We will prove that $A=B.$
			
			- Assume that there exists a $j\geq 0,$ such that $k_{i_{j}}=-\infty $ for
			all $i_{j}\in I_{j}.$ Then, $A=B=-\infty .$
			
			- Assume that for all $j\geq 0,$ $\exists k_{i_{j}}\in I_{j\text{ }}$ with $%
			k_{i_{j}}\neq -\infty ,$ and for all $j\geq 0$ there exists an $i_{j}\in
			I_{j}$ such that $k_{i_{j}}=\infty .$ Then, using the same arguments as in
			(I) in the previous proof we get that $A=B=\infty .$
			
			- Finally we assume that for all $j\geq 0,$ $\exists k_{i_{j}}\in I_{j\text{
			}}$ with $k_{i_{j}}\neq -\infty ,$ and there exist at least one $j\geq 0$
			with $k_{i_{j}}\neq \infty $ for all $i_{j}\in I_{j}.$ Then, $%
			\sup_{i_{j}}\left( k_{i_{j}}\right) \neq -\infty $ for all $j\geq 0,$ and
			since there exists $j\geq 0$ with $\sup_{i_{j}\in I_{j}}\left(
			k_{i_{j}}\right) \neq \infty $ we~have
			
			\begin{equation*}
				A=\sup \left\{ \underset{i_{j}\in
					I_{j}}{\sup }k_{i_{j}}\mid j\geq 0,\underset{i_{j}\in I_{j}}{\sup }%
				k_{i_{j}}\neq \infty \right\} ,
			\end{equation*}
			and
			\begin{equation*}
				B=\underset{\left( i_{j}\right)
					_{j}\in I_{0}\times I_{1}\times \ldots }{\sup }\left( \sup \left\{
				k_{i_{j}}\mid j\geq 0,k_{i_{j}}\neq \infty \right\} \right) .
			\end{equation*}
			With arguments similar with the ones used in case (III) in the previous proof we get that it holds $%
			A=B.$\\
			
			Property \ref{Property 1} is proved with similar arguments used in the proof of the corresponding property for Example \ref{Example 2}.
			We prove now Property \ref{Property 5}. By definition it holds $k\leq \infty$ for
			every $k\in \overline{%
				\mathbb{Q}
			}.$
			
			Finally, we prove Property \ref{Property 6}. Let $k,$ $k_{i}\in \overline{%
				\mathbb{Q}
			} \left( i\geq 0\right)$ such that for every $i\geq 0 $ $k_{i}\geq
			k .$ We point out the cases $k=\infty ,k=-\infty $ or $%
			k\neq \infty ,-\infty .$ If $k=\infty ,$ then for all $i\geq 0,$ it holds $%
			k_{i}=\infty ,$ which implies that $\sup_{-\infty }\left(
			k_{0},k_{1},k_{2},\ldots \right) =\infty \geq k$ as desired. Assume now that $%
			k=-\infty ,$ then it trivially holds $\sup_{-\infty }\left(
			k_{0},k_{1},k_{2},\ldots \right) \geq -\infty .$ Finally, if $k\neq \infty
			,-\infty ,$ then $k_{i}\neq -\infty $ for all $i\geq 0,$ and thus $%
			\sup_{-\infty }\left( k_{0},k_{1},k_{2},\ldots \right) =k_{i}\geq k$ for
			some $i\geq 0.$
		\end{proof}
		
		\bigskip
		\bigskip
		
		\begin{proof}[Proof of Theorem \protect\ref{From formulas to automata-tr copy(1)}]
			We follow the inductive proof of \cite[Theorem~1]{Ma-At}, and then conclude our claim by the generalization of \cite[Lemma 5]{Ma-At} for idempotent ordered TGP-$\omega$-valuation monoids. We only present
			some clarifications on why the arguments in the inductive proof of \cite[Theorem~1]{Ma-At} can be adopted.

			In the proof of \cite[Lemma 8, page 253]{Ma-At}, we justify the
			equality
			\begin{eqnarray*}
				&&Val^{\omega }\left( k_{1}+k_{2},wt\left( B_{\varphi _{re}^{1}},\pi
				_{1},B_{\varphi _{2}}\right) ,wt\left( B_{\varphi _{re}^{2}},\pi
				_{2},B_{\varphi ^{2}}\right) ,\ldots \right)  \\
				&=&Val^{\omega }\left( k_{1},wt\left( B_{\varphi _{re}^{1}},\pi
				_{1},B_{\varphi _{2}}\right) ,wt\left( B_{\varphi _{re}^{2}},\pi
				_{2},B_{\varphi ^{2}}\right) ,\ldots \right)  \\
				&&+Val^{\omega }\left( k_{2},wt\left( B_{\varphi _{re}^{1}},\pi
				_{1},B_{\varphi _{2}}\right) ,wt\left( B_{\varphi _{re}^{2}},\pi
				_{2},B_{\varphi ^{2}}\right) ,\ldots \right)
			\end{eqnarray*}%
			as follows: Due to the restriction on the fragment, $k_{i}\neq \mathbf{0}$
			implies $k_{i}\in K\backslash \left\{ \mathbf{0,1}\right\} $ for $i\in
			\left\{ 1,2\right\} ,$ and thus it suffices to point out the following cases:
			
			\begin{case}
				If $k_{1}=\mathbf{0,}$ and $k_{2}\neq \mathbf{0,}$ then
				\begin{equation*}
					k_{1}+k_{2}=k_{2}\text{ and }Val^{\omega }\left( k_{1},wt\left( B_{\varphi
						_{re}^{1}},\pi _{1},B_{\varphi _{2}}\right) ,wt\left( B_{\varphi
						_{re}^{2}},\pi _{2},B_{\varphi ^{2}}\right) ,\ldots \right) =\mathbf{0,}
				\end{equation*}
				which implies
				\begin{eqnarray*}
					&&Val^{\omega }\left( k_{1}+k_{2},wt\left( B_{\varphi _{re}^{1}},\pi
					_{1},B_{\varphi _{2}}\right) ,wt\left( B_{\varphi _{re}^{2}},\pi
					_{2},B_{\varphi ^{2}}\right) ,\ldots \right)  \\
					&=&Val^{\omega }\left( k_{2},wt\left( B_{\varphi _{re}^{1}},\pi
					_{1},B_{\varphi _{2}}\right) ,wt\left( B_{\varphi _{re}^{2}},\pi
					_{2},B_{\varphi ^{2}}\right) ,\ldots \right)  \\
					&=&Val^{\omega }\left( k_{1},wt\left( B_{\varphi _{re}^{1}},\pi
					_{1},B_{\varphi _{2}}\right) ,wt\left( B_{\varphi _{re}^{2}},\pi
					_{2},B_{\varphi ^{2}}\right) ,\ldots \right)  \\
					&&+Val^{\omega }\left( k_{2},wt\left( B_{\varphi _{re}^{1}},\pi
					_{1},B_{\varphi _{2}}\right) ,wt\left( B_{\varphi _{re}^{2}},\pi
					_{2},B_{\varphi ^{2}}\right) ,\ldots \right)
				\end{eqnarray*}
			\end{case}
			
			\begin{case}
				If $k_{1}\neq \mathbf{0,}$ and $k_{2}=\mathbf{0,}$ we prove the equality
				with the same arguments as in the previous case.
			\end{case}
			
			\begin{case}
				If $k_{1}=k_{2}=\mathbf{0,}$ then both parts of the equality are equal to $%
				\mathbf{0.}$
			\end{case}
			
			\begin{case}
				If  $k_{1},k_{2}\in K\backslash \left\{ \mathbf{0,1}\right\} ,$ the equality
				is obtained by the distributivity of $Val^{\omega }$ over finite sums for
				TGP-$\omega $-valuation monoids (Property 1).
			\end{case}
			
			Moreover, in the proof of \cite[Lemma 8, page 254]{Ma-At}, we obtain
			the inequality
			\begin{equation*}
				weight_{\mathcal{A}_{\varphi}}\left(P_{w}\right)\geq weight_{\mathcal{A}_{\psi}}\left(P^{1}_{w}\right)
			\end{equation*}%
			by Lemma \ref{Valuation inequality} (of the current work). The use of Lemma \ref{Valuation inequality} is possible since the
			following hold: In the case $\psi^{1}\in next\left(M_{B^{\prime}_{\varphi},\xi}\right)$, due to the restrictions on the fragment it holds : $%
			wt_{1}\left( B_{\psi }^{\prime },\pi _{0},B_{\psi ^{1}}\right) \neq \mathbf{0%
			}$ implies $wt_{1}\left( B_{\psi }^{\prime },\pi _{0},B_{\psi ^{1}}\right)
			\in K\backslash \left\{ \mathbf{0,1}\right\} ,$ and  $v_{M_{B_{\varphi
						^{\prime }},\xi} }\left( \xi ^{1}\right) \neq \mathbf{0}$ implies $%
			v_{M_{B_{\varphi ^{\prime }},\xi} }\left( \xi ^{1}\right) \in K\backslash
			\left\{ \mathbf{0,1}\right\} .$ Also, since $K$ is totally ordered, it holds $%
			wt_{1}\left( B_{\psi }^{\prime },\pi _{0},B_{\psi ^{1}}\right)
			+v_{M_{B_{\varphi ^{\prime }},\xi }}\left( \xi ^{1}\right) \in \left\{
			wt_{1}\left( B_{\psi }^{\prime },\pi _{0},B_{\psi ^{1}}\right)
			,v_{M_{B_{\varphi ^{\prime }},\xi} }\left( \xi ^{1}\right) \right\} .$
			
			In the proofs of \cite[Lemmas 9 and 13]{Ma-At}, we use Lemma \ref{Valuation
				inequality} (of the current work), instead of of \cite[ Lemma 3ii]{Ma-At}.
			
			From the constructive proofs of \cite[Lemma 12]{Ma-At} (resp. \cite[Lemma 14]{Ma-At}),
			we can further derive that for $\psi \in bLTL\left( K,AP\right) ,$ and $\xi
			\in r$-$stLTL\left( K,AP\right) ,$ or $\xi =\lambda U\zeta ,$ or $\xi
			=\square \zeta $ with $\lambda ,\zeta \in r$-$stLTL\left( K,AP\right) ,$ it holds $%
			v_{B_{\varphi}}\left( \varphi ^{\prime }\right) =v_{B_{\psi}}\left( \psi^{\prime
			}\right) \cdot v_{B_{\xi}}\left( \xi^{\prime}\right) =\mathbf{1},$
			or ($v_{B_{\varphi}}\left( \varphi ^{\prime }\right) \neq \mathbf{1}$ and $v_{B_{\psi
			}}\left( \psi ^{\prime }\right) \cdot v_{B_{\xi }}\left( \xi
			^{\prime }\right) \neq \mathbf{1})$ (resp. for $\psi ,\xi \in r$-$stLTL\left(
			K,AP\right) $, we can further conclude that $v_{B_{\varphi }}\left( \varphi
			^{\prime }\right) =v_{B_{\psi }}\left( \psi ^{\prime }\right) \cdot \underset%
			{1\leq j\leq k}{\prod }v_{B_{\xi _{j}}}\left( \xi _{j}^{\prime }\right) =%
			\mathbf{1},$ or ($v_{B_{\varphi }}\left( \varphi ^{\prime }\right) \neq
			\mathbf{1}$ and $v_{B_{\psi }}\left( \psi ^{\prime }\right) \cdot \underset{%
				1\leq j\leq k}{\prod }v_{B_{\xi _{j}}}\left( \xi _{j}^{\prime }\right) \neq
			\mathbf{1}$).
			
			In the proof of \cite[Lemma 16, page 271]{Ma-At}, the last equality
			is true because of the following: For every $i\geq 0$, it holds
			\begin{eqnarray*}
				&&\underset{k_{i}\in pri_{\mathcal{A}_{\xi }}\left( w_{\geq i}\right) }{%
					\underset{k_{j}\in pri_{\mathcal{A}_{\psi }}\left( w_{\geq j}\right) \left(
						0\leq j<i\right) }{\sum }}Val^{\omega }\left( k_{0},\ldots ,k_{i-1},k_{i},%
				\mathbf{1,1,1,\ldots }\right) \\
				&=&\underset{k_{i}\in pri_{\mathcal{A}_{\xi }}\left( w_{\geq i}\right)
					\backslash \left\{ \mathbf{0}\right\} }{\underset{k_{j}\in pri_{\mathcal{A}%
							_{\psi }}\left( w_{\geq j}\right) \backslash \left\{ \mathbf{0}\right\}
						\left( 0\leq j<i\right) }{\sum }}Val^{\omega }\left( k_{0},\ldots
				,k_{i-1},k_{i},\mathbf{1,1,1,\ldots }\right) \\
				&=&Val^{\omega }\left(
				\begin{array}{c}
					\underset{k_{0}\in pri_{\mathcal{A}_{\psi }}\left( w_{\geq 0}\right)
						\backslash \left\{ \mathbf{0}\right\} }{\sum }k_{0},\ldots ,\underset{%
						k_{i-1}\in pri_{\mathcal{A}_{\psi }}\left( w_{\geq i-1}\right) \backslash
						\left\{ \mathbf{0}\right\} }{\sum }k_{i-1}, \\
					\underset{k_{i}\in pri_{\mathcal{A}_{\xi }}\left( w_{\geq i}\right)
						\backslash \left\{ \mathbf{0}\right\} }{\sum }k_{i},\mathbf{1,1,1,\ldots }%
				\end{array}%
				\right) \\
				&=&Val^{\omega }\left(
				\begin{array}{c}
					\underset{k_{0}\in pri_{\mathcal{A}_{\psi }}\left( w_{\geq 0}\right) }{\sum }%
					k_{0},\ldots ,\underset{k_{i-1}\in pri_{\mathcal{A}_{\psi }}\left( w_{\geq
							i-1}\right) }{\sum }k_{i-1}, \\
					\underset{k_{i}\in pri_{\mathcal{A}_{\xi }}\left( w_{\geq i}\right) }{\sum }%
					k_{i},\mathbf{1,1,1,\ldots }%
				\end{array}%
				\right)
			\end{eqnarray*}%
			where the first equality is obtained by the fact that $Val_{\omega}\left(\left(d_{i}\right)_{i\geq 0}\right)=\mathbf{0}$ whenever $d_{i}=
			\mathbf{0}$ for some $i\geq 0$, and by the completeness axioms of the monoid. The second equality is now obtained by the distributivity of $%
			Val^{\omega }$ over finite sums for TGP-$\omega $-valuation monoids
			(Property 1). Observe that since $\psi ,\xi \in r$-$stLTL\left(
			K,AP\right) $ it holds: (i) $k_{j}\in pri_{\mathcal{A}_{\psi }}\left(
			w_{\geq j}\right) \backslash \left\{ \mathbf{0}\right\} $ implies $k_{j}\in
			K\backslash \left\{ \mathbf{0,1}\right\} $ for every $0\leq j<i,$ and (ii) $%
			k_{i}\in pri_{\mathcal{A}_{\xi }}\left( w_{\geq i}\right) \backslash \left\{
			\mathbf{0}\right\} $ implies $k_{i}\in K\backslash \left\{ \mathbf{0,1}%
			\right\} .$
			
			In the proof of \cite[Lemma 16, page 274]{Ma-At}, we use Lemma \ref{Valuation inequality} of the
			current work instead of  of \cite[Lemma 3]{Ma-At}. The use of Lemma \ref{Valuation inequality} is
			possible because of the following: By the fact that $\psi ,\xi \in r$-$%
			stLTL(K,AP)$, and since $K$ is totally ordered, we can further conclude that
			the sets $\left\{ wt_{1}\left( B_{\psi _{re}^{\left( m,0\right) }},\pi
			_{m},B_{\psi ^{\left( m,1\right) }}\right) ,wt\left( B_{\varphi
				_{re}^{m}},\pi _{m},B_{\varphi ^{m+1}}\right) \right\} $ (where $ 0\leq m<l$) and $\left\{
			wt_{2}\left( B_{\xi _{re}^{0}},\pi _{l},B_{\xi ^{1}}\right) ,wt\left(
			B_{\varphi _{re}^{l}},\pi _{l},B_{\varphi ^{l+1}}\right) \right\} $ are
			subsets of $K\backslash \left\{ \mathbf{0,1}\right\}$, and since $%
			\varphi _{re}^{k}$ is boolean for every $k>l$ , we also get that $wt\left(
			B_{\varphi _{re}^{k}},\pi _{k},B_{\varphi ^{k+1}}\right) =\mathbf{1.}$
			In the proof of \cite[Lemma 16, page 275]{Ma-At}, the second implication is justified with the same arguments we used for the justification of the equality in \cite[page 275]{Ma-At}.
			
			In the proof of \cite[Lemma 17, page 276]{Ma-At}, we use Lemma \ref{Valuation inequality}, instead of
			\cite[Lemma 3ii]{Ma-At}, since $\psi \in r$-$stLTL\left( K,AP\right)$ and $K$ is totally ordered, we have
			that $pri_{\mathcal{A}_{\psi }}\left(
			w_{\geq i}\right) $ is finite for every $i\geq 0,$ and $\underset{k_{i}\in pri_{\mathcal{A}_{\psi }}\left( w_{\geq
					i}\right) }{\sum }k_{i}\neq \mathbf{1}$ with $\underset{k_{i}\in pri_{\mathcal{A}_{\psi }}\left( w_{\geq
					i}\right) }{\sum }k_{i}=\underset{k_{i}\in pri_{\mathcal{A}_{\psi }}\left(
				w_{\geq i}\right) \backslash \left\{ \mathbf{0}\right\} }{\sum k_{i}}\in
			pri_{\mathcal{A}_{\psi }}\left( w_{\geq i}\right) \backslash \left\{ \mathbf{%
				0}\right\} $. We use similar arguments to justify the use of Lemma \ref{Valuation inequality} (of the current work) and of the Distributivity of $Val^{\omega}$ over finite sums for TGP-$\omega$-valuation monoids (Property 1) in the rest of the proof.
			
			For the Definitions of $v_{B_{\varphi }}\left( \varphi \right) ,$ and $pri_{%
				\mathcal{M}}\left( w\right) $, where ($\varphi \in LTL\left( K,AP\right) ,$ $%
			\mathcal{M}$ is an $\epsilon$-wgBa, $w\in \left( \mathcal{P}\left( AP\right) \right)
			^{\omega }$), we refer the reader to \cite{Ma-At}.
		\end{proof}
		\bigskip
		
		\begin{proof} [Proof of Lemma \protect\ref{Lemma 2}]
			
			Let $v\in Q,$ and $\mathcal{M}^{\geq v}=\left(
			Q^{\geq v},A,q_{0}^{\geq v},\Delta ^{\geq v},F^{\geq v}\right) $ where
			
			\begin{itemize}
				\item $Q^{\geq v}=\left( Q\times \left\{ 0,1,3,4\right\} \times \left\{
				B,C\right\} \right) \cup \left( \widehat{Q}\times \left\{ 2,5,6\right\}
				\times \left\{ B,C\right\} \right) $
				
				where $\widehat{Q}=\left\{ s_{q}\mid q\in Q\right\} $ is a set of copies of
				the elements of $Q$
				
				\item $q_{0}^{\geq v}=\left( q_{0},0,C\right) $
				
				\item For all $\left( \left( q,i,k\right) ,a,\left( \overline{q}%
				,j,l\right) \right) \in Q^{\geq v}\times A\times Q^{\geq v}$, it holds $%
				\left( \left( q,i,k\right) ,a,\left( \overline{q},j,l\right) \right)\\ \in
				\Delta ^{\geq v}$ iff
				
				\begin{enumerate}[label=\arabic*.]
				\item $i=0$ and $j=1,$ $k=l=C,$ and $q=q_{0},\overline{q}\in Q\backslash F$ and $%
				wt\left( q,a,\overline{q}\right) <v$ with $wt\left( q,a,\overline{q}\right)
				\neq -\infty ,$ or
				
				\item  $i=0$ and $j=1,$ $k=C,l=B$ and $q=q_{0},\overline{q}\in F$ and $wt\left( q,a,%
				\overline{q}\right) <v$ with $wt\left( q,a,\overline{q}\right) \neq -\infty
				, $ or
				
				\item  $i=0$ and $j=3,k=l=C,$ and $q=q_{0},\overline{q}\in Q\backslash F,wt\left(
				q,a,\overline{q}\right) \geq v,$ and $wt\left( q,a,\overline{q}\right) \neq
				\infty ,$ or
				
				\item  $i=0$ and $j=3,k=C,l=B$, and $q=q_{0},\overline{q}\in F,wt\left( q,a,%
				\overline{q}\right) \geq v,$ and $wt\left( q,a,\overline{q}\right) \neq
				\infty ,$ or
				
				\item  $i=0$ and $j=4,k=l=C,$ and $q=q_{0},\overline{q}\in Q\backslash F,wt\left(
				q,a,\overline{q}\right) =\infty ,$ or
				
				\item  $i=0$ and $j=4,k=C,l=B,$ and $q=q_{0},\overline{q}\in F,wt\left( q,a,%
				\overline{q}\right) =\infty ,$ or
				
				\item  $i=0$ and $j=2,$ $k=l=C,$ and $q=q_{0},\overline{q}=s_{\widetilde{q}}\in
				\widehat{Q},\widetilde{q}\notin F,wt\left( q_{0},a,\widetilde{q}\right)
				=\infty ,$ or
				
				\item  $i=0$ and $j=2,$ $k=C,l=B,$ and $q=q_{0},\overline{q}=s_{\widetilde{q}}\in
				\widehat{Q},\widetilde{q}\in F,wt\left( q_{0},a,\widetilde{q}\right) =\infty
				,$ or
				
				\item  $i=2$ and $j=2,$ $k\in \left\{ C,B\right\} ,l=B,$ and $q=s_{q^{\prime \prime
				}},\overline{q}=s_{\widetilde{q}},\widetilde{q}\in F,$ $wt\left( q^{\prime
					\prime },a,\widetilde{q}\right) =\infty ,$ or
				
				\item  $i=2$ and $j=2,$ $k\in \left\{ C,B\right\} ,l=C,$ and $q=s_{q^{\prime \prime
				}},\overline{q}=s_{\widetilde{q}},\widetilde{q}\notin F,$ $wt\left(
				q^{\prime \prime },a,\widetilde{q}\right) =\infty ,$ or
				
				\item  $i=1,3,4$ and $j=3,$ $k\in \left\{ C,B\right\} ,l=B,q\in Q,\overline{q}\in
				F, $ and \ $wt\left( q,a,\overline{q}\right) \geq v,$ and $wt\left( q,a,%
				\overline{q}\right) \neq \infty ,$ or
				
				\item  $i=1,3,4$ and $j=3,$ $k\in \left\{ C,B\right\} ,l=C,q\in Q,\overline{q}\in
				Q\backslash F,$ and \ $wt\left( q,a,\overline{q}\right) \geq v,$ and $%
				wt\left( q,a,\overline{q}\right) \neq \infty ,$ or
				
				\item  $i=1,3,4$ and $j=1,$ $k\in \left\{ C,B\right\} ,l=B,q\in Q,\overline{q}\in
				F, $ and $wt\left( q,a,\overline{q}\right) <v,wt\left( q,a,\overline{q}%
				\right) \neq -\infty ,$ or
				
				\item  $i=1,3,4$ and $j=1,$ $k\in \left\{ C,B\right\} ,l=C,q\in Q,\overline{q}\in
				Q\backslash F,$ and $wt\left( q,a,\overline{q}\right) <v,wt\left( q,a,%
				\overline{q}\right) \neq -\infty ,$ or
				
				\item  $i=1,3,4$ and $j=4,$ $k\in \left\{ C,B\right\} ,l=C,q\in Q,\overline{q}\in
				Q\backslash F,$ and $wt\left( q,a,\overline{q}\right) =\infty ,$ or
				
				\item  $i=1,3,4$ and $j=4,$ $k\in \left\{ C,B\right\} ,l=B,q\in Q,\overline{q}\in
				F, $ and $wt\left( q,a,\overline{q}\right) =\infty ,$ or
				
				\item  $i=3,$ and $j=5,$\ $k\in \left\{ C,B\right\} ,l=C,$ and $q\in Q,\overline{q}%
				=s_{\widetilde{q}}\in \widehat{Q},\widetilde{q}\notin F,$ $wt\left( q,a,%
				\widetilde{q}\right) \neq -\infty ,$ or
				
				\item  $i=3,$ and $j=5,$\ $k\in \left\{ C,B\right\} ,l=B,$ and $q\in Q,\overline{q}%
				=s_{\widetilde{q}}\in \widehat{Q},\widetilde{q}\in F,$ $wt\left( q,a,%
				\widetilde{q}\right) \neq -\infty ,$ or
				
				\item  $i=5,$ and $j=5,$ \ $k\in \left\{ C,B\right\} ,l=C,$ and $q,\overline{q}\in
				\widehat{Q},q=s_{q^{\prime \prime }},\overline{q}=s_{\widetilde{q}},%
				\widetilde{q}\notin F,wt\left( q^{\prime \prime },a,\widetilde{q}\right)
				\neq -\infty ,$ or
				
				\item  $i=5,$ and $j=5,$ \ $k\in \left\{ C,B\right\} ,l=B,$ and $q,\overline{q}\in
				\widehat{Q},q=s_{q^{\prime \prime }},\overline{q}=s_{\widetilde{q}},%
				\widetilde{q}\in F,wt\left( q^{\prime \prime },a,\widetilde{q}\right) \neq
				-\infty ,$ or
				
				\item  $i=5,$ and $j=6,$ \ $k\in \left\{ C,B\right\} ,l=C,$ and $q=s_{q^{\prime
						\prime }},\overline{q}=s_{\widetilde{q}},\widetilde{q}\notin F,wt\left(
				q^{\prime \prime },a,\widetilde{q}\right) =\infty ,$ or
				
				\item  $i=6,$ and $j=6,$ \ $k\in \left\{ C,B\right\} ,l=C,$ and $q=s_{q^{\prime
						\prime }},\widetilde{q}\notin F,\overline{q}=s_{\widetilde{q}},wt\left(
				q^{\prime \prime },a,\widetilde{q}\right) =\infty ,$ or
				
				\item  $i=5,$ and $j=6,$ \ $k\in \left\{ C,B\right\} ,l=B,$ and $q=s_{q^{\prime
						\prime }},\overline{q}=s_{\widetilde{q}},\widetilde{q}\in F,wt\left(
				q^{\prime \prime },a,\widetilde{q}\right) =\infty ,$ or
				
				\item  $i=6,$ and $j=6,$ \ $k\in \left\{ C,B\right\} ,l=B,$ and $q=s_{q^{\prime
						\prime }},\overline{q}=s_{\widetilde{q}},\widetilde{q}\in F,wt\left(
				q^{\prime \prime },a,\widetilde{q}\right) =\infty .$
			\end{enumerate}
				
				\item The set of final sets is defined as follows
				
				$F^{\geq v}=\left\{ \widehat{F}\mid \widehat{F}\subseteq F_{1}\cup F_{2}\cup
				F_{3},F_{1}\cap \widehat{F}\neq \emptyset ,F_{2}\cap \widehat{F}\neq
				\emptyset \right\} $
				
				$F_{1}=\left( Q\times \left\{ 3\right\} \times \left\{ C,B\right\} \right)
				\cup \left( \widehat{Q}\times \left\{ 2\right\} \times \left\{ C,B\right\}
				\right) \cup \left( \widehat{Q}\times \left\{ 6\right\} \times \left\{
				C,B\right\} \right) $
				
				$F_{2}=\left( Q\times \left\{ 0,1,3,4\right\} \times \left\{ B\right\}
				\right) \cup \left( \widehat{Q}\times \left\{ 2,5,6\right\} \times \left\{
				B\right\} \right) $
				
				$F_{3}=Q\times \left\{ 1,4\right\} \times \left\{ B,C\right\} $
			\end{itemize}
			
			We prove now that $\left\Vert \mathcal{M}^{\geq v}\right\Vert =\left\{ w\in
			A^{\omega }\mid \left( \left\Vert \mathcal{N}\right\Vert ,w\right) \geq
			v\right\} .$ We let first $w=a_{0}a_{1}a_{2}\ldots \in \left\{ w\in
			A^{\omega }\mid \left( \left\Vert \mathcal{N}\right\Vert ,w\right) \geq
			v\right\} .$ This implies that%
			\begin{equation*}
				\left( \left\Vert \mathcal{N}\right\Vert ,w\right) =\underset{P_{w}\in succ_{%
						\mathcal{N}}\left( w\right) }{\sup }\left( weight_{\mathcal{N}}\left(
				P_{w}\right) \right) =\infty
			\end{equation*}%
			or%
			\begin{equation*}
				\left( \left\Vert \mathcal{N}\right\Vert ,w\right) =\underset{P_{w}\in succ_{%
						\mathcal{N}}\left( w\right) }{\sup }\left( weight_{\mathcal{N}}\left(
				P_{w}\right) \right) =v^{\prime \prime }\in
				\mathbb{Q}
				\text{ with }v^{\prime \prime }\geq v.
			\end{equation*}%
			If the first case is true there exists a successful path $P_{w}=\left(
			q_{i},a_{i},q_{i+1}\right) _{i\geq 0}$ of $\mathcal{N}$ over $w$ such that $%
			wt\left( q_{i},a_{i},q_{i+1}\right) =\infty $ for all $i\geq 0.$ We consider
			the path%
			\begin{equation*}
				P_{w}^{\prime \prime }=\left( \left( q_{0},0,C\right) ,a_{0},\left(
				s_{q_{1}},2,D_{1}\right) \right) \left( \left( s_{q_{i}},2,D_{i}\right)
				,a_{i},\left( s_{q_{i+1}},2,D_{i+1}\right) \right) _{i\geq 1}
			\end{equation*}%
			\noindent where $D_{i}=B$ if $q_{i}\in F,$ and $D_{i}=C$ if $q_{i}\notin F$ \ for
			every $i\geq 1.$ (We shall call such a path $P_{w}^{\prime \prime }$ a path
			of form A). Clearly, $P_{w}^{\prime \prime }$ is a successful path of $%
			\mathcal{M}^{\geq v}$ over $w,$ i.e., $w\in \left\Vert \mathcal{M}^{\geq
				v}\right\Vert $ as desired. Now, if\ the second case is true there exist a
			successful path $P_{w}=\left( q_{i},a_{i},q_{i+1}\right) _{i\geq 0}$ of $%
			\mathcal{N}$ over $w$ such that $wt\left( q_{i},a_{i},q_{i+1}\right) \neq
			\infty $ for at least one $i\geq 0.$ Assume that there exist infinitely many
			$i\geq 0$ such that $wt\left( q_{i},a_{i},q_{i+1}\right) \neq \infty .$
			Then, there are also infinitely many $i\geq 0$ with $wt\left(
			q_{i},a_{i},q_{i+1}\right) =v^{\prime \prime }\geq v,$ and thus we can
			define a successful path $P_{w}^{\prime \prime }$ of $\mathcal{M}^{\geq v}$
			over $w$ of the form $\left( \left( q_{i},j_{i},D_{i}\right) ,a_{i},\left(
			q_{i+1},j_{i+1},D_{i+1}\right) \right) _{i\geq 0}$ with $\left(
			q_{0},j_{0},D_{0}\right) =\left( q_{0},0,C\right) ,$\ and for all $i\geq 1,$
			$j_{i}=1,$ if $j_{i-1}=0,1,3,4$, and $wt\left( q_{i-1},a_{i-1},q_{i}\right)
			<v,$ or $j_{i}=3$ if $j_{i-1}=0,1,3,4$, and $wt\left(
			q_{i-1},a_{i-1},q_{i}\right) \geq v,$ or $j_{i}=4$ if $j_{i-1}=0,1,3,4$, and
			$wt\left( q_{i-1},a_{i-1},q_{i}\right) =\infty ,\ $and $D_{i}=C$ if $%
			q_{i}\notin F,$ and $D_{i}=B$ if $q_{i}\in F.$ Next, we assume that there
			exist only a finite number of $i\geq 0$ with $wt\left(
			q_{i},a_{i},q_{i+1}\right) \neq \infty .$ Then, there exist an $h\geq 0$
			with $wt\left( q_{h},a_{h},q_{h+1}\right) =v^{\prime \prime }\geq v,$ and a $%
			j>h$ with $wt\left( q_{k},a_{k},q_{k+1}\right) =\infty $ for all $k\geq j.$
			Then, we can define a successful path $P_{w}^{\prime \prime }$ of $\mathcal{M%
			}^{\geq v}$ over $w$ as follows%
			\bigskip
			
			$P_{w}^{\prime \prime } =\left( \left( q_{i},j_{i},D_{i}\right)
			,a_{i},\left( q_{i+1},j_{i+1},D_{i+1}\right) \right) _{0\leq i\leq h-1}$
			
			$\left( \left( q_{h},j_{h},D_{h}\right) ,a_{h},\left(
			s_{q_{h+1}},5,D_{h+1}\right) \right) $
			
			$\left( \left( s_{q_{i}},5,D_{i}\right)
			,a_{i},\left( s_{q_{i+1}},5,D_{i+1}\right) \right) _{h+1\leq i\leq j-1}$
			
			$\left( \left( s_{q_{j}},5,D_{j}\right) ,a_{j},\left(
			s_{q_{j+1}},6,D_{j+1}\right) \right)$
			
			$\left( \left( s_{q_{k}},6,D_{k}\right)
			,a_{k},\left( s_{q_{k+1}},6,D_{k+1}\right) \right) _{k\geq j+1},$
			
			\bigskip
			\noindent where $\left( q_{0},j_{0},D_{0}\right) =\left( q_{0},0,C\right) .$ Moreover,
			for all $0\leq i\leq h,$ $j_{i}=1,$ if $j_{i-1}=0,1,3,4$, and $wt\left(
			q_{i-1},a_{i-1},q_{i}\right) <v,$ or $j_{i}=3$ if $j_{i-1}=0,1,3,4$, and $%
			wt\left( q_{i-1},a_{i-1},q_{i}\right) \geq v,$ or $j_{i}=4$ if $%
			j_{i-1}=0,1,3,4$, and $wt\left( q_{i-1},a_{i-1},q_{i}\right) =\infty .$
			Finally, $D_{l}=C$ if $q_{l}\notin F,$ and $D_{l}=B$ if $q_{l}\in F$ for
			every $l\geq 1.$ (We shall call such a path $P_{w}^{\prime \prime }$, a path
			of form B). Hence, $w\in \left\Vert \mathcal{M}^{\geq v}\right\Vert $.
			
			We proved that $\left\{ w\in A^{\omega }\mid \left( \left\Vert \mathcal{N}%
			\right\Vert ,w\right) \geq v\right\} \subseteq \left\Vert \mathcal{M}^{\geq
				v}\right\Vert .$ In order to prove the converse inclusion, we let $%
			P_{w}^{\prime \prime }=\left( \left( q_{i},j_{i},D_{i}\right) ,a_{i},\left(
			q_{i+1},j_{i+1},D_{i+1}\right) \right) _{i\geq 0}$ with $\left(
			q_{0},j_{0},D_{0}\right) =\left( q_{0},0,C\right) $ be a\ successful path of
			$\mathcal{M}^{\geq v}$ over $w$. This implies that $In^{Q}\left(
			P_{w}^{\prime \prime }\right) =\widehat{F}\in F^{\geq v},$ i.e., $\widehat{F}%
			\cap F_{1}\neq \emptyset ,$ and $\widehat{F}\cap F_{2}\neq \emptyset .$
			Moreover, by definition of $\Delta ^{\geq v}$ we have that $F_{1}\cap
			\widehat{F}\subseteq Q\times \left\{ 3\right\} \times \left\{ C,B\right\} $
			(case 1), or $F_{1}\cap \widehat{F}\subseteq \widehat{Q}\times \left\{
			2\right\} \times \left\{ C,B\right\} $ (case 2), or $F_{1}\cap \widehat{F}%
			\subseteq \widehat{Q}\times \left\{ 6\right\} \times \left\{ C,B\right\} $
			(case 3), and $F_{2}\cap \widehat{F}\subseteq Q\times \left\{
			0,1,3,4\right\} \times \left\{ B\right\} ,$ or $\widehat{F}\cap
			F_{2}\subseteq \widehat{Q}\times \left\{ 2,5,6\right\} \times \left\{
			B\right\} .$
			
			Moreover,%
			\begin{equation*}
				\left( \left( \widehat{Q}\times \left\{ 6\right\} \times \left\{ C,B\right\}
				\right) \cap \widehat{F}\neq \emptyset ,\text{ or }\left( \widehat{Q}\times
				\left\{ 2\right\} \times \left\{ C,B\right\} \right) \cap \widehat{F}\neq
				\emptyset \right) ,
			\end{equation*}%
			implies
			\begin{equation*}
				\widehat{F}\cap \left( \widehat{Q}\times \left\{ 2,5,6\right\} \times
				\left\{ B\right\} \right) \neq \emptyset ,
			\end{equation*}
			and vice-versa. We assume first that for $P_{w}^{\prime \prime }$ case (2),
			or case (3) holds, then $P_{w}^{\prime \prime }$ is either of form (A), or
			of form (B). If $P_{w}^{\prime \prime }$ is of form (A), then $P_{w}=\left(
			q_{i},a_{i},q_{i+1}\right) _{i\geq 0}\in succ_{\mathcal{N}}\left( w\right) $
			with $weight_{\mathcal{N}}\left( P_{w}\right) =\infty \geq v.$ If $%
			P_{w}^{\prime \prime }$ is of form (B),\ then $P_{w}=\left(
			q_{i},a_{i},q_{i+1}\right) _{i\geq 0}\in succ_{\mathcal{N}}\left( w\right) $
			with $weight_{\mathcal{N}}\left( P_{w}\right) =v^{\prime }\geq v,v^{\prime
			}\in
			\mathbb{Q}
			.$ If case (1) holds, then $P_{w}^{\prime \prime }$ is of form (A), and $P_{w}=\left(
			q_{i},a_{i},q_{i+1}\right) _{i\geq 0}$ is a successful path of $%
			\mathcal{N}$ over $w$ with $weight_{\mathcal{N}}\left( P_{w}\right)
			=v^{\prime }\geq v,v^{\prime }\in
			\mathbb{Q}
			.$ Thus, $\left\Vert \mathcal{M}^{\geq v}\right\Vert \subseteq \left\{ w\in
			A^{\omega }\mid \left( \left\Vert \mathcal{N}\right\Vert ,w\right) \geq
			v\right\} $.
			
		\end{proof}

		\bigskip
		
		\begin{proof}[Proof of Lemma \protect\ref{Lemma 2 inf}]
			First, we assume that $v\in Q.$ We let $\mathcal{M}^{\geq v}=\left(
			Q^{\geq v},A,q_{0}^{\geq v}\Delta ^{\geq v},\mathcal{R}^{\geq v}\right) $
			where
			
			\begin{itemize}
				\item $Q^{\geq v}=\left( Q\times \left\{ 0,1,3,4\right\} \times \left\{
				B,C\right\} \right) \cup \left( \widehat{Q}\times \left\{ 2,5,6\right\}
				\times \left\{ B,C\right\} \right) $
				
				where $\widehat{Q}=\left\{ s_{q}\mid q\in Q\right\} $ is a set of copies of
				the elements of $Q$
				
				\item $q_{0}^{\geq v}=\left( q_{0},0,C\right) $
				
				\item For all $\left( \left( q,i,k\right) ,a,\left( \overline{q}%
				,j,l\right) \right) \in Q^{\geq v}\times A\times Q^{\geq v}$, it holds $%
				\left( \left( q,i,k\right) ,a,\left( \overline{q},j,l\right) \right) \\\in
				\Delta ^{\geq v}$ iff
				\begin{enumerate}[label=\arabic*.]				
				\item $i=0$ and $j=1,$ $k=l=C,$ and $q=q_{0},\overline{q}\notin F$ and $wt\left(
				q,a,\overline{q}\right) <v$ with $wt\left( q,a,\overline{q}\right) \neq
				-\infty ,$ or
				
				\item $i=0$ and $j=1,$ $k=C,l=B$ and $q=q_{0},\overline{q}\in F$ and $wt\left( q,a,%
				\overline{q}\right) <v$ with $wt\left( q,a,\overline{q}\right) \neq -\infty
				, $ or
				
				\item $i=0$ and $j=3,k=l=C,$ and $q=q_{0},\overline{q}\notin F,wt\left( q,a,%
				\overline{q}\right) \geq v,$ and $wt\left( q,a,\overline{q}\right) \neq
				\infty ,$ or
				
				\item $i=0$ and $j=3,k=C,l=B$, and $q=q_{0},\overline{q}\in F,wt\left( q,a,%
				\overline{q}\right) \geq v,$ and $wt\left( q,a,\overline{q}\right) \neq
				\infty ,$ or
				
				\item $i=0$ and $j=4,k=l=C,$ and $q=q_{0},\overline{q}\notin F,wt\left( q,a,%
				\overline{q}\right) =\infty ,$ or
				
				\item $i=0$ and $j=4,k=C,l=B,$ and $q=q_{0},\overline{q}\in F,wt\left( q,a,%
				\overline{q}\right) =\infty ,$ or
				
				\item $i=0$ and $j=2,$ $k=l=C,$ and $q=q_{0},\overline{q}=s_{\widetilde{q}}\in
				\widehat{Q},\widetilde{q}\notin F,wt\left( q_{0},a,\widetilde{q}\right)
				=\infty ,$ or
				
				\item $i=0$ and $j=2,$ $k=C,l=B,$ and $q=q_{0},\overline{q}=s_{\widetilde{q}}\in
				\widehat{Q},\widetilde{q}\in F,wt\left( q_{0},a,\widetilde{q}\right) =\infty
				,$ or
				
				\item $i=2$ and $j=2,$ $k\in \left\{ C,B\right\} ,l=B,$ and $q=s_{q^{\prime \prime
				}},\overline{q}=s_{\widetilde{q}},\widetilde{q}\in F,$ $wt\left( q^{\prime
					\prime },a,\widetilde{q}\right) =\infty ,$ or
				
				\item $i=2$ and $j=2,$ $k\in \left\{ C,B\right\} ,l=C,$ and $q=s_{q^{\prime \prime
				}},\overline{q}=s_{\widetilde{q}},\widetilde{q}\notin F,$ $wt\left(
				q^{\prime \prime },a,\widetilde{q}\right) =\infty ,$ or
				
				\item $i=1,3,4$ and $j=3,$ $k\in \left\{ C,B\right\} ,l=B,q\in Q,\overline{q}\in
				F, $ and \ $wt\left( q,a,\overline{q}\right) \geq v,$ and $wt\left( q,a,%
				\overline{q}\right) \neq \infty ,$ or
				
				\item $i=1,3,4$ and $j=3,$ $k\in \left\{ C,B\right\} ,l=C,q\in Q,\overline{q}\in
				Q\backslash F,$ and \ $wt\left( q,a,\overline{q}\right) \geq v,$ and $%
				wt\left( q,a,\overline{q}\right) \neq \infty ,$ or
				
				\item $i=1,3,4$ and $j=1,$ $k\in \left\{ C,B\right\} ,l=B,q\in Q,\overline{q}\in
				F, $ and $wt\left( q,a,\overline{q}\right) <v,wt\left( q,a,\overline{q}%
				\right) \neq -\infty ,$ or
				
				\item $i=1,3,4$ and $j=1,$ $k\in \left\{ C,B\right\} ,l=C,q\in Q,\overline{q}\in
				Q\backslash F,$ and $wt\left( q,a,\overline{q}\right) <v,wt\left( q,a,%
				\overline{q}\right) \neq -\infty ,$ or
				
				\item $i=1,3,4$ and $j=4,$ $k\in \left\{ C,B\right\} ,l=C,q\in Q,\overline{q}\in
				Q\backslash F,$ and $wt\left( q,a,\overline{q}\right) =\infty ,$ or
				
				\item $i=1,3,4$ and $j=4,$ $k\in \left\{ C,B\right\} ,l=B,q\in Q,\overline{q}\in
				F, $ and $wt\left( q,a,\overline{q}\right) =\infty ,$ or
				
				\item $i=0,$ and $j=5,$\ $k\in \left\{ C,B\right\} ,l=C,$ and $q=q_{0},\overline{q}%
				=s_{\widetilde{q}}\in \widehat{Q},\widetilde{q}\notin F,$ $wt\left(
				q_{0},a,\widetilde{q}\right) \geq v,$ or

                \item $i=0,$ and $j=5,$\ $k\in \left\{ C,B\right\} ,l=B,$ and $q=q_{0},\overline{q}%
				=s_{\widetilde{q}}\in \widehat{Q},\widetilde{q}\in F,$ $wt\left(
				q_{0},a,\widetilde{q}\right) \geq v,$ or

				\item $i=5,$ and $j=5,$ \ $k\in \left\{ C,B\right\} ,l=C,$ and $q,\overline{q}\in
				\widehat{Q},q=s_{q^{\prime \prime }},\overline{q}=s_{\widetilde{q}},%
				\widetilde{q}\notin F,wt\left( q^{\prime \prime },a,\widetilde{q}\right)
				\geq v,$ or
				
				\item $i=5,$ and $j=5,$ \ $k\in \left\{ C,B\right\} ,l=B,$ and $q,\overline{q}\in
				\widehat{Q},q=s_{q^{\prime \prime }},\overline{q}=s_{\widetilde{q}},%
				\widetilde{q}\in F,wt\left( q^{\prime \prime },a,\widetilde{q}\right) \geq
				v. $
			\end{enumerate}
				\item The family $\mathcal{R}^{\geq v}$ is defined as follows
				
				$\mathcal{R}^{\geq v}=\left\{ \left( \widehat{S},\widehat{F}\right) \mid
				\widehat{F}\in F^{\geq v},\widehat{S}\subseteq Q\times \left\{ 1\right\}
				\times \left\{ B,C\right\} \right\} \cup \left\{ \left( \emptyset ,\overline{%
					F}\right) \mid \overline{F}\in \overline{F}^{\geq v}\right\}\\ \cup \left\{
				\left( \emptyset ,\widetilde{F}\right) \mid \widetilde{F}\in \widetilde{F}%
				^{\geq v}\right\} $ where
				
				$\widehat{F}^{\geq v}=\left\{ \widehat{F}\mid \widehat{F}\subseteq F_{1}\cup
				F_{2},F_{1}\cap \widehat{F}\neq \emptyset ,F_{2}\cap \widehat{F}\neq
				\emptyset \right\} $ with
				
				$F_{1}=Q\times \left\{ 3\right\} \times \left\{ C,B\right\} ,F_{2}=Q\times
				\left\{ 0,1,3,4\right\} \times \left\{ B\right\} ,$ and
				
				$\overline{F}^{\geq v}=\left\{ \overline{F}\mid \overline{F}\subseteq
				\overline{F}_{1}\cup \overline{F}_{2},\overline{F}_{1}\cap \overline{F}\neq
				\emptyset ,\overline{F}_{2}\cap \overline{F}\neq \emptyset \right\} $ with
				
				$\overline{F}_{1}=\widehat{Q}\times \left\{ 2\right\} \times \left\{
				C,B\right\} ,\overline{F}_{2}=\widehat{Q}\times \left\{ 2\right\} \times
				\left\{ B\right\} ,$ and
				
				$\widetilde{F}^{\geq v}=\left\{ \widetilde{F}\mid \widetilde{F}\subseteq
				\widehat{Q}\times \left\{ 5\right\} \times \left\{ B\right\} \right\} .$
			\end{itemize}
						We prove now that $\left\Vert \mathcal{M}^{\geq v}\right\Vert =\left\{ w\in
			A^{\omega }\mid \left( \left\Vert \mathcal{N}\right\Vert ,w\right) \geq
			v\right\} .$ We let first $w=a_{0}a_{1}a_{2}\ldots \in \left\{ w\in
			A^{\omega }\mid \left( \left\Vert \mathcal{N}\right\Vert ,w\right) \geq
			v\right\} .$ This implies that%
			\begin{equation*}
				\left( \left\Vert \mathcal{N}\right\Vert ,w\right) =\underset{P_{w}\in succ_{%
						\mathcal{N}}\left( w\right) }{\sup }\left( weight_{\mathcal{N}}\left(
				P_{w}\right) \right) =\infty
			\end{equation*}%
			or%
			\begin{equation*}
				\left( \left\Vert \mathcal{N}\right\Vert ,w\right) =\underset{P_{w}\in succ_{%
						\mathcal{N}}\left( w\right) }{\sup }\left( weight_{\mathcal{N}}\left(
				P_{w}\right) \right) =v^{\prime \prime }\in
				\mathbb{Q}
				\text{ with }v^{\prime \prime }\geq v.
			\end{equation*}%
			If the first case is true there exist a successful path $P_{w}=\left(
			q_{i},a_{i},q_{i+1}\right) _{i\geq 0}$ of $\mathcal{N}$ over $w$ such that $%
			wt\left( q_{i},a_{i},q_{i+1}\right) =\infty $ for all $i\geq 0.$ We consider
			the path%
			\begin{equation*}
				P_{w}^{\prime \prime }=\left( \left( q_{0},0,C\right) ,a_{0},\left(
				s_{q_{1}},2,D_{1}\right) \right) \left( \left( s_{q_{i}},2,D_{i}\right)
				,a_{i},\left( s_{q_{i+1}},2,D_{i+1}\right) \right) _{i\geq 1}
			\end{equation*}%
			where $D_{i}=B$ if $q_{i}\in F,$ and $D_{i}=C$ if $q_{i}\notin F$ \ for
			every $i\geq 1.$ (We shall call such a path $P_{w}^{\prime \prime }$ a path
			of form A). Clearly, $P_{w}^{\prime \prime }$ is a successful path of $%
			\mathcal{M}^{\geq v}$ over $w,$ i.e., $w\in \left\Vert \mathcal{M}^{\geq
				v}\right\Vert $ as desired. Now, if the second case is true there exist a
			successful path $P_{w}= \left( q_{i},a_{i},q_{i+1}\right)
			_{i\geq 0}$ of $\mathcal{N}$ over $w$ such that $wt\left(
			q_{i},a_{i},q_{i+1}\right) \neq \infty $ for at least one $i\geq 0.$ Assume
			that there exist infinitely many $i\geq 0$ such that $wt\left(
			q_{i},a_{i},q_{i+1}\right) \neq \infty .$ Then, there are also infinitely
			many $i\geq 0$ with $wt\left( q_{i},a_{i},q_{i+1}\right) =v^{\prime \prime
			}\geq v,$ and only a finite number of $i\geq 0$ with $wt\left(
			q_{i},a_{i},q_{i+1}\right) <v.$ Thus, we can define a
			successful path $P_{w}^{\prime \prime }$ of $\mathcal{M}^{\geq v}$ over $w$
			of the form $\left( \left( q_{i},j_{i},D_{i}\right) ,a_{i},\left(
			q_{i+1},j_{i+1},D_{i+1}\right) \right) _{i\geq 0}$ with $\left(
			q_{0},j_{0},D_{0}\right) =\left( q_{0},0,C\right) ,$\ and for all $i\geq 1,$
			$j_{i}=1,$ if $j_{i-1}=0,1,3,4$, and $wt\left( q_{i-1},a_{i-1},q_{i}\right)
			<v,$ or $j_{i}=3$ if $j_{i-1}=0,1,3,4$, and $wt\left(
			q_{i-1},a_{i-1},q_{i}\right) \geq v,$ or $j_{i}=4$ if $j_{i-1}=0,1,3,4$, and
			$wt\left( q_{i-1},a_{i-1},q_{i}\right) =\infty ,\ $and $D_{i}=C$ if $%
			q_{i}\notin F,$ and $D_{i}=B$ if $q_{i}\in F.$ Next, we assume that
			there exist only a finite number of $i\geq 0$ with $wt\left(
			q_{i},a_{i},q_{i+1}\right) \neq \infty .$ Then, for all $h\geq 0$ with $%
			wt\left( q_{h},a_{h},q_{h+1}\right) \neq \infty ,$ it holds $wt\left(
			q_{h},a_{h},q_{h+1}\right) =v^{\prime \prime }\geq v.$ Then, we can define a
			successful path $P_{w}^{\prime \prime }$ of $\mathcal{M}^{\geq v}$ over $w$
			as follows%
			\begin{equation*}
				P_{w}^{\prime \prime }=\left( \left( q_{0},0,C\right) ,a_{0},\left(
				s_{q_{1}},5,D_{1}\right) \right) \left( \left( q_{i},5,D_{i}\right)
				,a_{i},\left( q_{i+1},5,D_{i+1}\right) \right) _{i\geq 1}
			\end{equation*}%
			where $D_{i}=C$ if $q_{i}\notin F,$ and $D_{i}=B$ if $q_{i}\in F$ for all $%
			i\geq 1$ (We shall call such a path $P_{w}^{\prime \prime }$, a path of form
			B). Hence, $w\in \left\Vert \mathcal{M}^{\geq v}\right\Vert $.
			
			We proved that $\left\{ w\in A^{\omega }\mid \left( \left\Vert \mathcal{N}%
			\right\Vert ,w\right) \geq v\right\} \subseteq \left\Vert \mathcal{M}^{\geq
				v}\right\Vert .$ In order to prove the converse inclusion, we let $%
			P_{w}^{\prime \prime }=\left( \left( q_{i},j_{i},D_{i}\right) ,a_{i},\left(
			q_{i+1},j_{i+1},D_{i+1}\right) \right) _{i\geq 0}$ with $\left(
			q_{0},j_{0},D_{0}\right) =\left( q_{0},0,C\right) $ be a\ successful path of
			$\mathcal{M}^{\geq v}$ over $w$. This implies that $In^{Q}\left(
			P_{w}^{\prime \prime }\right) \cap \widehat{F}\neq \emptyset ,$ and $%
			In^{Q}\left( P_{w}^{\prime \prime }\right) \cap \widehat{S}=\emptyset $
			(case 1)$,$ or $In^{Q}\left( P_{w}^{\prime \prime }\right) \cap \overline{F}%
			\neq \emptyset $ (case 2)$,$ or $In^{Q}\left( P_{w}^{\prime \prime }\right)
			\cap \widetilde{F}\neq \emptyset $ (case 3)$.$
			
			We assume first that for $P_{w}^{\prime \prime }$ case (2), or case (3)
			holds, then $P_{w}^{\prime \prime }$ is either of form (A), or of form (B).
			If $P_{w}^{\prime \prime }$ is of form (A),\ and $P_{w}=\left(
			q_{i},a_{i},q_{i+1}\right) _{i\geq 0}\in succ_{\mathcal{N}}\left(
			P_{w}\right) $ with $weight_{\mathcal{N}}\left( P_{w}\right) =\infty \geq v.$
			If $P_{w}^{\prime \prime }$ is of form (B),\ then $P_{w}=\left(
			q_{i},a_{i},q_{i+1}\right) _{i\geq 0}\in succ_{\mathcal{N}}\left(
			P_{w}\right) $ with $weight_{\mathcal{N}}\left( P_{w}\right) \geq v
			.$ If case (1) holds, then $P^{\prime \prime }_{w}$ is also a successful path of
			$\mathcal{N}$ over $w$ with $weight_{\mathcal{N}}\left( P_{w}\right)
			=v^{\prime }\geq v,v^{\prime }\in
			\mathbb{Q}
			.$
			
			Thus, $\left\Vert \mathcal{M}^{\geq v}\right\Vert \subseteq \left\{ w\in
			A^{\omega }\mid \left( \left\Vert \mathcal{N}\right\Vert ,w\right) \geq
			v\right\} $, and this concludes our proof.
			
		\end{proof}

	\end{document}